\documentclass[11pt,draftclsnofoot,journal,letterpaper, onecolumn]{IEEEtran} 
\ifCLASSOPTIONcompsoc
  \usepackage[nocompress]{cite}
\else
  \usepackage{cite}
\fi
\usepackage{bm}
\usepackage{url}
\usepackage{bbm}
\usepackage{array}
\usepackage{xfrac}
\usepackage{amsthm}
\usepackage{widetext}
\usepackage{booktabs}
\usepackage{multirow}
\usepackage{makecell}
\usepackage{multicol}
\usepackage{hyperref}
\usepackage{threeparttable}
\usepackage{color,xcolor,colortbl}
\usepackage{amsmath,amssymb,amsfonts}
\usepackage{graphicx}
\usepackage{subfigure}
\usepackage{cases}
\usepackage{lscape,lipsum}
\usepackage{cuted}
\usepackage{tikz}
\usetikzlibrary{decorations.pathreplacing}
\usetikzlibrary{tikzmark,calc,matrix,positioning,arrows.meta}

\newtheorem*{game}{Game}
\newtheorem{fact}{Fact}
\newtheorem{lemma}{Lemma}

\newtheorem{belief}{Belief}
\newtheorem{theorem}{Theorem}
\newtheorem{corollary}{Corollary}
\newtheorem{assumption}{Assumption}

\newtheorem{observation}{Observation}
\newtheorem{proposition}{Proposition}
\newtheorem{keyquestion}{Question}

\usepackage{cite}
\usepackage{hyperref}
\hypersetup{
    colorlinks,
    linkcolor={red!70!black},
    citecolor={blue!70!black},
    urlcolor={blue!70!black}
}

\begin{document}

\title{Personalized Pricing Through Strategic User Profiling in Social Networks}

\author{Qinqi~Lin,
        Lingjie~Duan,~\IEEEmembership{Senior Member,~IEEE,} Jianwei~Huang,~\IEEEmembership{Fellow,~IEEE}
\IEEEcompsocitemizethanks{\IEEEcompsocthanksitem This paper is published in IEEE/ACM Transactions on Networking, October 2024 [\href{https://ieeexplore.ieee.org/document/10574176}{DOI: 10.1109/TNET.2024.3410976}]. Part of the results appeared in WiOpt 2021 Conference~\cite{lin2021personalized}.

Qinqi Lin is with the School of Science and Engineering, Shenzhen Institute of Artificial Intelligence and Robotics for Society, The Chinese University of Hong Kong, Shenzhen, Shenzhen 518172, China (e-mail: qinqilin@link.cuhk.edu.cn).

Lingjie Duan is with the Engineering Systems and Design Pillar, Singapore University of Technology and Design, Singapore 487372 (e-mail: lingjie\_duan@sutd.edu.sg).  

Jianwei Huang is with the School of Science and Engineering, Shenzhen Institute of Artificial Intelligence and Robotics for Society, The Chinese University of Hong Kong, Shenzhen, Shenzhen 518172, China (e-mail: jianweihuang@cuhk.edu.cn).  
}
}


\IEEEtitleabstractindextext{%
\vspace{-45pt}
\begin{abstract}
Traditional user profiling techniques rely on browsing history or purchase records to identify users' willingness to pay. This enables sellers to offer personalized prices to profiled users while charging only a uniform price to non-profiled users. However, the emergence of privacy-enhancing technologies has caused users to actively avoid on-site data tracking. Today, major online sellers have turned to public platforms such as online social networks to better track users' profiles from their product-related discussions. This paper presents the first analytical study on how users should best manage their social activities against potential personalized pricing, and how a seller should strategically adjust her pricing scheme to facilitate user profiling in social networks. We formulate a dynamic Bayesian game played between the seller and users under asymmetric information. The key challenge of analyzing this game comes from the double couplings between the seller and the users as well as among the users. Furthermore, the equilibrium analysis needs to ensure consistency between users' revealed information and the seller's belief under random user profiling. We address these challenges by alternately applying backward and forward induction, and successfully characterize the unique perfect Bayesian equilibrium (PBE) in closed form. Our analysis reveals that as the accuracy of profiling technology improves, the seller tends to raise the equilibrium uniform price to motivate users' increased social activities and facilitate user profiling. However, this results in most users being worse off after the informed consent policy is imposed to ensure users' awareness of data access and profiling practices by potential sellers. This finding suggests that recent regulatory evolution towards enhancing users' privacy awareness may have unintended consequences of reducing users' payoffs. Finally, we examine prevalent pricing practices where the seller breaks a pricing promise to personalize final offerings, and show that it only slightly improves the seller's average revenue while introducing higher variance.
\end{abstract}

}

\maketitle

\IEEEdisplaynontitleabstractindextext
\IEEEpeerreviewmaketitle


\section{Introduction}\label{sec:introduction}

\IEEEPARstart{M}{odern} advances in information technology have enabled user profiling, which is a data analytic tool that outlines users' characteristics like product preferences~and~demands. The availability of user profiling allows a seller~to~identify a user's willingness to pay and exercise personalized pricing accordingly. For example, major online stores like Amazon and Sears track users' web browsing history and online~purchase records to personalize their offerings~\cite{howe_a,datadriven}. Orbitz further differentiates users based on their computer's operating systems and offers higher hotel prices to~Mac~users~\cite{mattioli_2012}.~However, with the emergence of privacy-enhancing technologies (e.g., cookie blockers) that make user information no longer personally identifiable, today's product sellers have encountered increasing obstacles when tracking users' profiles only based on their on-site store data \cite{anant_2020_consumer}.

With the ever-increasing penetration of online social networks and the proliferation of user-generated data, more~\mbox{sellers} have turned to tracking users' profiles from their social network activities. While users register personal accounts for primary features on social media, sellers can easily associate a specific user's online activities (e.g., likes, posts, and comments) with his identity. This allows product sellers to extract user-specific information by analyzing one's social network activities, especially those discussions or sharing~related to the product. For instance, users sharing fitness photos frequently tend to value gym services more than average users.

Recent years have witnessed a rapid development of technologies that enable such user profiling with social network data \cite{2019Social}. For example, Facebook has launched APIs~that enable several online stores' access to users' Facebook identity and profile information since 2008~\cite{announcing}. Amazon and EyeBuyDirect have bought data from Facebook to train their personalization features \cite{bustos_2011_facebook,bustos_2011_7}. Apparently, such tracking and analysis of social network data can facilitate user profiling for product sellers to intensify price discrimination.

As various governments have recently tightened privacy regulations with a marked shift toward protecting users (e.g., \cite{a2013_general,a2018_california}), users are increasingly aware of such social data access and profiling practices from potential sellers. Particularly, the well-known informed consent policy requires online stores to inform people of their data access and the purpose of data usage. Following this, Amazon now needs to notify users of its social profile access in great detail when users sign in and connect to their Facebook accounts~\cite{bustos_2011_7,boulton_2010_amazon}. Given such awareness, many users have reduced social media usage or even remain silent online to thwart the seller's~attempts~to~profile and exercise personalized pricing \cite{cisco_2019_cisco,daxtheduck_2019_new}. We are thus motivated to ask the following two key questions.

\begin{keyquestion}\label{Q1}
How should users best manage their online social activities when facing the possibility of user profiling by the seller?
\end{keyquestion}

\begin{keyquestion}\label{Q2}
How should the seller react to users' reduced social activities and strategically adjust her pricing schemes to facilitate user profiling?
\end{keyquestion}

Addressing these questions presents several key challenges. Firstly, there exists a coupling between the seller and the~users, as well as a coupling among the users themselves. These~two types of couplings also interact with each other, which has not been systematically studied before. Regarding the first type of coupling, users are aware of the profiling risks when interacting with each other online and may adjust their interactions accordingly. Hence, the seller needs to carefully balance the uniform pricing for non-profiled users and the personalized pricing for profiled users to maximize profit. Previous works on~traditional user profiling (e.g., \cite{acquisti2005conditioning,conitzer2012hide,2015Monopoly}) viewed users as isolated individuals and analyzed the effect of each individual's endeavor~to~conceal private information and bypass the seller's profiling independently. However, these works ignored the coupling among users and, more specifically, the positive network externality of users' information revelation. Thus, the insights gained from previous studies may not be applicable in this context, and this externality will significantly complicate the analysis of the implications of user profiling.

Another major challenge is the uncertainty involved in user profiling. Even if a user has exposed his personal information online, whether the seller can successfully profile this user is still probabilistic due to technological or regulatory constraints (e.g., \cite{a2013_general,a2018_california}). On the technical level, updating the seller's beliefs about users' private information based on user profiling results is subject to inherent uncertainty, as the profiling results may not directly relate to users' online social behaviors. Therefore, this randomness in user profiling makes it more difficult to ensure consistency between users' revealed information and the seller's belief for further pricing. Note that satisfying such belief consistency is fundamental in analyzing interactions in a dynamic Bayesian setting \cite{fudenberg1991perfect}.

We summarize the main contributions of this work below.

\begin{itemize}

\item \emph{Novel personalized pricing through user profiling in~social networks:} To the best of our knowledge, this~is~the first analytical study regarding how users should best manage their social activities against potential personalized pricing, and how a seller should strategically~adjust her pricing schemes to facilitate user profiling in social networks. Our work provides valuable insights for current regulatory efforts regarding users’ awareness of potential profiling practices in social networks and calls for further review of the current informed consent policy.

\item \emph{Doubly-coupled game formulation \& analysis:} We model the unique doubly-coupled interactions between the seller and the users under asymmetric information as a dynamic Bayesian game. Our game analysis is challenging due to the double couplings between the seller and the users as well as amongst the users themselves. Moreover, we must ensure consistency between users' revealed information and the seller's belief under random user profiling, which prevents us from relying solely on the traditional method of backward induction. Instead, we alternate backward induction with forward induction and successfully characterize the unique perfect Bayesian equilibrium (PBE) in closed form.

\item \emph{Regulatory implications of user profiling:} As the accuracy of profiling technology improves, the seller will raise the equilibrium uniform price to encourage users' increased social activities and facilitate user profiling. However, this results in most users being worse off after the informed consent policy is implemented to ensure users' awareness of data access and profiling practices by potential sellers. This finding then suggests that recent regulatory evolution towards enhancing users' privacy awareness may have~unintended consequences of reducing users' payoffs.

\item \emph{Tradeoff involving breaking a pricing promise:} We further examine some prevalent pricing practices in our setting, where the seller breaks a pricing promise to personalize final offerings \cite{thecouncilofeconomicadvisers_2015_the,chen2020competitive}. Our analysis demonstrates that while this practice allows for the seller's proactive control of users' incentives in social decisions, it only marginally improves the seller's average revenue while introducing higher variance. Considering the additional reputation and regulation risks \cite{lee_2021_chinese}, it is not recommended for sellers to break pricing promises in practice.

\end{itemize}

The remainder of the paper is organized as follows. Section \ref{sec:relatedwork} reviews the related work, and we present the system model in Section \ref{Sec:Model}. To characterize the PBE, we first alternately~apply backward induction and forward induction in Section~\ref{Sec:alternate}. Then the analysis of PBE and its implications follows in~Section \ref{Sec:PBE}. We further investigate the effect of~social network benefits~in Section \ref{Sec:Network} and explore the seller's flexibility in pricing timing in Section~\ref{Sec:Sequential}. Furthermore, we extend our analysis to incorporate the heterogeneity in users' social network benefits and network positions in Section~\ref{extension:heterogeneous}. Section~\ref{appendix:general} relaxes our prior assumption of a uniform distribution over user valuations, to demonstrate the robustness of our major insights. Section~\ref{sec:conclusion} concludes this paper.

\section{Related Work}\label{sec:relatedwork}

A growing literature on personalized pricing has studied users' purchase behaviors in online markets while considering the implications of revealing their private information~to~product sellers (e.g., \cite{conitzer2012hide,2015Monopoly,2017Is,2020Consumer,chen2020competitive,bimpikis2021data}). For example, Conitzer et al. in \cite{conitzer2012hide} investigated a repeated purchases~scenario, where users could hide their past purchase records to hinder personalized pricing. Koh et al. in \cite{2017Is} further allowed users' voluntary participation in profiling, while considering the benefit of reducing search costs for ideal products through users' information revelation. Valletti et al. in \cite{2020Consumer} analyzed how users reveal or conceal their private information, given the seller's strategic investment in consumer profiling technology. Chen et al. in~\cite{chen2020competitive} explored a duopoly setting, where sellers compete to target users and attempt to price discriminate while users make strategic purchase decisions to bypass.

The existing studies did not consider the seller's strategic pricing strategy to profile users and the corresponding impact on users' behaviors in social networks. As a result, the linkage between user profiling in social networks and personalized pricing has been largely missing in current studies. This paper faces the unique challenge of analyzing the doubly coupled interactions among users in social networks and between the seller and the coupled users under random user profiling. This leads to a rather involved PBE analysis for such dynamic Bayesian interactions.

Several studies in the literature on personalized pricing have adopted PBE as the solution concept (e.g., \cite{2017Is,2020Consumer,2020Voluntary,conitzer2012hide}). The standard analysis for PBE proceeds in three steps as follows. First, a belief about the structure of certain players' strategies is proposed. Based on this proposed belief, the analysis for the remaining players' strategies is then conducted through backward induction. Finally, the equilibrium strategies and the proposed belief are verified to constitute a PBE satisfying belief consistency and sequential rationality. However, the coupling among users' social behaviors in our work significantly complicates the belief proposal in the first step. This departs from previous studies (e.g., \mbox{\cite{conitzer2012hide,2015Monopoly,2017Is}}) that analyze users' online behaviors individually, which allows for an easily identifiable threshold structure for the equilibrium belief. Motivated by the advanced game theory literature \cite{battigalli2006rationalization}, we propose to explicitly conduct a forward analysis of users' coupled social interactions to derive the structure of the equilibrium belief. In a nutshell, we alternate backward induction with forward induction to analyze the whole dynamic Bayesian game and provide a complete PBE characterization.

\section{System Model}\label{Sec:Model}

Consider the coupled interactions among the seller and a group of users $\mathcal{N}\triangleq\{1,\dots,n\}$ under user profiling. Specifically, users share information and interact with each other in the social network, while the shared social activity data is monitored or even accessible by the seller (e.g., through data-sharing deals with the social network platform \cite{x_2018}). Then, the seller manages to extract users' private profiles from their online social activities to enable personalized pricing.\footnote{For convenience, we hereafter use female pronouns to refer to the seller and male (or plural) pronouns to refer to each user (or users).} 
 
\begin{figure}[h]
	\centering
	\includegraphics[width=0.8\linewidth]{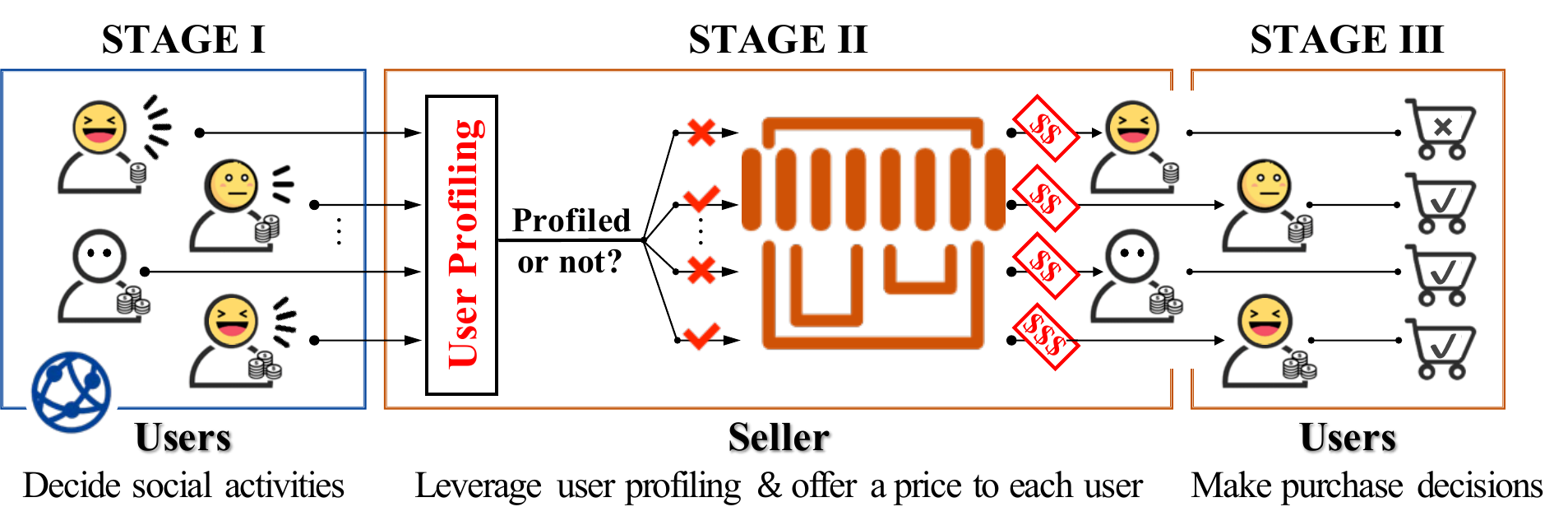}
	\caption{A Three-stage System Model.}
	\label{system_model}
\end{figure}
 
As illustrated in Fig. \ref{system_model}, we formulate such interactions as a three-stage dynamic Bayesian game below:

\begin{itemize}
\item \textbf{\mbox{Stage I:}} Initially, each user $i\in\mathcal{N}$ decides his social activity level $x_i$ in the normalized range of $[0,1]$. Each user would trade off the social network benefit (from interactions with the other users) against the potential risk of revealing his private information to the seller. Here, $x_i$ indicates the time or attention user $i$ devotes to the social network. In particular, the normalized maximum level $x_i=1$ means that the user $i$ spends all his available time using the social network, while the minimum level $x_i=0$ implies that the user is inactive on social media.
    
\item \textbf{\mbox{Stage II:}} 
Next, the seller first applies the profiling technology to users' social activity data, aiming to identify their product valuations (i.e., willingness to pay).\footnote{We consider that the seller cannot directly observe the social activity levels of users. Instead, the seller only observes the user profiling results. This well models the real practices, where the seller can only leverage the profiling technology to analyze users' digital traces in social networks without access to the raw data due to regulatory constraints.} Once receiving an accurate signal about a particular user $i$'s valuation $v_i$, the seller then tailors the price for that user accordingly. That is, the seller offers a personalized price exactly the user's valuation, i.e., $p_i=v_i$. For those users whom the seller fails to profile (due to users' inactive social behaviors or technological constraints), the seller offers one common uniform price $p_0$ to them. Overall, the price offered to each user $i$ depends on the seller's profiling result, i.e.,
\begin{equation}\label{offered_price}
p_i =\left\{
\begin{aligned}
v_i, & \quad\textrm{if profiled by the seller,}\\
p_0, & \quad\textrm{otherwise.}
\end{aligned}
\right. 
\end{equation}
    
\item \textbf{\mbox{Stage III:}} Finally, each user $i\in\mathcal{N}$ makes his binary purchase decision $d_i\in\{0,1\}$. By comparing one's product valuation $v_i$ with the offered price $p_i$, each user decides whether to purchase the product or not.
\end{itemize}

In what follows, we introduce the models under user profiling for users and the seller in Sections \ref{subsec:usermodel} and \ref{subsec:sellermodel}, respectively. Finally, we formally formulate our dynamic Bayesian game in Section \ref{subsection:game}.  

\subsection{Users' Model under User Profiling}\label{subsec:usermodel}

In this subsection, we first model users' social interactions in Stage I and product purchase in Stage III (see Fig.\ref{system_model}). We then formulate the final payoff for each user at the end of this subsection.

\subsubsection{Social Interactions in Stage I}\label{subsubsection:socialinteraction}

Users enjoy interacting with each other in the social network while leaving digital traces or data online. With awareness of the seller's profiling and personalized pricing, each user $i\in\mathcal{N}$ would strategically decide his social activity level $x_i\in[0,1]$ (e.g., through sharing posts and comments) in Stage I. Let $\boldsymbol{x}_{-i}$ summarize the social activity levels of all other users except user $i$.

Each user $i$ gains more satisfaction in the social network if the social activity levels (of his and other users') are higher. For instance, users feel connected, gain empathy and identity, and strengthen communication through more social interactions online \cite{seiter2015secret}. Formally, we model the social network benefit $J_i$ for each user $i$  as below,
\begin{equation}\label{social_network_benefit}
	J_i(x_i,\boldsymbol{x}_{-i})= x_i\ln\left(\sum_{j\neq i}x_j+\omega_0\right),
\end{equation}
which captures the network effect. Namely, the other users' social activities impose a positive externality on one user's social network benefit. The choice of the logarithmic function is motivated by the Zipf's law \cite{2006Metcalfe}, which emphasizes that the marginal network benefit each user $i$ experiences is diminishing as the others intensify their social activities. 

In particular, a user $i$ receives no social network benefit in (\ref{social_network_benefit}) if he is totally inactive online with $x_i=0$, no matter how socially active the other users are. Yet if $x_i$ is positive, even when all the other users are inactive (i.e., $x_j=0, \forall j\neq i$), user $i$ still has positive intrinsic benefit in (\ref{social_network_benefit}) with $\omega_0>1$. This arises from the usage of some basic services on social media, without the need to interact with others.

\subsubsection{Product Purchase in Stage III}

The seller sells products or services to the users, each with one unit demand. Denote each user $i$'s valuation for the product or service as~$v_i$. Here, we assume that each user's valuation is independent and identically distributed (i.i.d.) over $[0,\bar{v}]$ according to a uniform distribution. This uniform assumption enables tractable analyses to provide sharper insights. We extend to more general distributions to demonstrate the robustness of our major results in Section~\ref{appendix:general}, where our analysis method is still applicable. 

In Stage III, each user $i$ decides whether to purchase the product or not based on the offered price $p_i$ in (\ref{offered_price}), i.e.,
\begin{equation}\label{user: purchase decision}
d_i^*(v_i,p_i) = \mathbbm{1}(v_i\ge p_i),
\end{equation}
where $\mathbbm{1}(\cdot)$ is an indicator function. Then, each user $i$ receives a purchase surplus of $\max\{v_i-p_i,0\}$ in Stage III.

\subsubsection{Final User Payoff}
Finally, we formulate the final payoff $\pi_i$ for each user $i$ as below, which consists of the social network benefit in Stage I and the purchase surplus in Stage~III, i.e.,
\begin{equation}\label{final_user_payoff}
\pi_i(x_i,\boldsymbol{x}_{-i})=J_i(x_i,\boldsymbol{x}_{-i})+\max\{v_i-p_i,0\},
\end{equation}
which in fact depends on the user profiling results, i.e.,
\begin{equation}\label{user_utility}
\pi_i(x_i,\boldsymbol{x}_{-i})=\left\{
\begin{aligned}
&J_i(x_i,\boldsymbol{x}_{-i}), & \textrm{if profiled,}\\
&J_i(x_i,\boldsymbol{x}_{-i})+\max\{v_i-p_0,0\},  & \textrm{otherwise.}
\end{aligned}
\right. 
\end{equation}
Especially, if user $i$ is successfully profiled by the seller, user $i$ only has the social network benefit $J_i(x_i,\boldsymbol{x}_{-i})$ of (\ref{social_network_benefit}) in the final. This is because the seller can exercise personalized pricing and leave zero purchase surplus to the user. Otherwise, the user may have a positive purchase surplus with the uniform price $p_0$ (if affordable).

\subsection{Seller's Model with User Profiling}\label{subsec:sellermodel}

The seller employs the profiling technology (e.g., \cite{2013Private,2018Collective}) to analyze users' social activity data and statistically infer their valuations (if possible). The user profiling signals the seller about one user's valuation for the product or service. Based on the profiling results, the seller determines the prices (either personalized or uniform) offered to different users.

\subsubsection{User Profiling} 
Similar to \cite{2015Monopoly,2017Is,2020Consumer}, we consider a typical binary profiler that outputs a signal about users' valuation if and only if it is sure about the result. We assume that the profiling result about user $i$ follows a Bernoulli distribution. That is, user $i$'s valuation $v_i$ is revealed with a probability $\lambda_i\in[0,1]$; otherwise, user $i$ is non-profiled with unrevealed valuation. Here, the signal accuracy $\lambda_i$ is increasing in user $i$'s social activity level $x_i$, i.e.,
 \begin{equation}\label{profiling_accuracy}
 	\lambda_i = \delta x_i^\alpha,
 \end{equation}
where $0<\alpha<1$ and $0\le\delta\le 1$. 

The signal accuracy in (\ref{profiling_accuracy}) is a concave function of $x_i$ with $0<\alpha<1$. It indicates that, as the seller already harvests lots of critical information, the extra data extracted from user $i$'s online social activity becomes less informative. Note that (\ref{profiling_accuracy}) is upper bounded by $\delta$, referred to as user profiling accuracy. This practically captures the fact that, due to technological and regulatory constraints (e.g., \cite{a2013_general,a2018_california}), whether the seller can successfully profile a user is probabilistic even when this user is fully active in the social network.

One interpretation of such profiler assumption is to consider a binary hypothesis testing of user $i$'s valuation at value $v_i$ with a confidence level $\lambda_i$ \cite{2020Consumer}. Our analysis method also applies to more general profiling models. For instance, we can extend our analysis to the case with noisy profiling signals, where a user $i$'s valuation $v_i$ is revealed as a random variable over the support $[a_i,b_i]$ following a cumulative distribution function $G(\cdot)$. Instead of choosing the price $p_i^*=v_i$ in (\ref{offered_price}), the seller in this more general case can sets the optimal personalized price as $p_i^*=\arg\max_{p_i\in[a_i,b_i]} p_i(1-G(p_i))$. For clarity, we focus on the binary profiler to crystalize insights in this work.

\subsubsection{Seller's Revenue}
In Stage III, the seller gains revenue if users purchase the product at the pricing vector $(p_1,p_2,\dots,p_n)$ offered to them. Denote the sets of users who are non-profiled and profiled as $\mathcal{N}_{0}$ and $\mathcal{N}_{1}$, respectively. As the seller charges $p_0$ for the non-profiled users and $p_i=v_i$ for the profiled users $i\in\mathcal{N}_{1}$, we model the seller's revenue as follows,

\begin{equation}\label{platform_sale_revenue}
\Pi(p_1,p_2,\dots,p_n)=\sum_{i\in \mathcal{N}_{0}}p_0d_i+\sum_{i\in \mathcal{N}_{1}}v_i.
\end{equation}

\subsection{Dynamic Bayesian Game Formulation}\label{subsection:game}

We formally model the interactions between the seller and users as a three-stage dynamic Bayesian game, as illustrated in Fig. \ref{system_model}. Besides the social coupling among users in Stage~I (see Section \ref{subsubsection:socialinteraction}), another coupling exists between users and the seller across Stages I and II. Specifically, the seller's profiling-enabled pricing in Stage II depends on users' social activities in Stage I. Meanwhile, users in Stage I make social decisions while predicting the seller's pricing in Stage~II to infer the profiling risk. In particular, a higher uniform price would reduce users' incentives to avoid personalized pricing. This then stimulates users' increased social activities towards better profiling results for the seller but may discourage some non-profiled users from purchasing the product.

Next, we explain the information structure for our model. In Stage I, each user's valuation is initially private and only known to himself. The seller only knows the prior distribution, and each user only knows the common distribution of other users' valuations (rather than the precise values). This setting aligns with today's social media platforms, where users usually interact with a large-scale social network and~have~limited~information about other users' specific valuations. Furthermore, users are uninformed about each other's choices before determining their own social activity levels in Stage~I. Through~user profiling, the seller in Stage II obtains extra~information~regarding the valuations of profiled users. Yet, she still~does~not know those non-profiled users' private valuations.~Overall,~the seller's incomplete information about each user's product~valuation decreases along the stages in Fig. \ref{system_model}.

A standard approach to dynamic game analysis is backward induction. However, if we use backward induction alone, we only manage to derive the following two decisions.

\begin{itemize}
\item \textbf{\mbox{Users' Purchase Decisions in Stage III:}} Given the seller's offered price $p_i$, each user $i$'s optimal purchase decision is characterized as in (\ref{user: purchase decision}).
    
\item \textbf{\mbox{Seller's Personalized Pricing Scheme in Stage II:}} Given knowledge of any profiled user $i$'s valuation $v_i$, the seller's optimal personalized pricing scheme is to offer a tailored price exactly his valuation, i.e., $p_i=v_i$.
\end{itemize}

Due to the double-coupled decision-making of users' social activities and the seller's (uniform) pricing in Stages I and~II, backward induction alone cannot solve the equilibrium. To ensure belief consistency over stages, we further combine forward induction motivated by the advanced game theory literature \cite{battigalli2006rationalization}. That is, we alternate backward induction with forward induction to analyze the PBE of our dynamic Bayesian game (see Section \ref{Sec:alternate} for details).

\section{Alternating Backward Analysis with Forward Analysis in Stages I \& II}\label{Sec:alternate}

After characterizing the seller's personalized pricing scheme in Stage II and users' purchase decisions in Stage III (See~Section \ref{subsection:game}), we now continue to analyze the PBE for Stages I and II. 

Our PBE analysis develops as follows. First, in Section~\ref{backward}, we use backward induction to analyze the seller's uniform pricing in Stage II, based on a belief of users' social activity structure. This belief is derived through forward induction considering the users' social decisions in Stage I, which we elaborate on in Section \ref{forward}. Finally, we combine the backward and forward analyses to derive the PBE in the next section.

\subsection{Backward Analysis of Seller's Uniform Pricing in Stage II}\label{backward}

This subsection explores the seller's uniform pricing scheme in Stage II. Specifically, we first propose a belief about users' social activity structure in Stage I, based on which we then analyze how the seller determines the uniform price in Stage~II.

Notice that we cannot directly analyze the seller's uniform pricing scheme through traditional backward induction. A standard backward analysis relies on users' social activities in Stage I. However, there exists a great diversity of possible outcomes of users' social activity decisions, which brings great complexity to the uniform pricing discussion through traditional backward induction. To enable a backward analysis of the seller's uniform pricing, we first provide the structural result of users' social activity decisions in Stage I, which is derived using forward induction to be explained in Section~\ref{forward}. 

\begin{belief}\label{threshold_structure}
\emph{(Belief of Social Activity Structure)} There exists a common valuation threshold $v^* \in [0, \bar{v}]$ for users' social activity decisions in Stage I, such that any user $i$ with a valuation $v_i> v^*$ chooses the minimal social activity level $x_i(v_i)=0$, whereas any user $i$ with a valuation $v_i\le v^*$ chooses the maximal social activity level $x_i(v_i)=1$. i.e.,
\begin{equation}
	x_i^*(v_i)=\mathbbm{1}(v_i\le v^*).
\end{equation}
\end{belief}

The above belief is formally proved in Section \ref{forward}. Note that all users share the same valuation threshold independent of their valuations. Intuitively, users with higher valuations prefer to be non-profiled, which would lead to a low uniform price $p_0$ rather than a high personalized price $v_i$ charged by the seller. Thus, high-valuation users with $v_i>v^*$ become inactive online to avoid revealing their private profiles (i.e., $x_i=0$). In contrast, if a user $i$'s valuation $v_i$ is less than $v^*$, the social network benefit outweighs his potential risk of being profiled. Therefore, this low-valuation user remains active in the social network (i.e., $x_i=1$). According to (\ref{profiling_accuracy}), this user's chance of being successfully profiled by the seller is $\delta$ then.

Hereafter, we thus restrict to consider an arbitrary value of the valuation threshold $v^*$ when conducting the backward analysis for the seller's uniform pricing in Stage II. After understanding this threshold structure of users' social activity decisions, the seller updates her belief of the profiled and non-profiled users' valuations in Stage II according to Bayes' rule.

Initially, in Stage I, the seller only knows the prior distribution of any user's valuation (i.e., uniform distribution over $[0,\bar{v}]$). After user profiling, the seller's posterior belief for a profiled user's valuation $f(v_i|i\in\mathcal{N}_{1})$ in Stage II is given by
\begin{equation}\label{profiled_belief}
\begin{aligned}
			f(v_i|i\in\mathcal{N}_{1})=\frac{1}{v^*},\quad\text{if}\ 0\le v_i\le v^*,
\end{aligned}
\end{equation}
whereas the posterior probability density function for a non-profiled user's valuation $f(v_i|i\in\mathcal{N}_{0})$ is as below:
\begin{equation}\label{nonprofiled_belief}
f(v_i|i\in\mathcal{N}_{0})=\left\{
	\begin{aligned}
		\frac{1-\delta}{\bar{v}-\delta v^*}, &\quad\textrm{if $0\le v_i\le v^*$,}\\
		\frac{1}{\bar{v}-\delta v^*},  &\quad\textrm{if $v^*< v_i\le \bar{v}$.}
	\end{aligned}
	\right.
\end{equation}
Here, we have $\delta v^*<\bar{v}$ with $v^*\leq \bar{v}$ and $\delta< 1$. Noticed that the non-profiled users consist of two types: inactive users and active users whom the seller fails to profile.

Before we further analyze the seller's uniform pricing in Stage II under the general case of random user profiling with $\delta\in(0,1)$, we first introduce a benchmark: \textit{no profiling} with $\delta=0$. In this case, the seller cannot exercise personalized pricing, and users thus decide social decisions only to maximize their social network benefits.

\begin{lemma}\label{Lem:noprofiling}
In the \textbf{no profiling} case with $\delta=0$, there exists a unique equilibrium characterized as follows. In Stage I, users' equilibrium valuation threshold is $v^*=\bar{v}$, i.e., all users are active in the social network with $x_i(v_i)=1,\forall i\in\mathcal{N}$. In Stage II, the seller only considers the uniform price, setting it as
\begin{equation}\label{noprofiling_price}
    p_0^*=\frac{\bar{v}}{2}.
\end{equation} 
\end{lemma}

Next, we characterize the seller's optimal uniform pricing scheme under random user profiling with $\delta\in(0,1)$, and compare it with the no profiling benchmark.

\begin{proposition}\label{base_price}
Given any users' valuation threshold $v^*$ in Stage I, the seller's optimal uniform pricing scheme (see Fig. \ref{fig:uniform}) in Stage II is given by 
\begin{numcases}{p_0^*(v^*)=}
\frac{\bar{v}}{2},  &\quad\textrm{if $v^*\le \frac{\bar{v}}{2}$,}\label{base_price_equation1}\\
v^*, &\quad\textrm{if $\frac{\bar{v}}{2}<v^*\le\frac{\bar{v}}{2-\delta}$,}\label{base_price_equation2}\\
\frac{\bar{v}-\delta v^*}{2(1-\delta)}, &\quad\textrm{if $\frac{\bar{v}}{2-\delta}<v^*\le \bar{v}$.}\label{base_price_equation3}
\end{numcases}
\end{proposition}

\begin{figure}[h]
	\centering  
	\begin{tikzpicture}[scale = 1, domain=0:4]
		\node (O) at (0,0) [below left] {\small $O$};
		\node at (2,0) [below] {\small $\frac{\bar{v}}{2}$};
		\draw[color=black,thick] (2,0) -- (2,0.07);
		\node (b) at (8/3,0) [below] {\small $\frac{\bar{v}}{2-\delta}$};
		\draw[color=black,thick] (8/3,0) -- (8/3,0.07);
		\node (c) at (4,0) [below] {\small $\bar{v}$};
		\draw[color=black,thick] (4,0) -- (4,0.07);
		\node (A) at (0,0.5) [left] {\small $\frac{\bar{v}}{2}$};
		\draw[color=black,thick] (0,2-1.5) -- (0.07,2-1.5);
		\node (B) at (0,8/3-1.5) [left] {\small $\frac{\bar{v}}{2-\delta}$};
		\draw[color=black,thick] (0,8/3-1.5) -- (0.07,8/3-1.5);
		\draw[-latex] (0,0) -- (4.5,0) node[right] {\small $v^*$};
		\draw[-latex] (0,0) -- (0,1.8) node[right] {\small $p_0^*$};
		\draw[color=red,thick] plot[domain=0:2,samples=200] (\x,2-1.5);
		\draw[color=red,thick] plot[domain=2:(8/3),samples=200] (\x,\x-1.5);
		\draw[color=red,thick] plot[domain=(8/3):4,samples=200] (\x,{(4-\x/2)-1.5});
		\draw[color=black,dotted] (8/3,0) -- (8/3,8/3-1.5);
		\draw[color=black,dotted] (0,8/3-1.5) -- (8/3,8/3-1.5);
		\draw[color=black,dotted] (2,0) -- (2,2-1.5);
		\draw[color=black,dotted] (2,2-1.5) -- (4,2-1.5);
		\draw[color=black,dotted] (4,0) -- (4,2-1.5);
		\fill[red] (2,0.5) circle (1pt);
		\fill[red] (8/3,8/3-1.5) circle (1pt);
		\fill[red] (4,0.5) circle (1pt);
	\end{tikzpicture}
	\caption{The seller's optimal uniform pricing scheme $p_0^*(v^*)$ in Stage II.}
	\label{fig:uniform}
\end{figure}
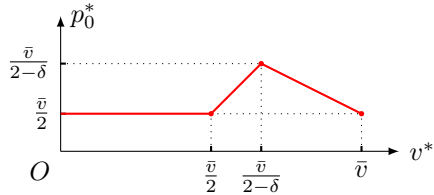

Proposition \ref{base_price} shows that if relatively few users prefer to reveal their profiles (i.e., $v^*\le\bar{v}/2$), the seller mainly cares about the revenue extracted from non-profiled users. Thus, the uniform price in (\ref{base_price_equation1}) degenerates to (\ref{noprofiling_price}) in Lemma \ref{Lem:noprofiling} (under the no profiling benchmark). As $v^*$ increases, more users with valuations $v_i\le v^*$ are active online and may get profiled by the seller. The seller's posterior belief of non-profiled users' valuations in (\ref{nonprofiled_belief}) thus skews towards higher values. Accordingly, the seller raises the uniform price in~(\ref{base_price_equation2}) higher than the mean user valuation, as shown in Fig. \ref{fig:uniform}. Differently, the seller reduces the uniform price in (\ref{base_price_equation3}) when $v^*$ approaches $\bar{v}$. This is because very few non-profiled users are with high valuations in this case. Hence, the seller tends to extract more revenue from those low-valuation users.

\subsection{Forward Analysis of Users' Social Decisions in Stage I}\label{forward}

In this subsection, we use forward induction to analyze the threshold structure of users' social activity decisions in Stage~I (see Belief \ref{threshold_structure}). Specifically, we investigate the equilibrium of users' social interaction game in Stage I, where users decide their social activity levels while predicting the seller's uniform price $p_0$ in Stage II.

Note that users in Stage I cannot fully control their chances of being profiled due to the randomness in profiling technology (see the signal accuracy in (\ref{profiling_accuracy})). In addition, users also need to estimate the seller's uniform pricing for the non-profiled case in Stage II, which in turn depends on their social activity levels in Stage I. Formally, each user $i$ with valuation $v_i$ in Stage I decides the social activity level $x_i$ to maximize his expected final payoff over all possible profiling results, i.e.,

\begin{equation}\label{expectpayoff}
    \tilde{\pi}_i(x_i,\boldsymbol{x}_{-i}) = x_i\ln\left(\sum_{j\neq i}x_j+\omega_0\right)+(1-\delta x_i^\alpha)\cdot\max\{v_i-p_0,0\},
\end{equation}
 which also hinges on the other users' social decisions $\boldsymbol{x}_{-i}$.
 
Recall Belief \ref{threshold_structure}, which describes the threshold structure of users' social activity decisions in Stage I. Its proof consists of the following three steps (see details in Appendix \ref{proofbelief}):

\begin{itemize}
\item \textbf{\mbox{Step I (Game Formulation and Solution Existence).}} We first identify users' social interactions in Stage I as a static Bayesian game with strategic complementarities \cite{1990Rationalizability}, where users make their social activity decisions while predicting the seller's uniform price $p_0$. This then indicates the existence of a pure-strategy Bayesian Nash equilibrium in users' social activity levels.
\item \textbf{Step II (Polarization in Equilibrium Social Decisions).} Next, we  show the convexity of users' expected utility in (\ref{expectpayoff}). It follows that the best response of each user $i$ is either the maximal social activity level $x_i^*=1$ or the minimal one $x_i^*=0$ instead of somewhere between (i.e., $x_i^*\in(0,1)$).
\item \textbf{Step III (Existence of Common Valuation Threshold).} Finally, we prove that there exists a valuation threshold for each user's equilibrium social decision. Further, we show that such a valuation threshold $v^*$ is the same for all users independent of their valuations.
\end{itemize}

Based on the threshold structure of users' equilibrium social decisions, we further characterize the valuation threshold in the following theorem. For clarity, we hereafter denote the cumulative distribution function (CDF) each user $i$'s valuation follows as $F(\cdot)$.

\begin{theorem}\label{Theorem:p0-vstar}
The common equilibrium valuation threshold $v^*$ in Belief \ref{threshold_structure} is $\min\{v^\dagger,\bar{v}\}$, where $v^\dagger$ satisfies
\begin{equation}\label{vstar_p0_summation}
v^\dagger=p_0+\frac{1}{\delta}\sum_{m=0}^{n-1}\tbinom{n-1}{m}\ln(m+\omega_0)F(v^\dagger)^m\left(1-F(v^\dagger)\right)^{n-1-m}.
\end{equation}
\end{theorem}

Theorem \ref{Theorem:p0-vstar} holds for any continuous distribution of users' valuations, not limited to the uniform distribution. To see why the valuation threshold is characterized as in (\ref{vstar_p0_summation}), we consider the two possible social activity levels in the equilibrium: $x_i=1$ (active) and $x_i=0$ (inactive). With~a valuation exactly the threshold $v^*$, the marginal user is indifferent between these two social decisions. That is, the purchase surplus he receives when inactive in the social network should be the same as the expected payoff he gains when active online, which consists of both the social network benefit and the purchase surplus.

Alternatively, we can rewrite the formulation in (\ref{vstar_p0_summation}) considering the expectation over the number of active users in the social network except for the user himself. Such a number, denoted by $k$, follows the binomial distribution $B(n-1,F(v^*))$. We then have 
\begin{equation}\label{vstar_p0_expectation}
	\begin{aligned}
		v^\ast&=p_0+\frac{1}{\delta}\mathbb{E}_{k}\ln(k+\omega_0),
	\end{aligned}
\end{equation}
where the last term on the right-hand side of (\ref{vstar_p0_expectation}) captures each user's expected social network benefit scaled by the profiling accuracy $\delta$. For convenience, we hereafter denote the expected social network benefit of each user in the equilibrium as $\tilde{J}(v^*)$, i.e.,
\begin{equation}\label{expected_social_network}
    \tilde{J}(v^*)\triangleq\mathbb{E}_{k}\ln(k+\omega_0).
\end{equation}

We conclude our forward analysis of users' social activity decisions with Proposition \ref{unique_user} below, which further shows the uniqueness of the equilibrium in Stage I.

\begin{proposition}\label{unique_user}
The Bayesian Nash equilibrium of users' social interaction game in Stage I is unique.
\end{proposition}

\section{Analysis of PBE}\label{Sec:PBE}

After alternating backward analysis with forward analysis in Stages I and II, we now derive the PBE of the whole dynamic Bayesian game in this section. Specifically, we combine users' equilibrium valuation threshold $v^\ast$ for their social activity decisions (see Theorem \ref{Theorem:p0-vstar}) and the seller's optimal uniform pricing scheme $p_0^*$ (see Proposition \ref{base_price}) together. In this way, we meet the two requirements of PBE: (i) sequential rationality: the seller sets the uniform price to maximize her sale revenue in Stage II given her belief of users' social activities, which is derived through Bayes' theorem; (ii) belief consistency: users make their social activity decisions in Stage I while predicting the seller's pricing in Stage II, and such equilibrium strategies should be consistent with the seller's belief \cite{fudenberg1991game}.

\begin{theorem}\label{Theorem_PBE}
Under \textbf{random user profiling} with $\delta\in(0,1)$, there exists a unique PBE characterized as follows.\footnote{We assume $\bar{v}> 2\ln(n-1+\omega_0)$ throughout the PBE analysis. Otherwise, only one trivial equilibrium exists under user profiling, where all users remain fully active (i.e., $x_i=1, \forall i\in\mathcal{N}$) and the seller sets the uniform price as $p_0^*=\bar{v}/2$. The reason for this is that, given $\bar{v}\le 2\ln(n-1+\omega_0)$, the social network benefit of any user is always sufficient to outweigh the potential loss in purchase surplus resulting from personalized pricing.\label{footnote3}}
\begin{itemize}
\item \emph{\textbf{Case I {(all active users in social networks)}}} with a small mean valuation
\begin{equation}
\frac{\bar{v}}{2}\le \frac{1}{\delta}\ln(n-1+\omega_0).
\end{equation}
\begin{itemize}
    \item [(i)] In Stage I, users' equilibrium valuation threshold is $v^*=\bar{v}$, i.e., every user $i\in\mathcal{N}$ is active in the social network with $x_i^*=1$.
    \item [(ii)] In Stage II, the seller sets a uniform price of $p_0^*=\bar{v}/2$.
\end{itemize}
		
\item \emph{\textbf{Case II {(partially active users in social networks)}}} with a large mean valuation
\begin{equation}
\frac{\bar{v}}{2}>{\frac{1}{\delta}\ln(n-1+\omega_0)}.
\end{equation}
\begin{itemize}
    \item [(i)] In Stage I, users' equilibrium valuation threshold is $v^*<\bar{v}$, i.e., there exist some high-valuation users who are inactive in the social network. Here, $v^\ast$ is the unique solution to 
    \begin{equation}\label{vstar_PBE}
        2\delta v^*-\delta^2 v^\ast-2(1-\delta)\tilde{J}(v^\ast)=\delta \bar{v}.
    \end{equation}
    Particularly, the fraction of socially active users $v^*/\bar{v}$ (i.e., users' overall social activity levels) decreases as $\bar{v}$ increases. 
    \item [(ii)] In Stage II, the seller sets a uniform price of 
    \begin{equation}
    p_0^*=v^\ast-\frac{\tilde{J}(v^\ast)}{\delta}>\frac{\bar{v}}{2}.
    \end{equation}
\end{itemize}
\end{itemize}
\end{theorem}

\begin{figure}[h]
  \centering
  \vspace{-15pt}
  \includegraphics[width=0.4\linewidth]{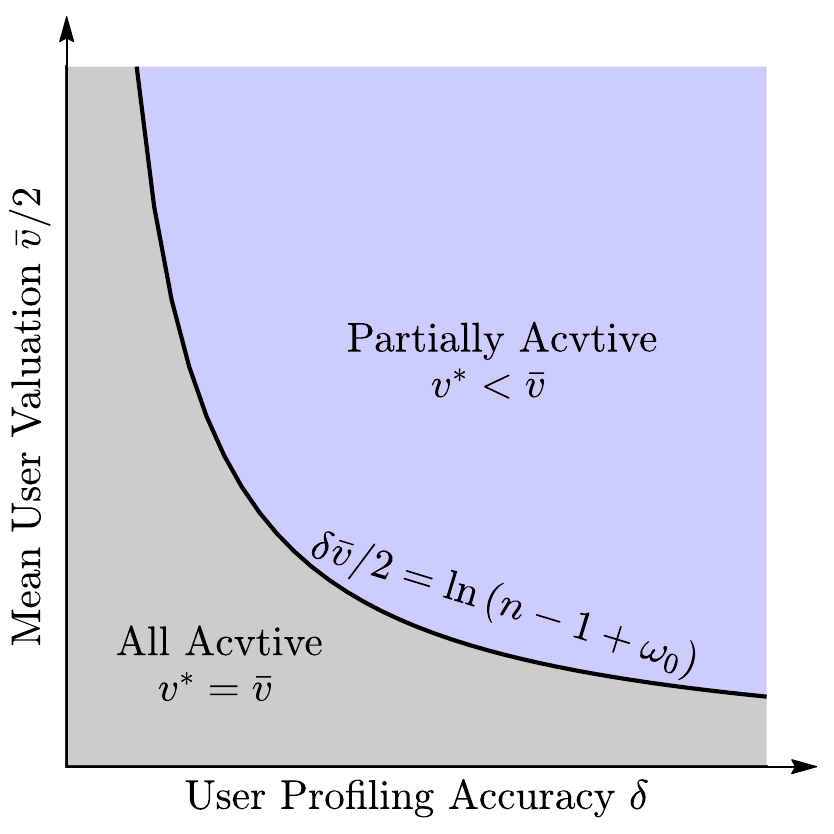}
  \vspace{-5pt}
  \caption{PBE under different values of mean user valuation $\bar{v}/2$ and user profiling accuracy $\delta$.}
  \label{system_parameter_impact}
  \vspace{-10pt}
\end{figure}

Fig. \ref{system_parameter_impact} depicts graphically the PBE characterized in Theorem \ref{Theorem_PBE}, which jointly considers the mean user valuation $\bar{v}/2$ and the user profiling accuracy $\delta$. In what follows, we explain the insights behind Theorem \ref{Theorem_PBE}.

\begin{itemize}
    \item In Case I \emph{(all active users)} with small mean valuation~$\bar{v}/2$, all users fully expose themselves in the social network. In this case, the purchase surplus is insignificant so that users do not worry too much about the loss from being profiled and charged personalized prices. In turn, this leads to the fact that the non-profiled users' valuations remain uniformly distributed. Hence, the seller simply sets the uniform price exactly the mean valuation to maximize her sale revenue from non-profiled users.
    
    \item In Case II \emph{(partially active users)} where the product valuation is large, the potential purchase loss from being profiled becomes significant. Thus, some high-valuation users would become inactive to avoid personalized pricing. In turn, the seller sets the uniform price higher than the mean valuation $\bar{v}/2$, aiming to extract more revenue from those non-profiled users with higher valuations.
\end{itemize}

In particular, we now shed light on one extreme case: \textit{perfect profiling} with $\delta=1$, where the seller can successfully profile any user as long as he is fully active in the social network.

\begin{proposition}\label{Proposition:perfectprofiling}
In the \textbf{perfect profiling} case with $\delta=1$, there exits a unique PBE characterized as follows. In Stage I, users' equilibrium valuation threshold is $v^*=\bar{v}$, i.e., all users are active in the social network with $x_i(v_i)=1,\forall i\in\mathcal{N}$. In Stage II, the seller sets the uniform price of
\begin{equation}
p_0^*=\bar{v}-\ln(n-1+\omega_0).
\end{equation}
\end{proposition}

Although both cases of no profiling (see Lemma \ref{Lem:noprofiling}) and perfect profiling here induce all users to be active online (i.e. $v^*=\bar{v}$), the seller's equilibrium uniform price $p_0^*$ is different. The difference lies in the role of the uniform price designed by the seller. When there is no profiling, the seller sets the uniform price purely to maximize sale revenue. As all users are non-profiled with $\delta=0$, the revenue can only be extracted through uniform pricing. In contrast, perfect profiling allows the seller to obtain more revenue through personalized pricing. Hence, the seller would set a uniform price to encourage more users to remain active online for better profiling. This accounts for the higher uniform price in the perfect profiling case (i.e., $p_0^*=\bar{v}-\ln(n-1+\omega_0)>\bar{v}/2$), which turns out to eliminate users' incentives to become inactive and bypass personalized pricing.

The juxtaposition of Theorem \ref{Theorem_PBE} with the two extreme cases (no profiling and perfect profiling) indicates that the users' equilibrium valuation threshold $v^*$ does not monotonically change in the user profiling accuracy $\delta$. Recall that under the perfect profiling case with $\delta=1$, the valuation threshold is $v^*=\bar{v}$, the same as in Case I of Theorem \ref{Theorem_PBE} with small $\delta$ as well as the no profiling benchmark with $\delta=0$. But the threshold turns out to be $v^*<\bar{v}$ for some $\delta\in(0,1)$ in between  (see Case II of Theorem~\ref{Theorem_PBE}). To explain such non-monotonicity in users' social activities, we further investigate the impact of user profiling accuracy $\delta$ in the following subsection.

\subsection{Impact of User Profiling Accuracy on PBE}\label{Subsection_PBEdelta}

We now analyze how the user profiling accuracy $\delta$ \mbox{affects} the equilibrium behaviors in this subsection. In particular, Proposition \ref{Prop:PBE_delta} (together with Fig. \ref{FIG:PBE_delta}) identifies the~specific~pattern of the non-monotonic changes in users' social activity~levels regarding~$\delta$. Hereafter, let $\hat{\delta}\triangleq2\ln\left(n-1+\omega_0\right)/\bar{v}$, and let $\tilde{\delta}$ denote the unique solution to $\partial v^*/\partial \delta=0$, where \eqref{vstar_PBE} defines $v^*$ as an implicit function of $\delta$.

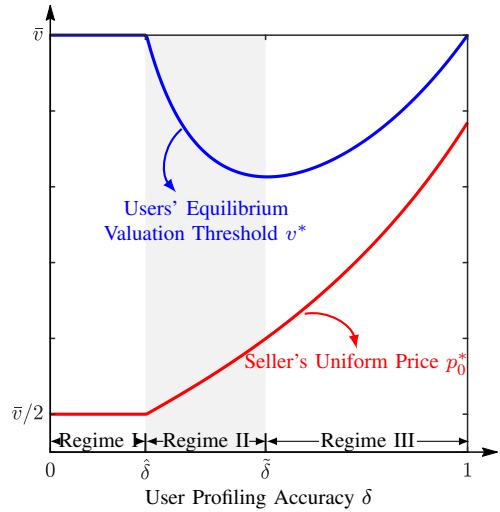
\begin{figure}[h]
	\centering
    \vspace{-10pt}
	\begin{tikzpicture}
		\node (image) at (0,0) {\includegraphics[width=0.4\linewidth]{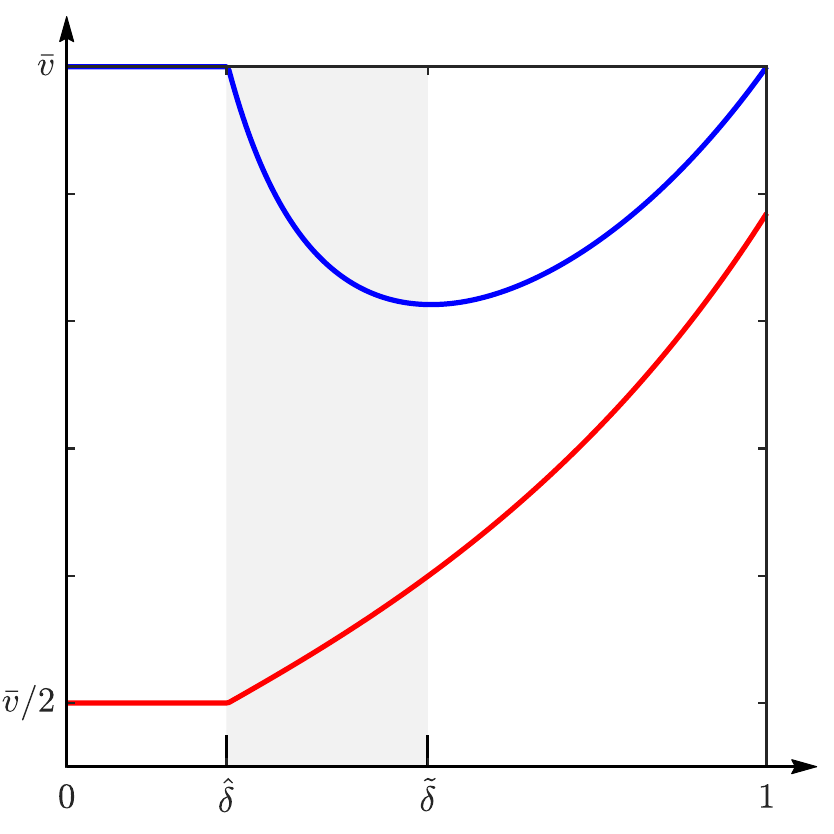}};
		\draw [latex-,red,thick] (1.3,-1.4) arc (12:100:0.6);
		\draw [latex-,blue,thick]  (-1.1,0.68)  arc (220:130:0.6);
		\node [blue,align=center]at (-0.7,0.45) {\scriptsize Users' Equilibrium};
        \node [blue,align=center]at (-0.7,0.15) {\scriptsize  Valuation Threshold $v^*$};
		\node [red] at (1.3,-1.6) {\scriptsize Seller's Uniform Price $p_{\tiny \text{0}}^*$};
		\node at (-2.13,-2.64) {\scriptsize  Regime I};
		\node at (-0.7,-2.64) {\scriptsize  Regime II};
		\node at (1.45,-2.64) {\scriptsize  Regime III};
		\node at (0,-3.4) {\scriptsize  User Profiling Accuracy $\delta$};
		\draw [-latex] (-2.68,-2.63) -- (-2.78,-2.63);
		\draw [-latex] (-1.6,-2.63) -- (-1.5,-2.63);
		\draw [-latex] (-1.27,-2.63) -- (-1.5,-2.63);
		\draw [-latex] (-0.13,-2.63) -- (0.1,-2.63);
		\draw [-latex] (0.8,-2.63) -- (0.1,-2.63);
		\draw [-latex] (2.1,-2.63) -- (2.8,-2.63);
	\end{tikzpicture}
    \vspace{-5pt}
	\caption{Users' equilibrium valuation threshold $v^\ast$ and the seller's uniform price $p_0^*$ versus user profiling accuracy $\delta$ in the PBE.}
	\label{FIG:PBE_delta}
    \vspace{-10pt}
\end{figure}

\begin{proposition}\label{Prop:PBE_delta}
The impact of user profiling accuracy $\delta$ on PBE is as follows.
\begin{itemize}
    \item [(i)] Users' equilibrium valuation threshold $v^\ast$ is non-monotonic in the profiling accuracy $\delta$. Specifically, it remains unchanged with $\delta$ over $[0,\hat{\delta})$, decreases with $\delta$ over $[\hat{\delta},\tilde{\delta})$, and increases over $[\tilde{\delta},1]$. 
    \item [(ii)] The seller's uniform price $p_0^*$ is non-decreasing in the profiling accuracy $\delta$. Specifically, it remains unchanged with $\delta$ over $[0,\hat{\delta})$, and then increases with $\delta$ over $[\hat{\delta},1]$.
\end{itemize}
\end{proposition}

Fig. \ref{FIG:PBE_delta} provides a graphical illustration of Proposition~\ref{Prop:PBE_delta}, where the blue and red curves denote users' equilibrium~valuation threshold $v^*$ and the seller's uniform price $p_0^\ast$ in the~PBE, respectively. In particular, we divide the feasible range of the possible profiling accuracy $\delta\in[0,1]$ into three regimes, as discussed below.

\begin{itemize}
\item \mbox{\textbf{Regime I with Low Accuracy $\delta\in[0,\hat{\delta}]$}} \emph{(all active users)}: 
The profiling accuracy $\delta$ is low in this regime, such that it is hard to profile users even when they~are fully active online. All users thus choose to remain~active in the social network. When the seller updates her belief on the non-profiled users' valuations, this results in~a distribution that remains uniform. Hence, the seller keeps the uniform price just the mean user valuation as in the no profiling benchmark ($\delta=0$).

\item \textbf{Regime II with Medium Accuracy $\delta\in(\hat{\delta},\tilde{\delta}]$} \emph{(partially active users with a decreasing $v^*$ in $\delta$)}: As the accuracy $\delta$ continues to increase, users are more likely to be~profiled and charged personalized prices when they are active~in the social network. This is especially true for users~with high valuations, who would suffer a great loss in purchase surplus and hence become inactive to thwart the seller's profiling. Consequently, the users' equilibrium valuation threshold $v^\ast$ decreases with $\delta$. This makes the distribution of the non-profiled users' valuations skew towards higher values as $\delta$ increases. Hence, the seller raises the uniform price to extract more revenue from those non-profiled users with higher valuations.

\item \textbf{Regime III with High Accuracy $\delta\in(\tilde{\delta},1]$} \emph{(partially~active users with an increasing $v^*$ in $\delta$)}: When $\delta$ becomes high, it is highly likely to successfully profile an active user. Thus, the seller has the incentive to stimulate users' social activities for better profiling. Specifically, the seller will keep raising the uniform price $p_0^*$ under larger $\delta$. This makes users suffer less from being profiled and charged personalized prices, and those inactive users in Regime II gradually choose to be active and enjoy social interactions again. Note that the seller's increasing uniform price $p_0^\ast$ also helps capture revenue from those remaining non-profiled users with high valuations.
\end{itemize}

\subsection{Welfare Implications for the Seller and Users}\label{Sec:welfare}

Next, we focus on the welfare implications of user profiling for both the seller and users. 

\begin{figure}[h]
    \vspace{-10pt}
	\centering
	\subfigure[The Seller's Expected Revenue]{
		\begin{minipage}{0.45\linewidth}
			\centering
            \includegraphics[width=0.8\linewidth]{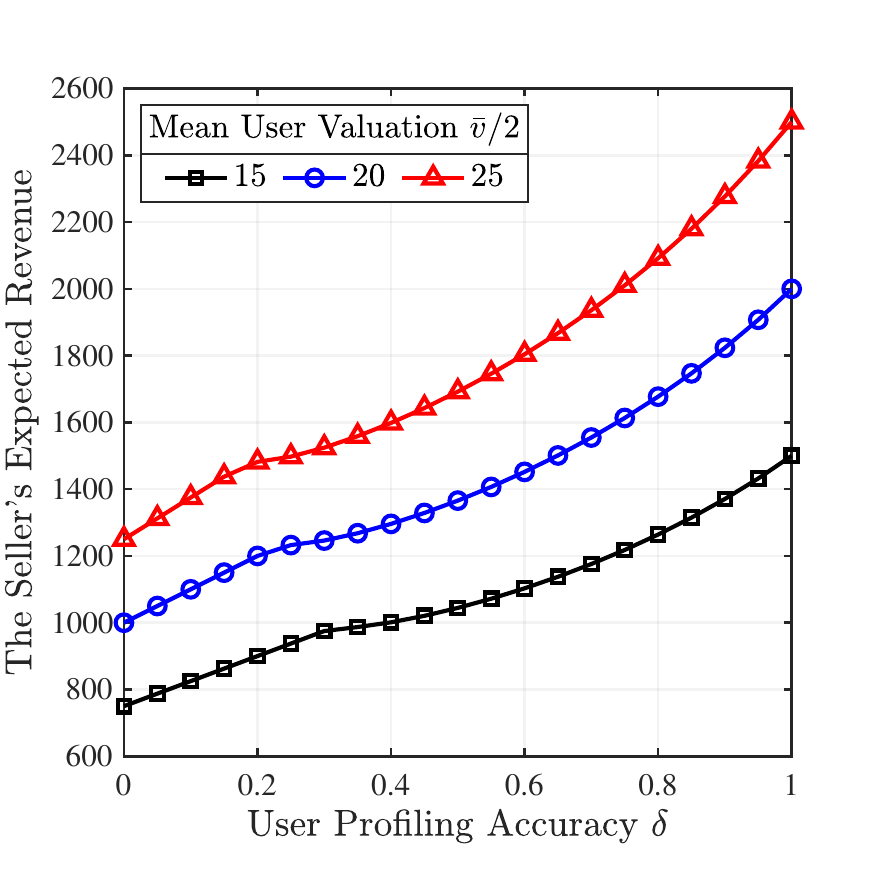}
			\label{fig_revenue_vbar}
		\end{minipage}
	}
	\subfigure[Revenue Proportions ($\bar{v}/2=20$)]{
		\begin{minipage}{0.45\linewidth}
			\centering
            \includegraphics[width=0.8\linewidth]{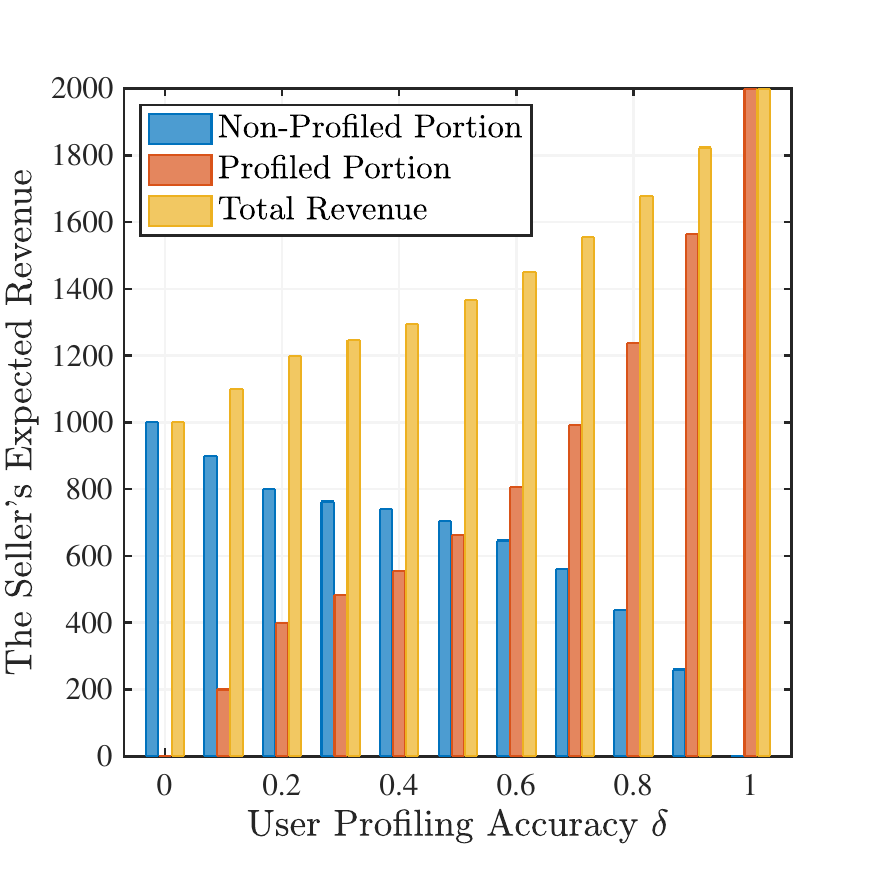}
			\label{fig_revenue_portions}
		\end{minipage}
	}
	\caption{The seller's expected revenue versus user profiling accuracy $\delta$.~We~consider $n=100$ users with uniformly distributed valuations, and we vary the mean user valuation $\bar{v}/2$ in set $\{15,20,25\}$ in Fig. \ref{fig_revenue_vbar} while fixing it as $20$~in Fig. \ref{fig_revenue_portions}.}
    \label{fig_revenue}
\end{figure}

Starting with the seller, we numerically analyze the changes in her expected revenue as the user profiling accuracy $\delta$~varies, as depicted in Fig. \ref{fig_revenue}. Specifically, we observe that the seller's expected revenue initially rises in a concave manner, then~increases convexly at higher $\delta$ values, as shown in Fig. \ref{fig_revenue_vbar}. This pattern can be attributed to the shifting proportions of revenue contributed by profiled and non-profiled users,~as~illustrated in Fig. \ref{fig_revenue_portions}. Notably, when $\delta$ is in a moderate range (e.g.,~$0.2$ to $0.5$ for $\bar{v}/2=20$), an increase in $\delta$ leads to more~users~going inactive to evade profiling (see Proposition~\ref{Prop:PBE_delta} and Fig. \ref{FIG:PBE_delta}). This inactivity then slows the revenue growth from profiled users, and consequently, the overall revenue growth is also affected (see Fig. \ref{fig_revenue_portions}).

Moving forward, we turn our attention to user implications, where Fig. \ref{fig_agguserpayoff} reveals that the average user payoff diminishes as the user profiling accuracy $\delta$ improves. This is because better user profiling enables the seller to exploit more information about the users from their social activities. However, this trend does not hold across all users, as suggested in the following proposition.

\begin{figure}[h]
	\centering
    \vspace{-10pt}
	\includegraphics[width=0.4\linewidth]{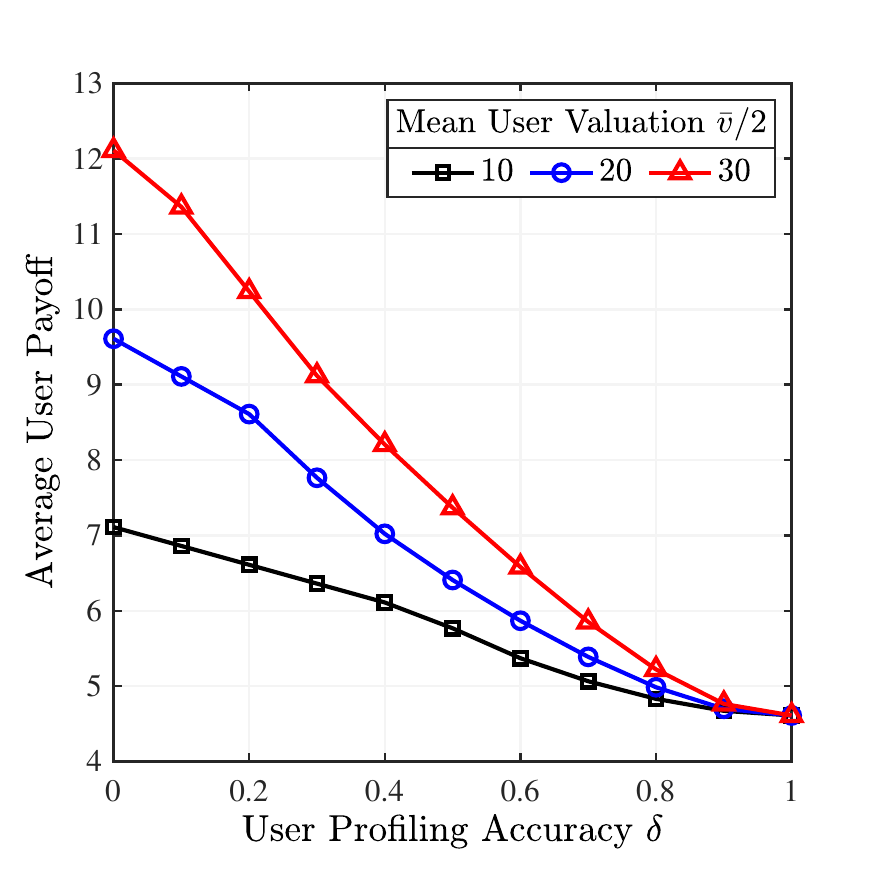}
    \vspace{-5pt}
	\caption{Average User Payoff versus User Profiling Accuracy $\delta$. Here,~we~consider $n=100$ users with uniformly distributed valuations and vary the mean user valuation $\bar{v}/2$ in set $\{10,20,30\}$.}
	\label{fig_agguserpayoff}
    \vspace{-5pt}
\end{figure}

\begin{proposition}\label{prop_userpayoff}
An individual user $i$'s payoff increases with the user profiling accuracy $\delta$ if \footnote{Here, the user valuation threshold $\hat{v}$ and the profiling accuracy threshold function $\delta^\dagger(\cdot)$ are given in Appendix \ref{Proof:Userpayoff_Increse}. In particular, $\delta^\dagger(v_i)$ is non-decreasing in $v_i$.\label{footnote4}}
\begin{equation}\label{prop_userpayoff_condition}
v_i<\hat{v} \text{ and } \delta > \delta^\dagger(v_i),
\end{equation}
i.e., for any low-valuation user $i$ with a valuation smaller than a threshold $v_i<\hat{v}$ and when the profiling accuracy is high within the range of $\left[\delta^\dagger(v_i),1\right]$.
\end{proposition}

Proposition \ref{prop_userpayoff} shows that some low-valuation users can~actually benefit from an increased profiling accuracy. This happens when the profiling accuracy $\delta$ is relatively high. In this case, as $\delta$ increases, more users become active in the social network due to the increasing uniform price (see Proposition \ref{Prop:PBE_delta}). Therefore, higher profiling accuracy results in more significant social network benefits, especially for those low-valuation users who are always active online and receive zero purchase surplus due to personalized pricing or an unaffordable uniform price.

Finally, we may question whether users benefit from knowing the seller's data access and profiling practices in social networks. With such awareness, some users may reduce their social activities and avoid personalized pricing in the equilibrium. To address this concern, we compare our PBE~in~Theorem \ref{Theorem_PBE} to a no-awareness benchmark, where the seller secretly accesses users' social activity data and performs user profiling without their knowledge or consent.

\begin{figure}[h]
\centering
\vspace{-15pt}
\begin{tikzpicture}
\node (image) at (0,0) {\includegraphics[width=0.4\linewidth]{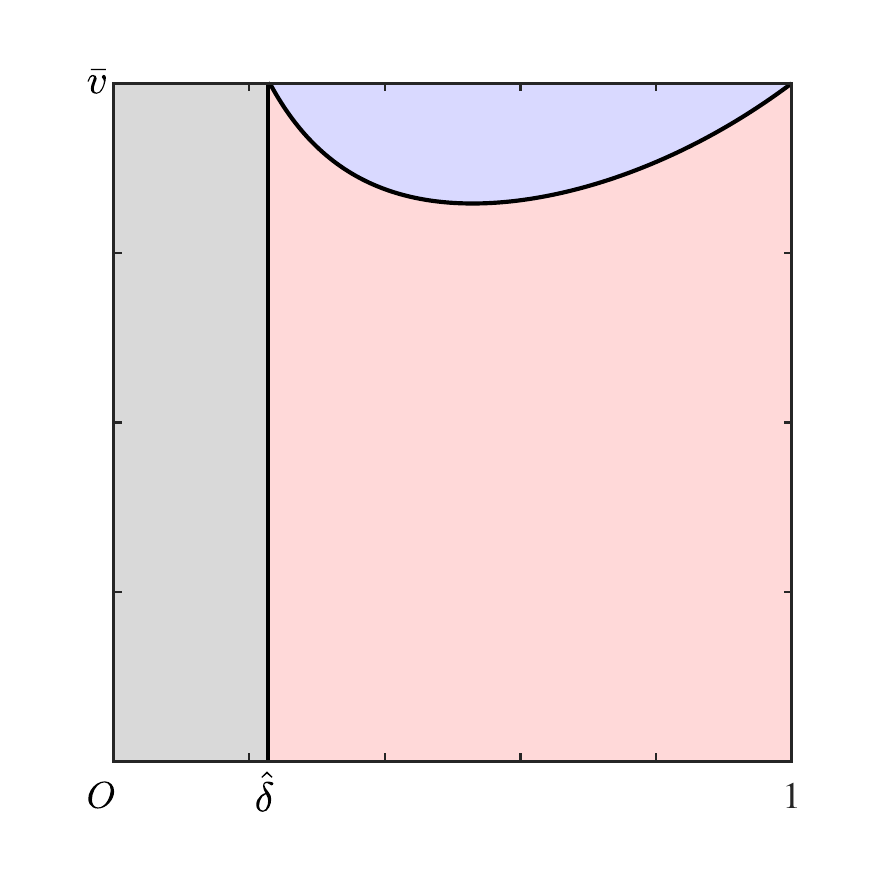}};
\node [align=center] at (0.43,2.3) {\footnotesize Aware $>$ Unaware};
\node [align=center] at (0.5,-0.1) {\footnotesize Aware $<$ Unaware \\\footnotesize for User Payoffs};
\node [align=center] at (-1.87,0) {\footnotesize Aware\\\footnotesize $=$\\\footnotesize Unaware};
\node at (0,-3.05) {\footnotesize User Profiling Accuracy $\delta$};
\node at (-2.85,0) {\footnotesize\rotatebox{90}{User Valuation $v_i$}};
\end{tikzpicture}
\vspace{-5pt}
\caption{Impact of user awareness: comparison of all users' payoffs, before~and after becoming aware of the seller's profiling in social networks under different profiling accuracy levels $\delta$. Here, the boundary function that separates the Blue and the Red regions is $2\delta v_i=2\ln(n-1+\omega_{\footnotesize\text{0}})+2p_{\footnotesize\text{0}}^*-(1-\delta)\bar{v}$.}
\label{fig_awareness}
\vspace{-10pt}
\end{figure}
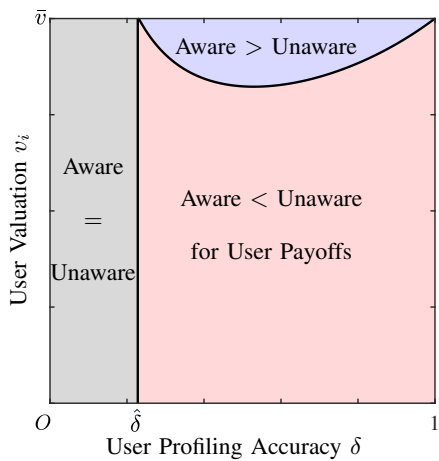

\begin{proposition}\label{Prop:Myopic-Compare}
Compared to the benchmark of users lacking awareness, an individual user $i$'s payoff reduces after becoming aware of the seller's profiling in social networks if and only if 
\begin{equation}
\hat{\delta}<\delta<1 \text{ and } v_i<v^\dagger\triangleq\frac{2\ln\left(n-1+\omega_0\right)+2p_0^*-(1-\delta)\bar{v}}{2\delta},
\end{equation}
as illustrated in Fig.~\ref{fig_awareness}.
\end{proposition}

As shown in Fig. \ref{fig_awareness}, users' awareness does not affect their payoffs when the profiling accuracy $\delta$ is small (see~the~Gray~region). Once $\delta$ is larger than a threshold $\hat{\delta}$, most users (except those with a high enough valuation $v_i$) are worse off after they become aware of the seller's profiling in social networks (see the Red region). Specifically, the fact that some high-valuation users become inactive with an increased profiling awareness imposes a social loss on users with low valuations. Moreover, non-profiled users who successfully avoid personalized pricing still experience a reduction in purchase surplus. This is because the seller tends to set a higher uniform price in response to users' awareness to stimulate their social activities.\footnote{In the no-awareness benchmark, all users remain fully active online in the equilibrium, and the valuations of those non-profiled users remain uniformly distributed over $[0,\bar{v}]$. The seller hence sets a uniform price of $p_0^{\text{no}}=\bar{v}/2$ in this benchmark. Compared to the uniform price $p_0^*$ in PBE of Theorem \ref{Theorem_PBE}, we have $p_0^*\ge  p_0^{\text{no}}$.}

On the other hand, only users with high enough valuations will benefit from their awareness of the seller's profiling (see the Blue region in Fig. \ref{fig_awareness}). Notice that these users~become~inactive online with awareness (i.e., $v^*<v^\dagger$), and thus the seller can only charge them a low uniform price $p_0$ instead of a high personalized price $v_i$. For such users with sufficiently high valuations, this results in a significant gain in their purchase surplus to offset the social cost of both their own and other users' endeavors to avoid the seller's profiling.

\section{Effect of Social Network Benefits}\label{Sec:Network}

Building on the well-established PBE in Section \ref{Sec:PBE}, we now investigate the effect of social network benefits on equilibrium behaviors and welfare outcomes. Specifically, we~begin~by~examining the effect on users' strategic information disclosure in Section~\ref{subsec:effect_user}, followed by an evaluation of the seller's benefits from leveraging users' social network profiles in Section \ref{subsec:effect_seller}.

\subsection{Effect on Users' Information Disclosure}\label{subsec:effect_user}

To gauge the effect of social network benefits, we introduce a no-social-network benchmark for comparisons, where~users determine whether to disclose their personal information (e.g., online purchase records and web browsing history) to~the~seller without considering any social network benefits. We adopt a simple voluntary-disclosure variant of the three-stage model~in Section \ref{Sec:Model}, replacing the users' decisions to be active/inactive in the social network in our original model as their individual choices regarding voluntary disclosure. In this context, users' decision-making becomes decoupled from each other, thereby eliminating the positive network externality among their information disclosure.

Similar to our analysis in Section~\ref{Sec:alternate}, we alternate backward induction with forward induction and characterize the PBE for this new benchmark in the following proposition.

\begin{proposition}\label{Benchmark:no-social-network}
In the no-social-network benchmark, there~exists a unique user-optimal PBE as follows.\footnote{Notice that multiple equilibria exist in this no-social-network benchmark, and we will focus on the user-optimal equilibrium to provide clear insights, which maximizes the user's ex-ante payoff across all equilibria. This refinement, known as the sender-optimal equilibrium, is widely adopted in the literature on incomplete information games (e.g., \cite{mailath1993belief}).}
\begin{itemize}
\item In Stage I, the users' equilibrium valuation threshold is $v^\star=\bar{v}/{2}$ such that $x_i^\star(v_i)=\mathbbm{1}(v_i\le v^\star)$,~i.e., any user $i$ with a valuation $v_i>v^\star$ chooses to conceal his private information ($x^\star_i=0$) in the online platform. In Stage II, the seller sets a uniform price of $p_0^\star=\bar{v}/{2}$.
\item The equilibrium valuation threshold in this benchmark~is strictly lower than that in the three-stage model with~social network benefits (see Theorem \ref{Theorem_PBE}), i.e., $v^\star < v^*$.
\end{itemize}
\end{proposition}

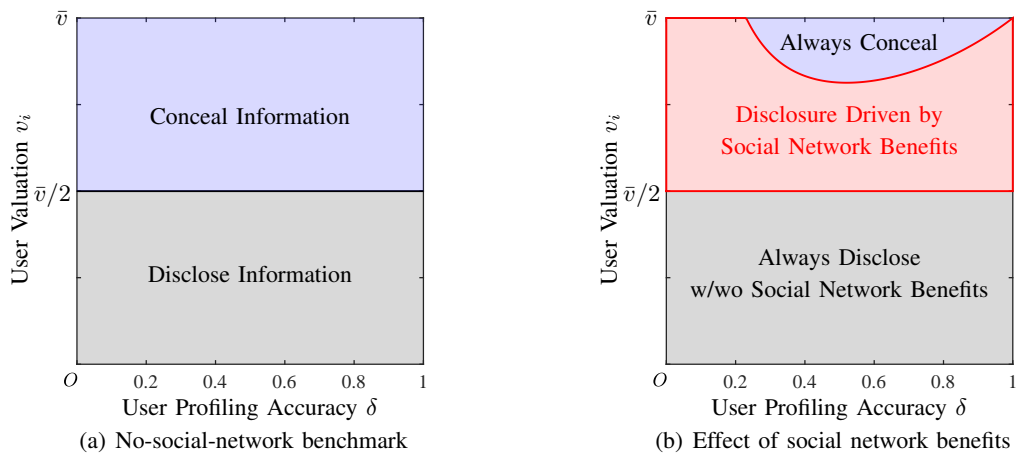
\begin{figure}[h]
    \vspace{-20pt}
    \subfigure[No-social-network benchmark]{
        \begin{minipage}{0.45\linewidth}
            \centering
            \begin{tikzpicture}
            \node (image) at (0,0) {\includegraphics[width=0.8\linewidth]{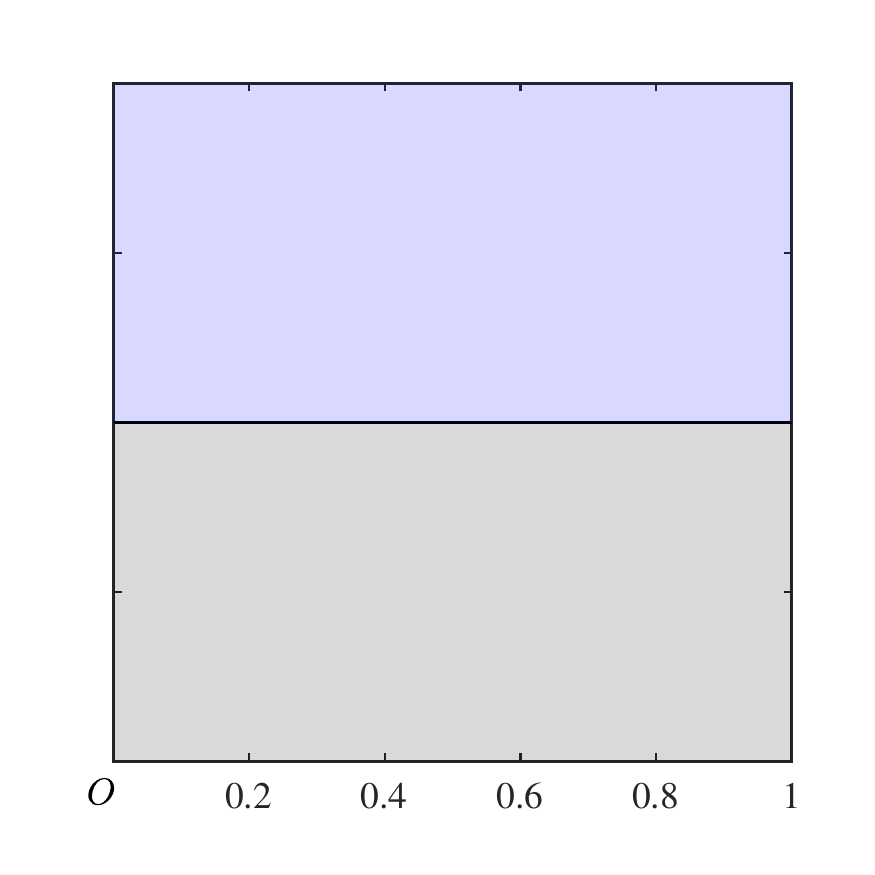}};
            \node [align=center] at (0.1,1.1) {\footnotesize Conceal Information};
            \node [align=center] at (0.1,-1) {\footnotesize Disclose Information};
            \node at (-2.95,0) {\footnotesize\rotatebox{90}{User Valuation $v_i$}};
            \node at (0.1,-2.8) {\footnotesize User Profiling Accuracy $\delta$};
            \node [align=center] at (-2.4,2.4) {\footnotesize $\bar{v}$};
            \node [align=center] at (-2.525,0.05) {\footnotesize $\bar{v}/{2}$};
            \end{tikzpicture}
            \label{fig_subnet_nobenchmark}
		\end{minipage}
	}
	\subfigure[Effect of social network benefits]{
        \begin{minipage}{0.45\linewidth}
            \centering
            \begin{tikzpicture}
            \node (image) at (0,0) {\includegraphics[width=0.8\linewidth]{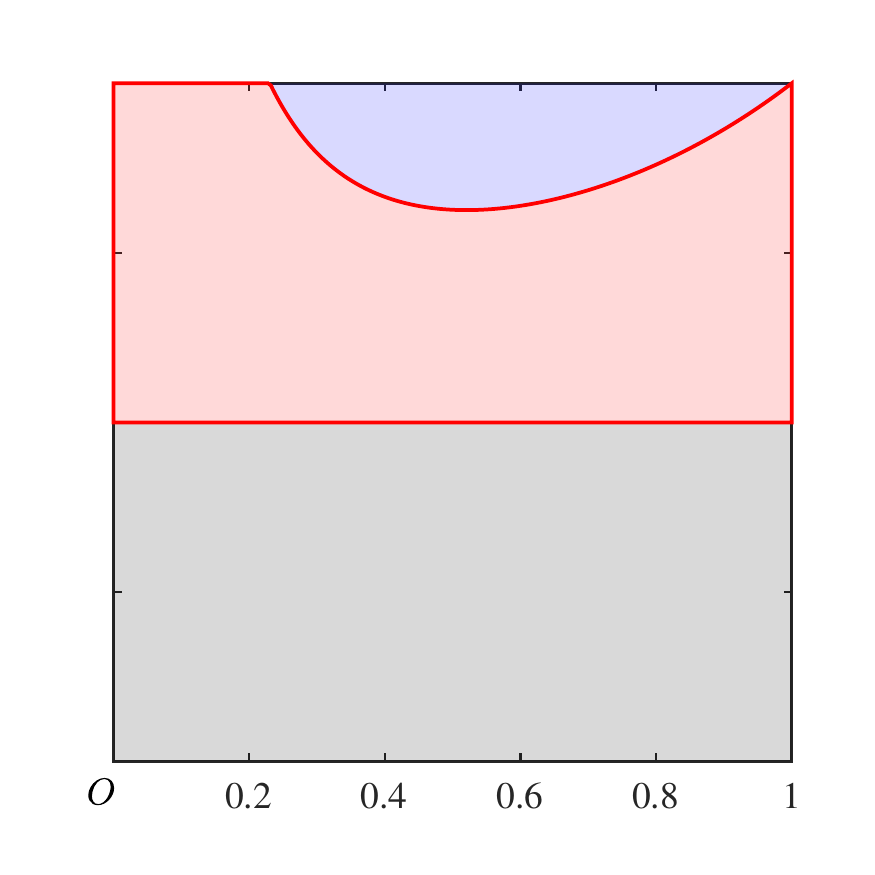}};
            \node [align=center] at (0.36,2.05) {\footnotesize Always Conceal};
            \node [align=left] at (0.1,1.1) {\footnotesize \color{red} Disclosure Driven by};
            \node [align=left] at (0.1,0.7) {\footnotesize \color{red} Social Network Benefits};
            \node [align=center] at (0.1,-0.80) {\footnotesize Always Disclose};
            \node [align=center] at (0.1,-1.20) {\footnotesize w/wo Social Network Benefits};
            \node at (-2.95,0) {\footnotesize\rotatebox{90}{User Valuation $v_i$}};
            \node at (0.1,-2.8) {\footnotesize User Profiling Accuracy $\delta$};
            \node [align=center] at (-2.4,2.4) {\footnotesize $\bar{v}$};
            \node [align=center] at (-2.525,0.05) {\footnotesize $\bar{v}/{2}$};
            \end{tikzpicture}
            \label{fig_no-social-network_effect}
		\end{minipage}
	}
\caption{(a) All users' voluntary disclosure decisions in the PBE under different profiling accuracy levels $\delta$ in the no-social-network benchmark,~and~(b)~the~effect of social network benefits: comparison of all users' disclosure decisions with versus without social network benefits. The boundary function $v^*(\delta)$, as specified in Theorem \ref{Theorem_PBE}, separates the blue from the red regions of Fig. \ref{fig_no-social-network_effect}.}
\label{fig:no-social-network-benchmark}
\end{figure}

Proposition \ref{Benchmark:no-social-network} highlights that some low-valuation users with $v_i<\bar{v}/2$ always voluntarily disclose their types to the seller despite potential profiling, as depicted in Fig. \ref{fig_subnet_nobenchmark}. The reason for this is that these users prefer being profiled to obtain~a lower personalized price $v_i$ rather than a higher uniform~price $p_0^\star=\bar{v}/{2}$. Furthermore, in this benchmark, the users' equilibrium valuation threshold $v^\star=\bar{v}/{2}$ remains unaffected by the profiling accuracy $\delta$. Without social network benefits, each user will independently make disclosure decisions and only need to consider the tradeoff between receiving a uniform price and a personalized one. This is in contrast to users' non-monotonic valuation threshold $v^*$ in the presence of social network benefits (see Proposition \ref{Prop:PBE_delta}), as illustrated in Fig. \ref{FIG:PBE_delta}.

Finally, we investigate the effect of social network benefits on all users' information disclosure decisions. Proposition \ref{Benchmark:no-social-network} suggests that fewer users will disclose their private information in the absence of social network benefits (i.e.,  \mbox{$v^\star < v^*$}). In particular, the red region in Fig. \ref{fig_no-social-network_effect} graphically illustrates the range of users $i$ who would not voluntarily disclose their types for user profiling but choose to do so for social network benefits, i.e., $v^\star<v_i<v^*$.

\subsection{Performance Evaluation for the Seller}\label{subsec:effect_seller}

Next, we evaluate the performance of our proposed pricing mechanism for the seller, through which we will understand the power of harnessing social network benefits. To evaluate the performance, we will introduce two additional benchmarks adapted from the existing literature (e.g., \cite{conitzer2012hide,2017Is,zhang2011perils,2020Voluntary}). Due to space constraints, we have deferred the detailed descriptions of the pricing benchmarks and the performance evaluation to Appendix \ref{Appendix:Performance}.

\section{Impact of Pricing Timing}\label{Sec:Sequential}

Thus far, we have assumed that the seller determines~the~uniform and personalized prices simultaneously in Stage II of~the three-stage model in Section \ref{Sec:Model}. However, some sellers currently employ alternative pricing practices where they initially make a pricing promise but break it later to personalize final offerings \cite{thecouncilofeconomicadvisers_2015_the,chen2020competitive}. This raises the issue of whether the seller should adopt such promise-breaking pricing practices while profiling users in the social network. To address this concern, in this section, we further consider the seller's flexibility in the pricing timing, i.e., the seller can announce the uniform price before users' social activity decisions. Such a pricing practice will make the uniform price publicly observable to all users. Notice that personalized pricing is always determined after user profiling and offered to the profiled users privately, the same as our original model. 

\begin{figure}[h]
	\centering
	\includegraphics[width=\linewidth]{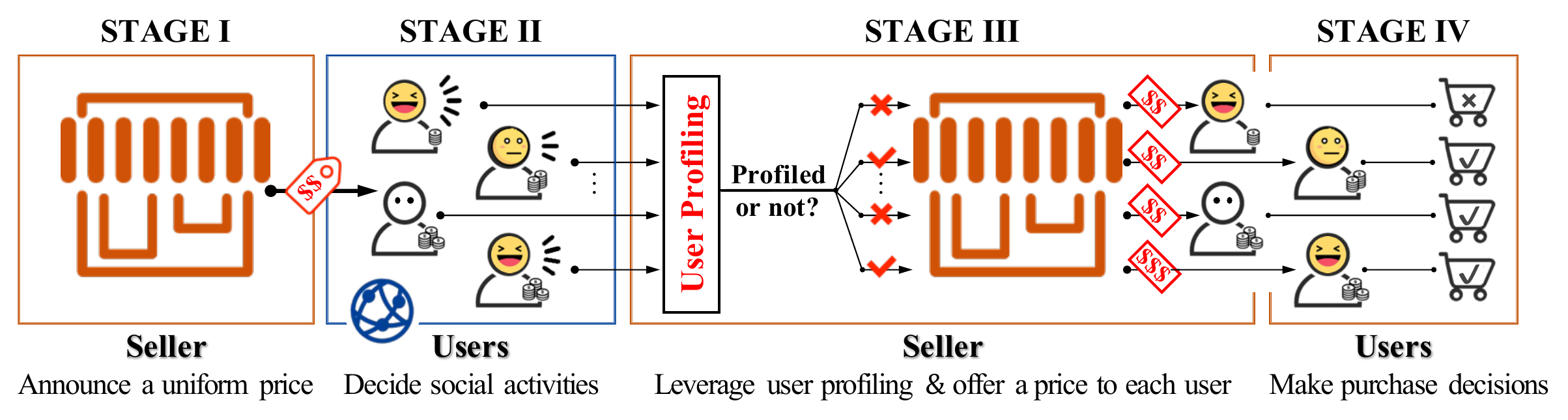}
	\caption{A Four-stage Extended Model: The seller now has the flexibility in pricing timing to announce the uniform price initially before users' social activity decisions.}
	\label{timeline:sequential}
\end{figure}

As illustrated in Fig. \ref{timeline:sequential}, we formulate the interactions~between the seller and users as a four-stage dynamic game:

\begin{itemize}
\item {\textbf{\mbox{Stage I}}}: The seller announces the uniform price~$p_0$.
\item {\textbf{\mbox{Stage II}}}: Each user $i\in\mathcal{N}$ decides his social activity~level $x_i\in[0,1]$.
\item {\textbf{\mbox{Stage III}}}: The seller first leverages the profiling technology to identify users' valuations (if possible). Based on the profiling results, the seller charges the profiled users a personalized price $p_i=v_i$. For non-profiled users, the seller charges the uniform price $p_0$ announced in Stage I.
\item {\textbf{\mbox{Stage IV}}}: Each user $i\in\mathcal{N}$ makes his purchase decision $d_i\in\{0,1\}$ after observing the final price offers.
\end{itemize}

As the seller provides personalized prices after announcing the uniform price, she actually breaks the pricing commitment to some users. This may harm the seller's reputation or even violate related regulations on pricing fraud \cite{lee_2021_chinese}. Later~in~this section, we will show the surprising result:~\textit{Even without~considering reputation or regulation issues, such flexibility in pricing timing only marginally improves the seller's average revenue but at the cost of higher variance.} Therefore, it is not reasonable for the seller to adopt a flexible pricing timing approach. This further justifies our use of the original model in Section \ref{Sec:Model} as the focus of this paper.

\subsection{New Equilibrium Analysis}\label{flexibility_new_equilibrium}

In this subsection, we analyze the equilibrium of the four-stage extended model in Fig. \ref{timeline:sequential}. Different from our original model in Fig. \ref{system_model}, we can analyze the extended model purely through backward induction. The key reason is that the coupling~between the seller's uniform pricing and the users' social activity decisions no longer exists in the four-stage model~here. This is because the seller determines the uniform price in Stage I only based on the prior belief of users' valuation distribution, which is not affected by users' social decisions. Therefore, we do not need forward analysis to derive an updated belief to determine the seller's uniform price as in Section \ref{Sec:alternate}. 

Notice that the coupling among users' social decisions~continues to hold. As users still have incomplete information~about the others' valuations, we formulate users' social interactions in Stage II as a Bayesian game. Similar to our original~model, we still manage to characterize the valuation threshold structure ($v^e$) of users' social decisions, such that each user $i$'s~equilibrium social activity level is given by $x_i^e=\mathbbm{1}(v_i\le v^e)$ as in Section \ref{Sec:alternate}. Hereafter, we use the superscript $e$ to refer to the extended four-stage model.

Next, we present the equilibrium result of our extended four-stage model. Notice that we only need to show the~seller's~uniform price in Stage I and the users' social activity decisions in Stage II. For the seller's personalized pricing scheme in~Stage III and users' purchase decisions in Stage IV, they are the~same as in Section~\ref{subsection:game}. 

\begin{proposition}\label{equilibrium_extension_model}
In the extended four-stage model, there exists a unique equilibrium as follows.
\begin{itemize}
\item \textbf{\emph{{Case I} (all active users with low uniform price)}} with a small mean valuation
\begin{equation}
\frac{\bar{v}}{2}\le\frac{1}{\delta}\ln\left(n-1+\omega_0\right).
\end{equation}
\begin{itemize}
    \item [(i)] In Stage I, users' equilibrium valuation threshold is $v^e=\bar{v}$, i.e., every user $i\in\mathcal{N}$ is active in the social network with $x_i^e=1$.
    \item [(ii)] In Stage II, the seller sets a uniform price of $p_0^e=\bar{v}/2$.   
\end{itemize}

\item \textbf{\emph{{Case II} (all active users with high uniform price)}} with a medium mean valuation\footnote{Here, the mean valuation threshold $\tilde{v}(\delta)$ is provided in Appendix \ref{Proof:PBE_extension}.\label{footnote6}}
\begin{equation}\label{case2_extension}
\frac{1}{\delta}\ln\left(n-1+\omega_0\right)<\frac{\bar{v}}{2}\le \tilde{v}(\delta).
\end{equation}
\begin{itemize}
    \item [(i)] In Stage I, users' equilibrium valuation threshold is $v^e=\bar{v}$, i.e., every user $i\in\mathcal{N}$ is active in the social network with $x_i^e=1$.
    \item [(ii)] In Stage II, the seller sets a uniform price of
    \begin{equation}\label{case2_extension_price}
    p_0^e=\bar{v}-\frac{1}{\delta}\ln\left(n-1+\omega_0\right),
    \end{equation}
    which is higher than $\bar{v}/2$.
\end{itemize}

\item \textbf{\emph{{Case III} (partially active users)}} with a large mean~valuation
\begin{equation}
\frac{\bar{v}}{2}>\tilde{v}(\delta).
\end{equation}
\begin{itemize}
    \item [(i)] In Stage I, users' equilibrium valuation threshold is $v^e<\bar{v}$, i.e., there exist some high-valuation users who are inactive in the social network. 
    \item [(ii)] In Stage II, the seller sets a uniform price of
    \begin{equation}
    p_0^e=v^e-\frac{\tilde{J}(v^e)}{\delta},
    \end{equation}
    which is less than the uniform price in \eqref{case2_extension_price}.
\end{itemize}
\end{itemize}
\end{proposition}

Compared with our original three-stage model (see Theorem \ref{Theorem_PBE}), the extended four-stage model here exhibits a new equilibrium: Case II in Proposition \ref{equilibrium_extension_model}. Given the flexibility~to~announce the uniform price before users' social decisions, the seller can now raise the uniform price $p_0^e$ from the mean user valuation in Case I to the value in (\ref{case2_extension_price}), so as to proactively motivate all users to remain active online as in Case I. To~better illustrate the role of such flexibility, we next formally compare the equilibrium under this extended four-stage model to our original three-stage model.

\begin{corollary}\label{extension_compar}
Compared to the original three-stage model in Section~\ref{Sec:Model}, both the seller's equilibrium uniform price in Stage I and the users' equilibrium valuation threshold in Stage II are higher in the extended four-stage model, i.e., $v^e\ge v^*$ and $p_{0}^e\ge p_{0}^*$.
\end{corollary}

Corollary \ref{extension_compar} demonstrates that the seller's flexibility in~pricing timing leads to an increase in users' social activity levels. This is because the flexibility allows the seller to announce a higher uniform price initially before users make their social decisions, which incentivizes their online activity levels.

\subsection{Impact of Flexibility in Pricing Timing on Seller's Revenue}

Next, we compare the seller's expected revenues in the~extended four-stage model and our original three-stage model.

\begin{proposition}\label{sale_revenue_compare}
Compared to the original three-stage model in Section \ref{Sec:Model}, the seller's expected profiled revenue is higher in the extended four-stage model, whereas the expected non-profiled revenue is lower. Overall, the seller's total expected revenue is higher in the extended four-stage model, where the improvement ratio is upper bounded by~\mbox{$(\delta-\delta^2)/2$} (no more than $12.5\%$).
\end{proposition}

Proposition \ref{sale_revenue_compare} indicates that the flexibility in pricing timing only mildly improves the seller's total expected revenue. Note that the bounded ratio $(\delta-\delta^2)/2$ in Proposition \ref{sale_revenue_compare} is not~tight. Compared to our original three-stage model, the seller extracts less revenue from the non-profiled users in the extended four-stage model. However, the seller generates more revenue with an increased number of profiled users,~which~echoes~the~findings in Corollary \ref{extension_compar}.

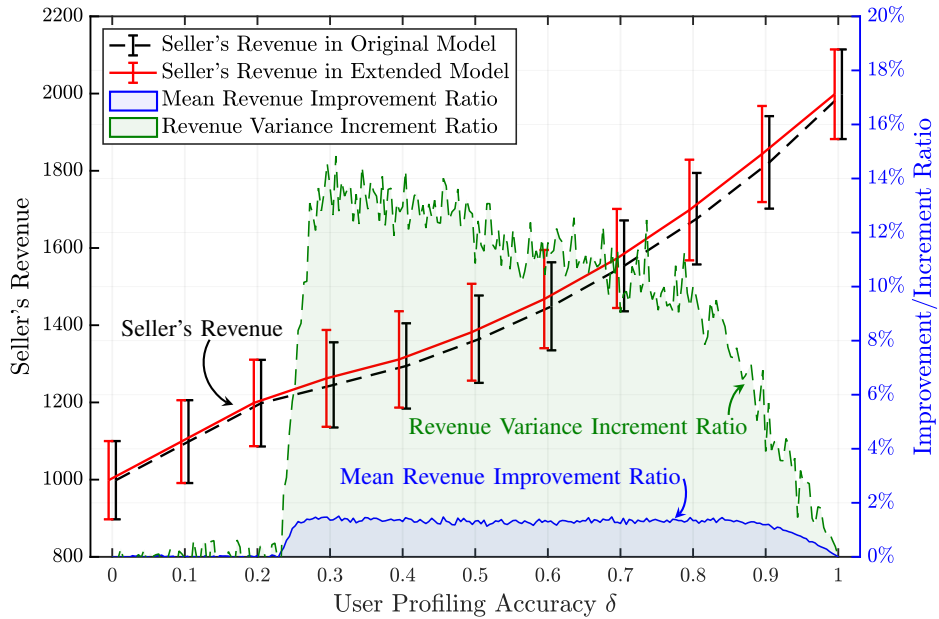
\begin{figure}[h]
	\centering
    \vspace{-10pt}
	\begin{tikzpicture}
	    \definecolor{darkgreen}{rgb} {0.00,0.5,0.00}
		\node (image) at (0,0) {\includegraphics[width=0.8\linewidth]{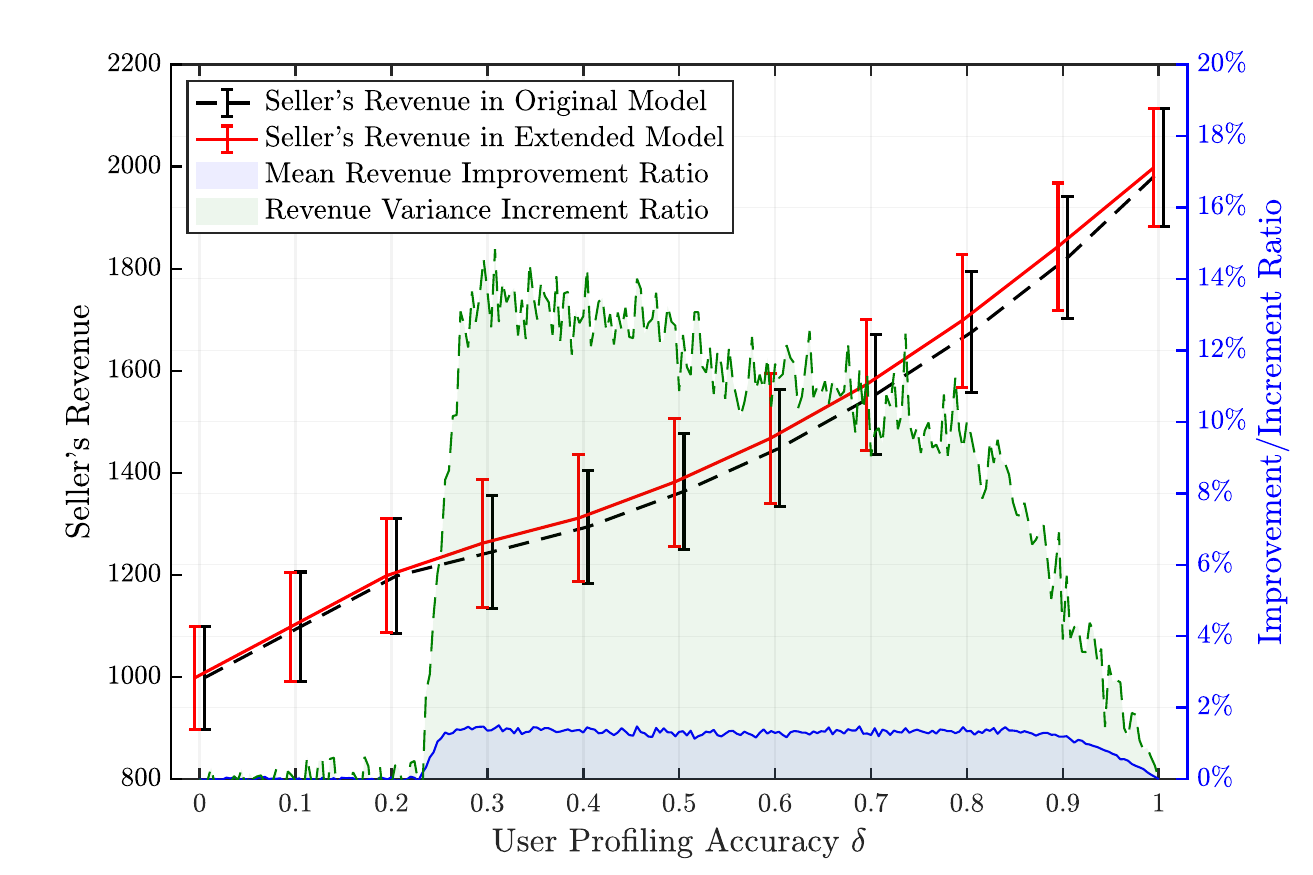}};
		\draw [thick,draw=blue] (-4.625,2.5) rectangle (-4,2.75);
		\draw [thick,draw=darkgreen] (-4.625,2.15) rectangle (-4,2.4);
		\draw [-stealth,thick,black] (-3.3,-0.55) arc (150:250:0.6);
		\node [align=center,black] at (-3.4,-0.4) {\footnotesize Seller's Revenue};
		\draw [-stealth,darkgreen,thick] (3.6,-1.55) arc (180:125:0.5);
		\node [align=center,darkgreen] at (1.6,-1.7) {\footnotesize Revenue Variance Increment Ratio};
		\draw [-stealth,blue,thick] (3,-2.4) arc (20:-45:0.5);
		\node [align=center,blue] at (0.7,-2.4) {\footnotesize Mean Revenue Improvement Ratio};
	\end{tikzpicture}
    \vspace{-5pt}
	\caption{The seller's mean revenue and revenue variance comparisons between the extended model and our original model, versus user profiling accuracy $\delta$. Here, we consider $n=100$ users with uniformly distributed valuations and fix the mean valuation $\bar{v}/2$ as $20$. We run $10000$ times the experiments to obtain the data results.}
	\label{variance_compare_fig}
\end{figure}

Finally, we turn our attention to the variance in the seller's revenue. Notice that examining revenue variance could assist the seller's risk assessment \cite{baron2022revenue} by telling the potential revenue volatility that mainly stems from randomness in user profiling. We conduct extensive simulation experiments to compare the actual revenue distributions in the two models. As shown in Fig.~\ref{variance_compare_fig}, the black dashed and the red solid curves represent the seller's average revenues of our original three-stage model and the extended four-stage model, respectively. We also plot vertical error bars across the curves to illustrate the revenue variances. Moreover, we graphically depict the improvement ratio of the seller's average revenue and the increment ratio of the seller's revenue variance in the extended model compared to our original model through blue and green curves in Fig.~\ref{variance_compare_fig}, respectively. Actually, we note that the seller can only improve the revenue by no more than $2\%$ with her flexibility in pricing timing.

\begin{observation}
In the extended four-stage model, the seller's revenue variance is up to $14\%$ higher compared to~the~original three-stage model.
\end{observation}

The variation of the seller's revenue increases in the extended four-stage model. This is because more users are~incentivized to be active and potentially subject to profiling in the extended model, which results in a larger variance under the random user profiling.\footnote{It is also natural to attribute an increased revenue variance to more socially active users within the original three-stage model. However, the extended four-stage model here differs from the original three-stage model in terms of how~it affects users' social network activities (see Corollary \ref{extension_compar}).} In conclusion, the flexibility in pricing timing only mildly improves the seller's expected revenue but at the cost of higher variance. Considering the additional risks associated with reputation and regulation, it is not a reasonable choice for the seller to adopt such flexibility in pricing timing.

\section{Extensions to Heterogeneous Online Social Networking}\label{extension:heterogeneous}

Thus far, our analysis has only considered the heterogeneity in user valuations and assumed a complete network for social interactions. In this section, we extend our original three-stage model in Section \ref{Sec:Model} to explore the heterogeneity in users' online social networking.

Specifically, we first incorporate users' heterogeneous social network benefits into our original three-stage model in Section \ref{subsec:extension:benefits}. Then, in Section \ref{subsec:extension:partial}, we further extend to consider a partially connected social network where the users possess heterogeneous network positions. Within these extensions, we can still identify the threshold structure of PBE (see Belief~\ref{threshold_structure}). However, it is important to notice that the threshold for users may vary depending on their social network benefits or network positions, as suggested in the subsequent analysis.


\subsection{Heterogeneous Social Network Benefits}\label{subsec:extension:benefits}

Let us first extend our original three-stage model in Section \ref{Sec:Model} to incorporate the heterogeneity of users' social network benefits. Suppose that $n$ users can be classified into $K$ different types of preferences for online social networking, denoted by $\mathcal{K}=\{1,2,\dots,K\}$. Each type $k\in\mathcal{K}$ appears in the population with a probability $\gamma_k$, where $\sum_{k\in\mathcal{K}}\gamma_k=1$.~Furthermore,~each type-$k$ user $i\in\mathcal{N}$ is associated with a non-negative preference parameter $\alpha_k$ for his social network benefit $J_i$ in \eqref{social_network_benefit}. Accordingly, we reformulate the payoff for a type-$k$ user $i$~with valuation $v_i$ as
\begin{equation*}
    \pi_{<i,k>}\left(x_i,\boldsymbol{x}_{-i}\right)=\alpha_k J_i\left(x_i,\boldsymbol{x}_{-i}\right)+\max\left\{v_i-p_i,0\right\},
\end{equation*}
which extends the original payoff formulation in \eqref{final_user_payoff} that assumes homogeneous preferences among users (i.e., normalized $\alpha_k=1$ for all $k\in\mathcal{K}$). We assume that each user's preference type is private and known only to himself. Moreover, the seller knows the prior distribution of users' preference types, and each user is aware of the distribution of other users' preference types. Similar to the analysis in Section \ref{Sec:alternate}, we will alternate forward induction (for users' social activity decisions) with backward induction (for the seller's uniform pricing) to guide the PBE analysis for this new extension.

Now, we start with the forward analysis of users' social interactions in Stage I. In this new extension, the inclusion of heterogeneity in social network benefits disrupts the existence of a common valuation threshold $v^*$ for users' social activity decisions, which is identified in our original model with~heterogeneity solely in user valuations (see Belief \ref{threshold_structure}). In fact, we reveal a natural belief variation wherein the valuation threshold varies across users with different preference types for online social networking, as detailed in the following structural belief.

\begin{lemma}\label{belief_hetergoneous}
For each user preference type $k\in\mathcal{K}$, there exists a valuation threshold $v^*_k\in[0,\bar{v}]$ such that the social activity decision in Stage I for any type-$k$ user $i$ with valuation $v_i$ is characterized as
\begin{equation*}
    x_{<i,k>}^*(v_i)=\mathbbm{1}\left(v_i\le v^*_k\right).
\end{equation*}
\end{lemma}\noindent
Furthermore, we characterize these valuation thresholds $\{v_k^*,\forall k\in\mathcal{K}\}$ in the following proposition, where we extend the original formulation in \eqref{vstar_p0_summation} for the common threshold $v^*$.

\begin{proposition}\label{prop:hetergoneous_threshold}
(i) For each preference type $k\in\mathcal{K}$, the~valuation threshold $v_k^*$ in Lemma \ref{belief_hetergoneous} is $\min\{v^\dagger_k,\bar{v}\}$, where $v^\dagger_k$ satisfies
\begin{equation}\label{equ:hetergoneous_threshold}
    v^\dagger_k=p_0+\frac{\alpha_k}{\delta}\sum_{m=0}^{n-1}\tbinom{n-1}{m}\ln(m+\omega_0)\left(\sum_{k=1}^K \gamma_k F(v^\dagger_k)\right)^m\left(1-\sum_{k=1}^K \gamma_k F(v^\dagger_k)\right)^{n-1-m}
\end{equation}

(ii) Moreover, the type-specific valuation threshold increases with the type-associated user preference for the social network benefit, i.e., for any preference type $k,l\in\mathcal{K}$, we have $v^*_k\le v^*_l$ if $\alpha_k<\alpha_l$.
\end{proposition}\noindent
Specifically, the formulation in \eqref{equ:hetergoneous_threshold} replaces the fraction of socially active users $F(v^\dagger)$ in \eqref{vstar_p0_summation} with a weighted variation $\sum_{k=1}^K \gamma_k F(v^\dagger_k)$, and further takes into account the heterogeneous preference $\alpha_k$ for the social network benefit. Moreover, Proposition \ref{prop:hetergoneous_threshold} demonstrates the intuition that users who value online social networking more are less willing to become inactive to avoid personalized pricing. This naturally results in a \textit{step-wise} threshold structure that monotonically increases when users are arranged in ascending order of their preferences for online social networking.

Building upon Lemma \ref{belief_hetergoneous}, we then turn to the backward analysis of the seller's uniform pricing $p_0$ in Stage II. The analysis will follow a similar rationale in Section \ref{backward}, although it will be more involved due to the diverse type-specific decision-making thresholds among users. We conclude with Corollary \ref{coro:hetergoneous_uniform} below, where the new formulation for the uniform pricing scheme in~\eqref{equ:hetergoneous_price} degenerates into \eqref{base_price_equation3} in Proposition \ref{base_price} under homogeneous preferences among users (i.e., $v^*_k=v^*$ for all $k\in\mathcal{K}$). 

\begin{corollary}\label{coro:hetergoneous_uniform}
(i) The seller's optimal uniform price satisfies $p_0^*<\min\{v_k^*;k\in\mathcal{K}\}$ in the PBE (whenever it exists). 

(ii) Moreover, given an arbitrary set of type-specific valuation thresholds $\{v_k^*,\forall k\in\mathcal{K}\}$ for users in Stage I, the seller's optimal uniform pricing scheme in Stage II is given by
\begin{equation}\label{equ:hetergoneous_price}
    p_0^*\left(v_k^*,\forall k\in\mathcal{K}\right)=\left(2\sum_{k\in\mathcal{K}}\gamma_k\frac{1-\delta}{\bar{v}-\delta v^*_k}\right)^{-1}.
\end{equation}
\end{corollary}

While previous analyses have derived key structural results for equilibrium outcomes, the final analysis for PBE remains challenging in general. The additional challenges stem from a more intricate information structure: (i) diverse type-specific decision-making thresholds among users within the context of doubly-coupled interactions, and (ii) the need to ensure belief consistency in the presence of users' multi-dimensional private information, encompassing both product valuation and preference type. On a technical level, replicating the analysis from Section \ref{Sec:PBE} to find PBE involves solving a rather complex system of non-linear, multivariate, and high-order polynomials. However, as an alternative approach, the step-wise threshold structure in users' social decision-making can potentially serve as a valuable guide for developing algorithmic solutions, and we will leave it as an interesting avenue for future research.

\subsection{Partially Connected Social Networks}\label{subsec:extension:partial}

In this subsection, we expand our original three-stage model in Section \ref{Sec:Model} to consider a partially connected social network. Specifically, we consider the set of users $\mathcal{N}$ interacting through an online social network, represented by a graph $\mathcal{G}$. We assume that the underlying graph is unweighted and publicly known. The edges of $\mathcal{G}$ represent the social relationship within the network. We denote the adjacency matrix of $\mathcal{G}$ as $G$, where the $(i,j)$-th entry encodes whether users $i$ and $j$ are connected in the underlying network, i.e., $g_{ij}\in\{1,0\}$, with $1$ for being connected and $0$ otherwise. Accordingly, we reformulate the social network benefit $J_i$ for each user $i$ as
\begin{equation*}
    J_i(x_i,\boldsymbol{x}_{-i})=x_i\ln\left(\sum_{j\neq i }g_{ij}x_j+\omega_0\right),
\end{equation*}
which extends the original formulation in \eqref{social_network_benefit} under a fully connected network (i.e., $g_{ij}=1$ for any $i,j\in\mathcal{N}$). Furthermore, in this subsection, we consider an extended pricing setting, where the seller's uniform pricing strategy for non-profiled users can vary depending on their network positions, i.e., $\{p_0^i,\forall i\in\mathcal{N}\}$. Notice that the seller still charges a personalized price \mbox{$p_i=v_i$} for each profiled user, the same as our original model.

\begin{figure}[h]
	\centering
    \vspace{-15pt}
	\subfigure[Subnetwork Visualization]{
        \begin{minipage}{0.45\linewidth}
            \centering
            \begin{tikzpicture}
            \node (image) at (0,0) {\includegraphics[width=0.8\linewidth]{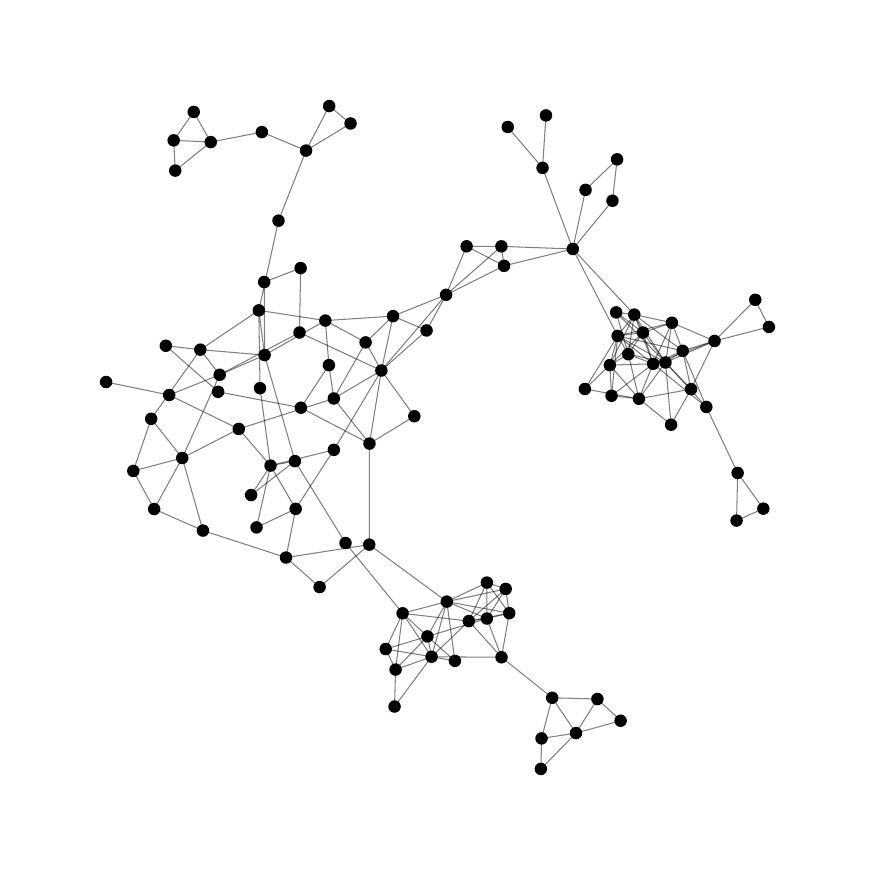}};
            \end{tikzpicture}
            \label{fig_subnet_origin}
            \vspace{-15pt}
		\end{minipage}
	}
    \hspace{-10pt}
	\subfigure[Social Activity Levels]{
        \begin{minipage}{0.45\linewidth}
            \centering
            \begin{tikzpicture}
            \node (image) at (0,0) {\includegraphics[width=0.8\linewidth]{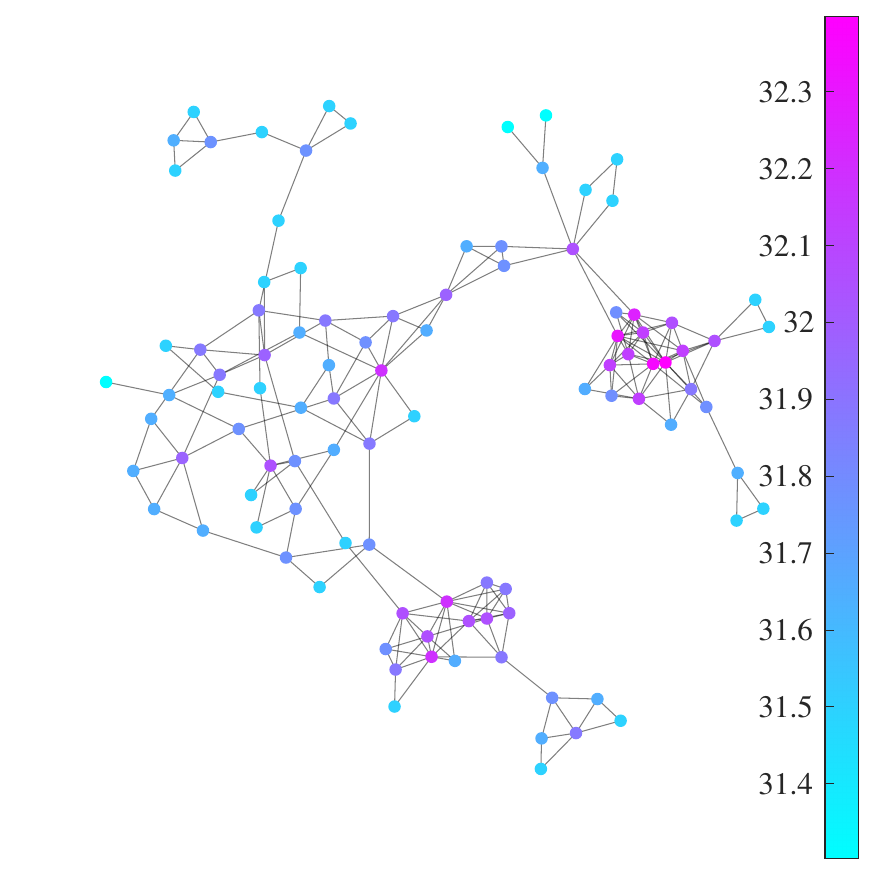}};
            \node at (3.2,0) {\footnotesize\rotatebox{90}{Different User $i$'s Valuation Threshold $v_i^*$}};
            \end{tikzpicture}
            \label{fig_subnet_activity.pdf}
            \vspace{-15pt}
		\end{minipage}
	}
	\caption{A subnetwork of the Facebook social network sourced from \cite{snapnets}, and an illustration of different users' social activity levels~in the PBE (indicated by the equilibrium valuation thresholds $v_i^*$ for each user $i\in\mathcal{N}$) with a profiling accuracy $\delta=0.7$.~Here,~we~consider that the valuation of each subnetwork user follows a uniform distribution over $[0, 40]$.}
    \vspace{-5pt}
    \label{fig_sub_visual}
\end{figure}

How should users manage online activities when interacting within a partially connected~social network? Moreover, considering the diverse and distinctive network positions among users, can the degree of each user node alone suffice to indicate their actions? To answer these questions, we will first present key structural results for the PBE of this new extension in Appendix \ref{Appendix:Partially}, based on which we then conduct empirical experiments utilizing Facebook social network data from \cite{snapnets}.

For our empirical studies, we utilize Facebook social network data from \cite{snapnets} and focus on a typical induced subnetwork of the original Facebook network in \cite{snapnets}. Specifically, the subnetwork consists of $100$ user nodes and $215$ edges, with an average degree of $4.3$. Fig. \ref{fig_subnet_origin} provides a visualization of this subnetwork. Furthermore, in Fig. \ref{fig_subnet_activity.pdf}, we graphically illustrate users' social activity levels in the PBE, as indicated by their equilibrium valuation thresholds $\{v^*_i,\forall i\in\mathcal{N}\}$, with a profiling accuracy $\delta=0.7$. In particular, the users who appear to be more connected and central in the network have higher valuation thresholds and are more likely to be active online.

\begin{figure}[h]
    \centering
    \vspace{-5pt}
    \includegraphics[width=0.6\linewidth]{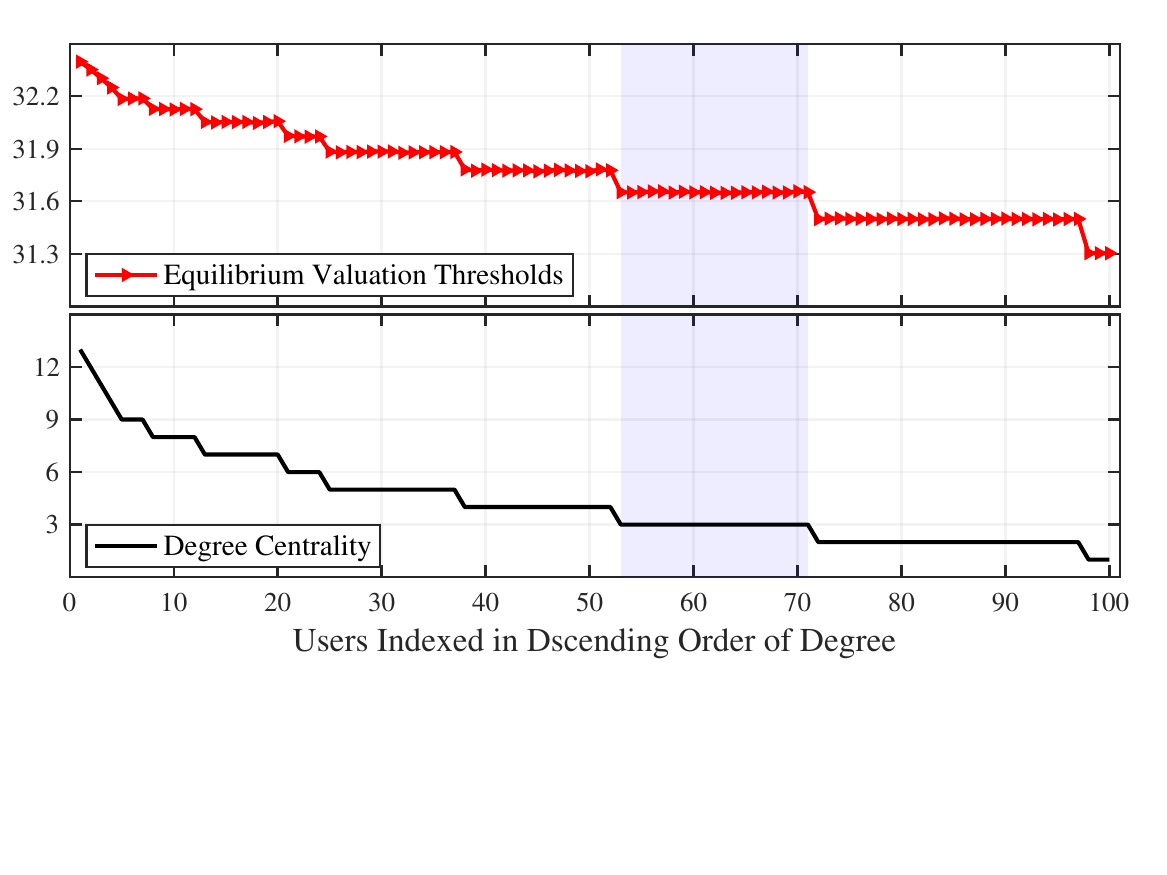}
    \vspace{-50pt}
    \caption{Equilibrium valuation thresholds and degree centrality for different users with a profiling accuracy $\delta=0.7$.}
    \label{fig:networ_thresholds}
\end{figure}

Finally, we delve into users' equilibrium valuation thresholds in more detail with Fig.~\ref{fig:networ_thresholds}. Specifically, we rank users in descending order of degree and plot their corresponding equilibrium valuation thresholds $v^*_i$ in Fig. \ref{fig:networ_thresholds}. Despite the heterogeneity in users' social network positions, our observation reveals a strong alignment between users' equilibrium valuation thresholds and their degree of centrality. For instance, within the blue regime of Fig. \ref{fig:networ_thresholds}, users with the same degree of $3$ exhibit similar valuation thresholds in the PBE. Therefore, our empirical findings suggest that \textit{the degree of each user node alone sufficiently serves as an indicator of users' social decision-making regarding information disclosure}. This result is surprising, as the degree is only one metric of each~user~node in the network: nodes with the same degree can have very~different centrality within the network.

\section{Robustness Beyond Uniform Valuation Distribution}\label{appendix:general}

For analytical tractability, we have assumed so far that users' valuations are uniformly distributed. In this section, we extend to more general distributions to demonstrate the robustness of our major results. Our analysis will be based on the three-stage model in Section \ref{Sec:Model}.

Consider that each user $i$'s valuation $v_i$ follows a more~general cumulative distribution function $F(v)$. However, such a dynamic Bayesian game is challenging to analyze in general. Nevertheless, we can still extend our analysis in closed form with the following assumption, which notably relaxes our prior assumption of uniformly-distributed users.

\begin{assumption}\label{ass:general_cdf}
The cumulative distribution function (CDF) of each user's valuation $F(v)$ satisfies the following conditions:
\begin{itemize}
    \item [(i)] $F(v)$ is concave over $[0,\bar{v}]$, and
    \item [(ii)] $p_0(1-F(p_0))$ is concave in $p_0$ over $[0,\bar{v}]$.
\end{itemize}
\end{assumption}

Assumption \ref{ass:general_cdf} presents a set of sufficient conditions to~ensure the existence and uniqueness of PBE beyond the uniform distribution. Typical examples that fall under these conditions include the uniform distribution, the exponential distribution, and certain beta distributions.\footnote{For instance, the CDF of beta distributions with shape parameters $\alpha=1$, $1<\beta<2$ or $1<\alpha<2$, $\beta=1$ is concave.} Particularly, as in the literature on price discrimination (e.g., \cite{conitzer2012hide,fudenberg2006behavior}), we also assume the concavity of the benchmark revenue function $p_0(1-F(p_0))$. This indicates that the seller's marginal revenue without user profiling decreases in the uniform price. This condition ensures the existence of a unique global maximum in the no-profiling benchmark case, which enables tractable analysis through the first-order condition. We hereafter denote the unique solution to $\arg\max_{p_0}p_0(1-F(p_0))$ as $\hat{p}_0$, i.e.,
\begin{equation*}
\hat{p}_0\triangleq\arg\max_{p_0}p_0(1-F(p_0)).
\end{equation*}

Similar to our analysis in Section~\ref{Sec:alternate}, we alternate backward induction with forward induction and characterize the PBE in the following proposition. Notice that our analysis here is based on Assumption \ref{ass:general_cdf}. Later in this section, we will further relax~these conditions in the numerical studies.

\begin{proposition}\label{Theorem:PBE_general}
Under Assumption \ref{ass:general_cdf}, there exists~a~unique PBE characterized as follows.
\begin{itemize}
\item \emph{\textbf{Case I (all active users in social networks)}} with a small maximum valuation
\begin{equation*}
\bar{v}-\hat{p}_0\le \frac{1}{\delta}\ln(n-1+\omega_0).
\end{equation*} 
\begin{itemize}
    \item [(i)] In Stage I, users' equilibrium valuation threshold is $v^\ast=\bar{v}$, i.e., every user $i\in\mathcal{N}$ is active in the social network with $x_i^*=1$.
    \item [(ii)] In Stage II, the seller sets a uniform price of $p_0^*=\hat{p}_0$.
\end{itemize}

\item \emph{\textbf{Case II (partially active users in social networks)}} with a large maximum valuation
\begin{equation*}
\bar{v}-\hat{p}_0> \frac{1}{\delta}\ln(n-1+\omega_0).
\end{equation*} 
\begin{itemize}
    \item [(i)] In Stage I, users' equilibrium valuation threshold is $v^*<\bar{v}$, i.e., there exist some high-valuation users who are inactive in the social network.
    \item [(ii)]  In Stage II, the seller sets a uniform price of 
    \begin{equation*}
        p_0^*=v^*-\tilde{J}(v^*)/\delta.
    \end{equation*}
\end{itemize}
Here, the users' equilibrium valuation threshold $v^*$ is the unique solution to 
\begin{equation*}
    1-\delta F(v^*)-(1-\delta)F(p_0^*)-(1-\delta)p_0^*f(p_0^*)=0.
\end{equation*}
\end{itemize}
\end{proposition}

Proposition \ref{Theorem:PBE_general} generalizes Theorem \ref{Theorem_PBE} beyond the uniform distribution yet retains similar major insights. As $\bar{v}$ increases from Case I to Case II, an average user's potential loss from being profiled becomes more significant. Thus, more users become inactive online to receive a low uniform price $p_0$ from the seller rather than a high personalized price~$v_i$. Furthermore, through numerical studies, we can also show that the non-monotonicity of users' social activity levels regarding~the~profiling accuracy $\delta$ still holds in the PBE, which is similar to Proposition~\ref{Prop:PBE_delta}.

We are also interested in whether such non-monotonicity~in users' equilibrium social behaviors carries through to more general valuation distributions. Below, we conduct an empirical study based on real-world consumer profiles sourced from an annually-updated consumer dataset in \cite{consumer}. Specifically, our analysis focuses on handbag purchases and the corresponding $92$ consumer profiles. To start, we utilize these empirical data to estimate the users' valuation distribution, which follows a normal distribution $\mathcal{N}(57.84,20.25^2)$ truncated over $[20,100]$ (see more details in Appendix \ref{Appendix:empirical:statistics}). Based on this estimated market information, we proceed to derive the equilibrium~behaviors of both (sampled) users and the seller in Fig. \ref{fig_expo}. Please refer to Appendix \ref{Appendix:Partially} for a complete version of this empirical study, where we delve into the PBE outcome in more detail and further examine the seller's empirical revenue gain.

\begin{figure}[h]
	\centering
    \vspace{-12pt}
	\subfigure[$F(v^\ast)$ versus $\delta$]{
		\begin{minipage}{0.45\linewidth}
			\centering
			\includegraphics[width=0.8\linewidth]{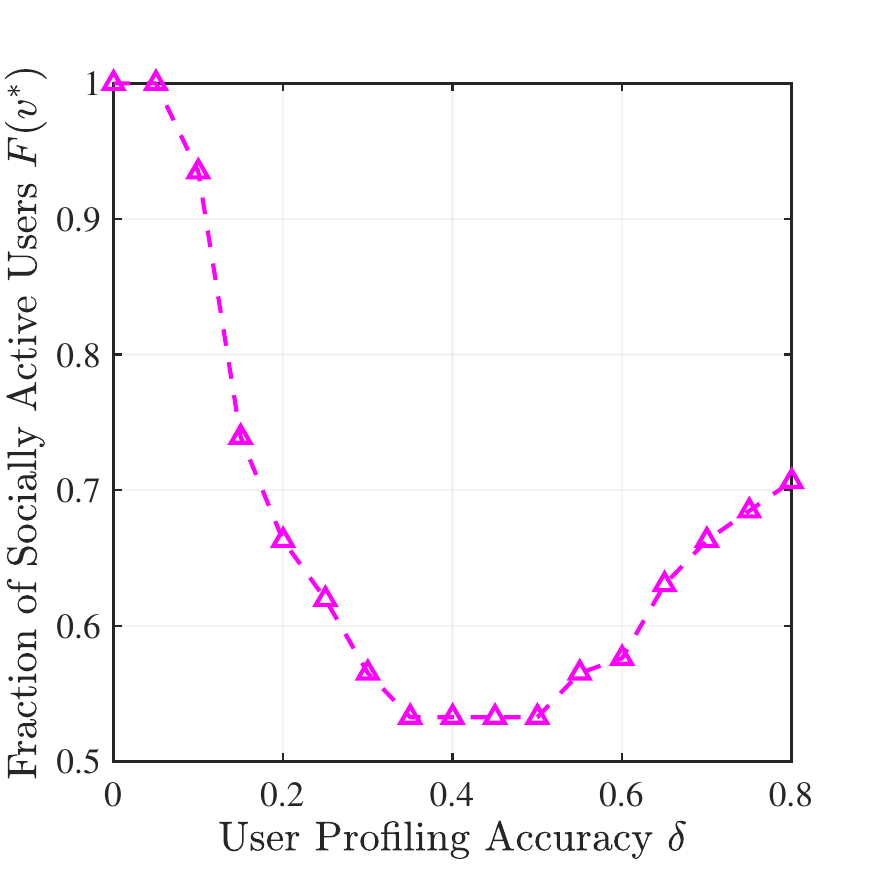}
			\label{fig_trunor_Fvstar}
		\end{minipage}
	}
	\subfigure[$p_0^*$ versus $\delta$]{
		\begin{minipage}{0.45\linewidth}
			\centering
			\includegraphics[width=0.8\linewidth]{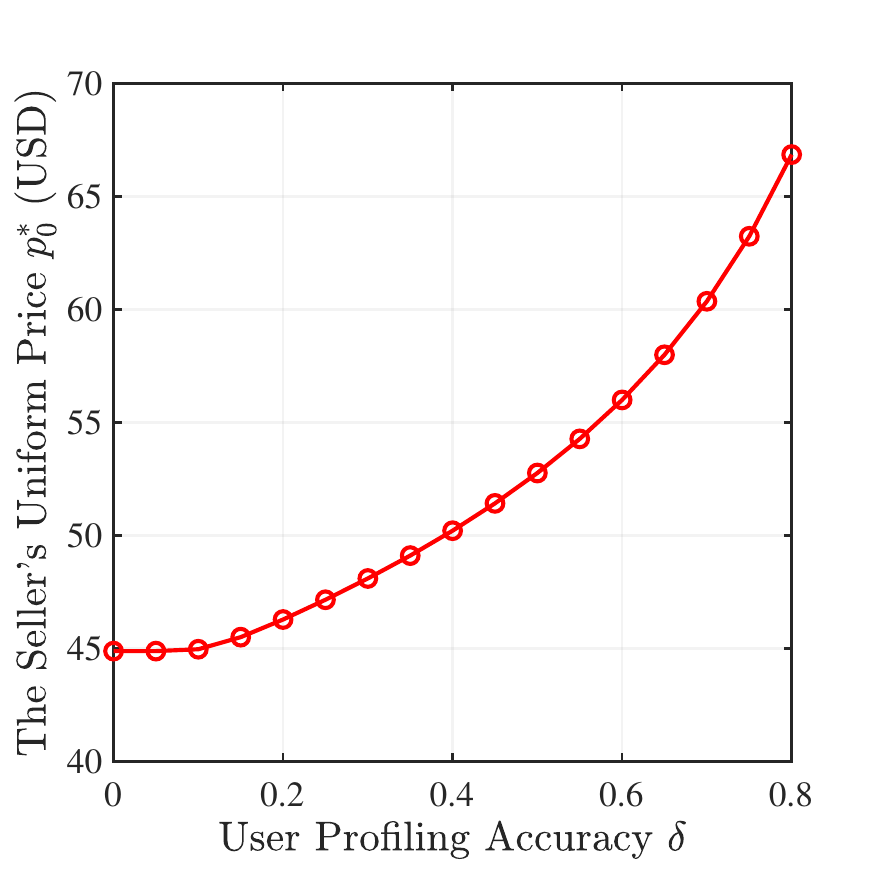}
			\label{fig_trunor_p0}
		\end{minipage}
	}
	\caption{The fraction of socially active users $F(v^\ast)$ and the seller's uniform price $p_0^*$ in the PBE (whenever it exists) based on real-world dataset \cite{consumer}.}
	\label{fig_expo}
    \vspace{-5pt}
\end{figure}

Fig. \ref{fig_expo} has a similar result to Fig. \ref{FIG:PBE_delta} and Proposition~\ref{Prop:PBE_delta}~under the uniform distribution. Specifically, the fraction of socially active users $F(v^\ast)$ in Fig. \ref{fig_trunor_Fvstar} first decreases and~then~increases in the user profiling accuracy $\delta$. Furthermore, the seller keeps raising the uniform price $p_0^*$ as $\delta$ increases, as shown in Fig.~\ref{fig_trunor_p0}, which motivates users to be active online. As a final remark, the main insights obtained from the analysis based on a uniform valuation distribution remain robust to more general distribution functions.

\section{Conclusion}\label{sec:conclusion}

This paper presents the first analytical study regarding~how users should best manage their social activities against potential personalized pricing, and how a seller should strategically adjust her pricing scheme to facilitate user profiling in social networks. To investigate these questions, we formulate~a~dynamic Bayesian game played between the seller and the users under asymmetric information. However, the~double couplings between the seller and the users, as well as amongst~the~users themselves, present significant challenges to our game analysis, where we also need to ensure consistency between users' revealed information and the seller's belief under random user profiling. Given the traditional backward induction method~no longer works, we alternate it with forward induction and~successfully characterize the unique PBE in closed form. 

Our analysis reveals that as the profiling technology accuracy improves, the seller tends to raise the equilibrium uniform price to motivate users' increased social activities and facilitate user profiling. However, this results in most users being worse off after the informed consent policy is implemented to ensure users' awareness of social data access and profiling practices by potential sellers. This finding suggests that recent regulatory evolution towards enhancing users' privacy awareness may have unintended consequences of reducing users' payoffs. In addition, we further examine some prevalent pricing practices where the seller breaks a pricing promise to personalize final offerings, which only marginally improves the seller's average revenue while introducing higher variance. Considering the additional reputation and regulation risks, it is not recommended for sellers to break pricing promises in practice.

For future work, we will investigate the case where the seller experiences costs when profiling users in the social~network. In particular, these costs may arise from data-sharing transactions and profiling technology development. It is interesting~to~understand how the seller manages the trade-off between profiling costs and the revenue from a larger pool of profiled users. Another future direction is to take into account the correlation among user profiles, where the data of one individual user may also be informative about his peers.




\section*{Acknowledgments}
This work was supported by the National Natural Science Foundation of China (Project $\text{62271434}$), Shenzhen Science and Technology Innovation Program (Project $\text{JCYJ20210324120011032}$), Guangdong Basic and Applied Basic Research Foundation (Project $\text{2021B1515120008}$), Shenzhen Key Lab of Crowd Intelligence Empowered Low-Carbon Energy Network (No. $\text{ZDSYS20220606100601002}$), Shenzhen Stability Science Program 2023, and the Shenzhen Institute of Artificial Intelligence and Robotics for Society. This work is also supported in part by the Ministry of Education, Singapore, under its Academic Research Fund Tier 2 Grant with Award no. MOE-T2EP20121-0001; in part by SUTD Kickstarter Initiative (SKI) Grant with no. SKI 2021\_04\_07; and in part by the Joint SMU-SUTD Grant with no. 22-LKCSB-SMU-053. 

\ifCLASSOPTIONcaptionsoff
  \newpage
\fi

\bibliographystyle{IEEEtran}
\bibliography{sample-base}

@article{buyukdaug2020effect,
  title={The effect of specific discount pattern in terms of price promotions on perceived price attractiveness and purchase intention: An experimental research},
  author={B{\"u}y{\"u}kda{\u{g}}, Naci and Soysal, Ay{\c{s}}e Nur and Kitapci, Olgun},
  journal={Journal of Retailing and Consumer Services},
  volume={55},
  pages={102112},
  year={2020},
  publisher={Elsevier}
}

@article{simonson2004anchoring,
  title={Anchoring effects on consumers' willingness-to-pay and willingness-to-accept},
  author={Simonson, Itamar and Drolet, Aimee},
  journal={Journal of consumer research},
  volume={31},
  number={3},
  pages={681--690},
  year={2004},
  publisher={The University of Chicago Press}
}

@misc{consumer,
  title = {Consumer Behavior and Shopping Habits Dataset:},
  url = {https://www.kaggle.com/datasets/zeesolver/consumer-behavior-and-shopping-habits-dataset/data},
  organization = {www.kaggle.com}
}

@misc{snapnets,
    author = {Jure Leskovec and Andrej Krevl},
    title = {{SNAP Datasets}: {Stanford} Large Network Dataset Collection},
    howpublished = {\url{http://snap.stanford.edu/data}},
    month= jun,
    year         = 2014
}

@article{mailath1993belief,
  title={Belief-based refinements in signalling games},
  author={Mailath, George J and Okuno-Fujiwara, Masahiro and Postlewaite, Andrew},
  journal={Journal of Economic Theory},
  volume={60},
  number={2},
  pages={241--276},
  year={1993},
  publisher={Elsevier}
}

@article{zhang2011perils,
  title={The perils of behavior-based personalization},
  author={Zhang, Juanjuan},
  journal={Marketing Science},
  volume={30},
  number={1},
  pages={170--186},
  year={2011},
  publisher={INFORMS}
}

@article{vives1990nash,
  title={Nash equilibrium with strategic complementarities},
  author={Vives, Xavier},
  journal={Journal of Mathematical Economics},
  volume={19},
  number={3},
  pages={305--321},
  year={1990},
  publisher={Elsevier}
}

@book{topkis1998supermodularity,
  title={Supermodularity and complementarity},
  author={Topkis, Donald M},
  year={1998},
  publisher={Princeton university press}
}

@article{baron2022revenue,
  title={Revenue volatility under uncertain network effects},
  author={Baron, Opher and Hu, Ming and Malekian, Azarakhsh},
  journal={Operations Research},
  volume={70},
  number={4},
  pages={2254--2263},
  year={2022},
  publisher={INFORMS}
}

@article{fudenberg1991perfect,
  title={Perfect Bayesian equilibrium and sequential equilibrium},
  author={Fudenberg, Drew and Tirole, Jean},
  journal={journal of Economic Theory},
  volume={53},
  number={2},
  pages={236--260},
  year={1991},
  publisher={Elsevier}
}

@misc{thecouncilofeconomicadvisers_2015_the,
  author = {The Council of Economic Advisers},
  title = {The Economics of Big Data and Differential Pricing},
  url = {https://obamawhitehouse.archives.gov/sites/default/files/whitehouse_files/docs/Big_Data_Report_Nonembargo_v2.pdf},
  year = {2015}
}

@misc{lee_2021_chinese,
  author = {Lee, Emma},
  month = {07},
  title = {Chinese regulators target price discrimination · TechNode},
  url = {https://technode.com/2021/07/15/regulators-target-price-discrimination/},
  year = {2021},
  organization = {TechNode}
}

@misc{boulton_2010_amazon,
  author = {Boulton, Clint},
  month = {07},
  title = {Amazon, Facebook Connect for Social Shopping},
  url = {https://www.eweek.com/it-management/amazon-facebook-connect-for-social-shopping/},
  urldate = {2022-08-22},
  year = {2010},
  organization = {eWEEK}
}

@misc{announcing,
  title = {Announcing Facebook Connect},
  url = {https://developers.facebook.com/blog/post/2008/05/09/announcing-facebook-connect/},
  organization = {Facebook for Developers}
}

@misc{bustos_2011_facebook,
  author = {Bustos, Linda},
  month = {02},
  title = {Facebook Connect for Ecommerce: Is It Right For Your Business?},
  url = {https://www.elasticpath.com/blog/facebook-connect-for-ecommerce-is-it-right-for-your-business},
  urldate = {2022-08-21},
  year = {2011},
  organization = {Elastic Path}
}

@misc{bustos_2011_7,
  author = {Bustos, Linda},
  month = {02},
  title = {7 Best Practices for Facebook Connect on Ecommerce Sites},
  url = {https://www.elasticpath.com/blog/facebook-connect-best-practices},
  urldate = {2022-08-21},
  year = {2011},
  organization = {Elastic Path}
}

@misc{anant_2020_consumer,
  author = {Anant, Venky and Donchak, Lisa and Kaplan, James and Soller, Henning},
  month = {04},
  title = {Consumer data protection and privacy | McKinsey},
  url = {https://www.mckinsey.com/business-functions/risk-and-resilience/our-insights/the-consumer-data-opportunity-and-the-privacy-imperative},
  year = {2020},
  organization = {www.mckinsey.com}
}

@article{2019Social,
  title={Social profiling: A review, taxonomy, and challenges},
  author={Bilal, Muhammad and Gani, Abdullah and Lali, Muhammad Ikram Ullah and Marjani, Mohsen and Malik, Nadia},
  journal={Cyberpsychology, Behavior, and Social Networking},
  volume={22},
  number={7},
  pages={433--450},
  year={2019},
  publisher={Mary Ann Liebert, Inc., publishers 140 Huguenot Street, 3rd Floor New~…}
}

@misc{howe_a,
  author = {Howe, Neil},
  title = {A Special Price Just for You},
  url = {https://www.forbes.com/sites/neilhowe/2017/11/17/a-special-price-just-for-you/?sh=2b797ee490b3},
  urldate = {2022-08-05},
  organization = {Forbes}
}

@misc{datadriven,
  title = {Data-Driven Personalization Drives Advanced Pricing Strategies},
  url = {https://www.retailtouchpoints.com/resources/data-driven-personalization-drives-advanced-pricing-strategies},
  urldate = {2022-08-05},
  organization = {Retail TouchPoints}
}

@misc{mattioli_2012, 
	title={On Orbitz, Mac Users Steered to Pricier Hotels},
	howpublished={\url{https://www.wsj.com/articles/SB10001424052702304458604577488822667325882}},
	year={2012}
}

@inproceedings{bimpikis2021data,
	title={Data Tracking under Competition},
	author={Bimpikis, Kostas and Morgenstern, Ilan and Saban, Daniela},
	booktitle={Proceedings of the 22nd ACM Conference on Economics and Computation},
	pages={137--137},
	year={2021}
}

@misc{x_2018, 
	title={As Facebook Raised a Privacy Wall, It Carved an Opening for Tech Giants}, url={https://www.nytimes.com/2018/12/18/technology/facebook-privacy.html}, journal={The New York Times},
	publisher={The New York Times}, 
	author={Dance, Gabriel J. X. and LaForgia, Michael and Confessore, Nicholas}, 
	year={2018}, 
	month={Dec}
}

@article{chen2020competitive,
  title={Competitive personalized pricing},
  author={Chen, Zhijun and Choe, Chongwoo and Matsushima, Noriaki},
  journal={Management Science},
  volume={66},
  number={9},
  pages={4003--4023},
  year={2020},
  publisher={INFORMS}
}

@article{2013private,
  title={Private traits and attributes are predictable from digital records of human behavior},
  author={Kosinski, Michal and Stillwell, David and Graepel, Thore},
  journal={Proceedings of the national academy of sciences},
  volume={110},
  number={15},
  pages={5802--5805},
  year={2013},
  publisher={National Acad Sciences}
}

@article{acquisti2005conditioning,
	title={Conditioning prices on purchase history},
	author={Acquisti, Alessandro and Varian, Hal R},
	journal={Marketing Science},
	volume={24},
	number={3},
	pages={367--381},
	year={2005},
	publisher={INFORMS}
}

@misc{a2018_california,
  month = {10},
  title = {California Consumer Privacy Act (CCPA)},
  url = {https://oag.ca.gov/privacy/ccpa},
  year = {2018},
  organization = {State of California - Department of Justice - Office of the Attorney General}
}

@misc{a2013_general,
  title = {General Data Protection Regulation (GDPR) },
  url = {https://gdpr-info.eu},
  year = {2013},
  organization = {General Data Protection Regulation (GDPR)}
}

@misc{daxtheduck_2019_new,
  author = {Dax the duck},
  month = {10},
  title = {New DuckDuckGo Research Shows People Taking Action on Privacy},
  url = {https://spreadprivacy.com/people-taking-action-on-privacy/},
  year = {2019},
  organization = {Spread Privacy}
}

@misc{cisco_2019_cisco,
  author = {Cisco},
  title = {Cisco 2021 Consumer Privacy Survey},
  url = {https://www.cisco.com/c/dam/en_us/about/doing_business/trust-center/docs/cisco-cybersecurity-series-2021-cps.pdf?CCID=cc000742&DTID=esootr000515&OID=rptsc027438},
  year = {2019}
}

@article{fudenberg2006behavior,
  title={Behavior-based price discrimination and customer recognition},
  author={Fudenberg, Drew and Villas-Boas, J Miguel},
  journal={Handbook on economics and information systems},
  volume={1},
  pages={377--436},
  year={2006},
  publisher={Citeseer}
}

@article{2017Is,
  title={Is Voluntary Profiling Welfare Enhancing?},
  author={Koh, Byungwan and Raghunathan, Srinivasan and Nault, Barrie R},
  journal={Management Information Systems Quarterly},
  volume={41},
  number={1},
  pages={23--41},
  year={2017}
}

@article{conitzer2012hide,
  title={Hide and seek: Costly consumer privacy in a market with repeat purchases},
  author={Conitzer, Vincent and Taylor, Curtis R and Wagman, Liad},
  journal={Marketing Science},
  volume={31},
  number={2},
  pages={277--292},
  year={2012},
  publisher={INFORMS}
}

@article{2015Monopoly,
  title={Monopoly price discrimination and privacy: The hidden cost of hiding},
  author={Belleflamme, Paul and Vergote, Wouter},
  journal={Economics Letters},
  volume={149},
  pages={141--144},
  year={2016},
  publisher={Elsevier}
}

@article{2020Consumer,
  title={Consumer profiling with data requirements: Structure and policy implications},
  author={Valletti, Tommaso and Wu, Jiahua},
  journal={Production and Operations Management},
  volume={29},
  number={2},
  pages={309--329},
  year={2020},
  publisher={Wiley Online Library}
}

@article{2006Metcalfe,
  title={Metcalfe's law is wrong-communications networks increase in value as they add members-but by how much?},
  author={Briscoe, Bob and Odlyzko, Andrew and Tilly, Benjamin},
  journal={IEEE Spectrum},
  volume={43},
  number={7},
  pages={34--39},
  year={2006},
  publisher={IEEE}
}

@article{2018collective,
  title={Collective aspects of privacy in the twitter social network},
  author={Garcia, David and Goel, Mansi and Agrawal, Amod Kant and Kumaraguru, Ponnurangam},
  journal={EPJ Data Science},
  volume={7},
  number={1},
  pages={3},
  year={2018},
  publisher={Springer}
}

@inproceedings{2020Voluntary,
  title={Voluntary disclosure and personalized pricing},
  author={Ali, S Nageeb and Lewis, Greg and Vasserman, Shoshana},
  booktitle={Proc. of the 21st ACM Conference on Economics and Computation},
  pages={537--538},
  year={2020}
}

@article{1990Rationalizability,
  title={Rationalizability, learning, and equilibrium in games with strategic complementarities},
  author={Milgrom, Paul and Roberts, John},
  journal={Econometrica: Journal of the Econometric Society},
  pages={1255--1277},
  year={1990},
  publisher={JSTOR}
}

@inproceedings{lin2021personalized,
	title={Personalized Pricing through User Profiling in Social Networks},
	author={Lin, Qinqi and Duan, Lingjie and Huang, Jianwei},
	booktitle={2021 19th International Symposium on Modeling and Optimization in Mobile, Ad hoc, and Wireless Networks (WiOpt)},
	pages={1--8},
	year={2021},
	organization={IEEE}
}

@book{fudenberg1991game,
  title={Game Theory},
  author={Fudenberg, Drew and Tirole, Jean},
  isbn={9780262061414},
  lccn={91002301},
  year={1991},
  publisher={MIT Press}
}

@article{battigalli2006rationalization,
  title={Rationalization in signaling games: Theory and applications},
  author={Battigalli, Pierpaolo},
  journal={International Game Theory Review},
  volume={8},
  number={01},
  pages={67--93},
  year={2006},
  publisher={World Scientific}
}

@article{seiter2015secret,
  title={The secret psychology of Facebook: Why we like, share, comment and keep coming back},
  author={Seiter, Courtney},
  journal={Date: April},
  volume={23},
  pages={2015},
  year={2015}
}

\clearpage
\onecolumn
\appendices

\setcounter{fact}{0}
\setcounter{lemma}{0}
\setcounter{claim}{0}
\setcounter{table}{0}
\setcounter{figure}{0}
\setcounter{theorem}{0}
\setcounter{equation}{0}
\setcounter{corollary}{0}
\setcounter{assumption}{0}
\setcounter{observation}{0}
\setcounter{proposition}{0}

\renewcommand{\thefact}{A\arabic{fact}}
\renewcommand{\theclaim}{A\arabic{claim}}
\renewcommand{\thelemma}{A\arabic{lemma}}
\renewcommand{\thetable}{A\arabic{table}}
\renewcommand{\thefigure}{A\arabic{figure}}
\renewcommand{\theremark}{A\arabic{remark}}
\renewcommand{\theequation}{A\arabic{equation}}
\renewcommand{\thecorollary}{A\arabic{corollary}}
\renewcommand{\theassumption}{A\arabic{assumption}}
\renewcommand{\theobservation}{A\arabic{observation}}
\renewcommand{\theproposition}{A\arabic{proposition}}


\section{Proof for Lemma \ref{Lem:noprofiling}}

In the absence of user profiling (i.e., $\delta=0$), users' social activities in Stage I and the seller's pricing strategy in Stage II are decoupled. Hence, each user $i$ decides his social activity level $x_i$ solely to maximize the social network benefit in Stage I. As \eqref{social_network_benefit} increases in $x_i$, we have $x_i^*(v_i)=1,\forall i\in\mathcal{N}$. As the seller cannot profile any user to exercise personalized pricing, she only considers the uniform price. Then, the seller sets the uniform price as
\begin{equation}
p_0^*=\arg\max p_0\left(1-F(p_0)\right)=\frac{\bar{v}}{2}
\end{equation}
to maximize her expected revenue extracted from non-profiled users alone (i.e., $\mathcal{N}_0=\mathcal{N}$).

\section{Proof for Proposition \ref{base_price}}

After the seller's profiling in Stage II, the set of profiled users $\mathcal{N}_1$ and the set of non-profiled users $\mathcal{N}_0$ are determined. Immediately, the seller's revenue extracted from the profiled users (see the second summation term in \eqref{platform_sale_revenue}) is also determined through personalized pricing (i.e., $p_i=v_i,\forall i\in\mathcal{N}_1$). Hence, the seller designs the uniform pricing scheme solely to maximize the revenue extracted from those non-profiled users. Notice that the true valuation of each non-profiled user is still unknown to the seller. Based on the seller's posterior belief of a non-profiled user's valuation $f(v_i|i\in\mathcal{N}_0)$ in \eqref{nonprofiled_belief}, given arbitrary users' valuation threshold $v^*$ in Stage I, we first calculate the seller's expected non-profiled revenue $\tilde{\Pi}_0$ below:
\begin{align}\label{Proposition2_revenue0}
\tilde{\Pi}_0
\triangleq \mathbb{E}_{v_i\sim f(\cdot|i\in\mathcal{N}_0)}\left[\sum_{i\in\mathcal{N}_0}p_0d_i\right]
&=|\mathcal{N}_{0}|\mathbb{E}_{v_i\sim f(\cdot|i\in\mathcal{N}_0)}\left[p_0\mathbbm{1}(v_i\ge p_0)\right],\nonumber\\
&=|\mathcal{N}_0|\int_{p_0}^{\bar{v}}p_0f(v_i|i\in\mathcal{N}_0)dv_i,\nonumber\\
&=\left\{
\begin{aligned}
&|\mathcal{N}_0|\frac{p_0(\bar{v}-p_0)}{\bar{v}-\delta v^*}, &\quad\textrm{if $p_0> v^*$,}\\
&|\mathcal{N}_0|p_0\left(1-\frac{1-\delta}{\bar{v}-\delta v^*}p_0\right),  &\quad\textrm{if $p_0\le v^*$;}
\end{aligned}
\right.
\end{align}
where we take expectations over all possible non-profiled users in $\mathcal{N}_0$. Given the first-order condition, we have $p_0^\dag\equiv\bar{v}/2$ and $p_0^\ddag\equiv\left(\bar{v}-\delta v^*\right)/2(1-\delta)$ to develop our discussions. Consider three cases in terms of $v^*$:
\begin{itemize}
	\item \mbox{Case 1 with $v^*\in[\bar{v}/(2-\delta),\bar{v}]$:} We have $p_0^\ddag\le v^*$ and $p_0^\dag<v^*$. Then \eqref{Proposition2_revenue0} increases with $p_0$ over $[0,p_0^\ddag]$ and decreases over $[p_0^\ddag,\bar{v}]$. Thus, we have the optimal uniform price $p_0^*=p_0^\ddag$.
	\item \mbox{Case 2 with $v^*\in[\bar{v}/2,\bar{v}/(2-\delta))$:} We have $p_0^\ddag> v^*$ and $p_0^\dag\le v^*$. Then \eqref{Proposition2_revenue0} increases with $p_0$ over $[0,v^*]$ and decreases over $[v^*,\bar{v}]$. Thus, we have the optimal uniform price $p_0^*=v^*$.
	\item \mbox{Case 3 with $v^*\in[0,\bar{v}/2)$:} We have $p_0^\ddag> v^*$ and $p_0^\dag>v^*$. Then \eqref{Proposition2_revenue0} increases with $p_0$ over $[0,\bar{v}/2]$ and decreases over $[\bar{v}/2,\bar{v}]$. Thus, we have the optimal uniform price $p_0^*=\bar{v}/2$.
\end{itemize}
We thus conclude with the seller's optimal uniform pricing scheme in \eqref{base_price_equation1}-\eqref{base_price_equation3}.

\section{Proof for Belief \ref{threshold_structure}}\label{proofbelief}

The proof follows the three steps stated in Section \ref{forward}.

\subsection{Step I: Game Formulation and Solution Existence}

We first formally formulate users' social interactions in Stage I as a static Bayesian game, where users make social decisions while predicting the seller's uniform price $p_0$ in Stage II.

\begin{game}[Social Interaction Game]
The users' social interaction game $\Gamma=(\mathcal{N},\{\mathcal{X}_i\}_{i\in \mathcal{N}},\{\mathcal{V}_i\}_{i\in \mathcal{N}},\mathcal{F},\{\tilde{\pi}_i\}_{i\in \mathcal{N}})$ is formulated as a static Bayesian game as follows.
\begin{itemize}
\item \emph{Players:} Users in the set $\mathcal{N}\triangleq\{1,\dots,n\}$.
\item \emph{Actions:} Each user $i\in\mathcal{N}$'s action is his social activity level $x_i\in\mathcal{X}_i\triangleq [0,1]$.
\item \emph{Types:} Each user $i\in\mathcal{N}$'s type is his valuation $v_i\in\mathcal{V}_i\triangleq[0,\bar{v}]$.
\item \emph{Common Prior: } Users' valuations are drawn from a prior distribution $\mathcal{F}(v_1,\dots,v_n)=\prod_{i=1}^{n}F(v_i)$, where $F$ is the i.i.d uniform distribution each user's valuation follows.
\item \emph{Utilities:} Each user $i\in\mathcal{N}$'s utility is $\tilde{\pi}_i$ in \eqref{expectpayoff}.
\end{itemize}
\end{game}

Next, we prove that the social interaction game defined above is a supermodular game \cite{1990Rationalizability,topkis1998supermodularity}.

\begin{proposition}\label{supermodular_game}
The users' social interaction game $\Gamma$ (under incomplete information) is a supermodular game.
\end{proposition}

\begin{proof}
Any Bayesian game is itself a supermodular game when its underlying certainty game (for any valuation realization) is supermodular (see Theorem 6.1 in \cite{vives1990nash}); thus we will next show the supermodularity of the underlying certainty game for $\Gamma$ to complete the proof. According to the definition of the supermodular game (see Theorem 4 in \cite{1990Rationalizability}), proofs can be done by verifying that (i) the action set $\mathcal{X}_i\triangleq [0,1]\in\mathbb{R}$ for each user $i$ is a compact subset of real number and thus a complete lattice, (ii) the utility function $\tilde{\pi}_i(x_i,\boldsymbol{x}_{-i})$ is supermodular in $x_i$ (for fixed $\boldsymbol{x}_{-i}$) since $x_i$ is a single variable instead of multidimensional vector, and (iii) $\tilde{\pi}_i(x_i,\boldsymbol{x}_{-i})$ is twice continuously differentiable in $x_i$ over $[0,1]$ and has increasing differences in $x_i$ and $x_j$, $\forall j\neq i$, i.e.,
\begin{align}
    \frac{\partial^2 \tilde{\pi}_i(x_i,\boldsymbol{x}_{-i})}{\partial x_i\partial x_j}=\frac{1}{\sum_{j\neq i}x_j+\omega_0}>0.
\end{align}
\end{proof}

Proposition \ref{supermodular_game} then helps establish the existence of a pure-strategy Bayesian Nash equilibrium (BNE) in the social interaction game, which is obtained through the properties of the supermodular game.

\begin{corollary}
A pure-strategy Bayesian Nash equilibrium exists in the social interaction game $\Gamma$.
\end{corollary}
\begin{proof}
This follows from Proposition \ref{supermodular_game} based on the properties of the supermodular game (see Theorem 5 in \cite{1990Rationalizability}).
\end{proof}

Even though we have confirmed the existence of BNE in the users' social interaction game, its specific structure still needs to be determined. In the next two steps, we will focus on figuring out the structural results of users' social decisions in Stage~I.

\subsection{Step II: Polarization in Equilibrium Social Decisions}\label{StepII_belief}

Next, we demonstrate the polarization in the users' equilibrium social activity levels. More specifically, we reduce the original continuous action space $[0,1]$ to binary choices, i.e., either $0$ or $1$, for each individual user. Let $x_i(\cdot)$ denote the strategy for user $i$, and $\boldsymbol{x}_{-i}(\cdot)\triangleq\left(x_1(\cdot),\dots, x_{i-1}(\cdot),x_{i+1}(\cdot),\dots,x_{n}(\cdot)\right)$ denote the strategy profile for users other than $i$. Also, we denote the pure BNE strategy profiles as $\left(x_i^*(\cdot);i\in\mathcal{N}\right)$.

To start with, we characterize each individual user $i$'s expected payoff with incomplete information. Given that user $i$'s type is $v_i$ and the strategy profile of the other users is $\boldsymbol{x}_{-i}(\cdot)$, this user $i$'s expected payoff from choosing the social activity level $x_i=x_i(v_i)$ is
\begin{align}\label{incomplete_payoff}
    \tilde{U}_i(x_i;\boldsymbol{x}_{-i}(\cdot),v_i)&\triangleq\int \tilde{\pi}_i(x_i,\boldsymbol{x}_{-i}(\boldsymbol{v}_{-i})) d\boldsymbol{F}_{-i}(\boldsymbol{v}_{-i}),\nonumber\\&=\int x_i\ln\left(\sum_{j\neq i}x_j(v_j)+\omega_0\right) d\boldsymbol{F}_{-i}(\boldsymbol{v}_{-i})+(1-\delta x_i^\alpha)\max\{v_i-p_0,0\};
\end{align}
where we take expectations over all possible valuations of users other than $i$ given incomplete information, and $\tilde{\pi}_i(x_i,\boldsymbol{x}_{-i}(\boldsymbol{v}_{-i}))$ is given in \eqref{expectpayoff}. Then we proceed with analyzing the best response correspondence for each individual user. Specifically, in the BNE, given all other users' strategy profile $\boldsymbol{x}^*_{-i}(\cdot)$, each individual user $i$ decides his social activity level $x_i^*(v_i)$ to maximize his expected payoff in \eqref{incomplete_payoff}, i.e.,
\begin{align}\label{argmax_incompletepayoff}
x_i^*(v_i)\in &\arg\max_{x_i\in[0,1]}\tilde{U}_i(x_i;\boldsymbol{x}^*_{-i}(\cdot),v_i),\nonumber\\
=&\arg\max_{x_i\in[0,1]}\int x_i\ln\left(\sum_{j\neq i}x_j^*(v_j)+\omega_0\right) d\boldsymbol{F}_{-i}(\boldsymbol{v}_{-i})+(1-\delta x_i^\alpha)\max\{v_i-p_0,0\}.
\end{align}
Furthermore, since the expected payoff $\tilde{U}_i(x_i;\boldsymbol{x}^*_{-i}(\cdot),v_i)$ is convex in $x_i$ over the closed interval of $[0,1]$ (see Lemma \ref{convex_payoff} at the end of this subsection), we then have $x_i^*(v_i)\in\{0,1\}$, i.e., each user $i$'s equilibrium social activity level is either $0$ or $1$. More specifically, we have
\begin{align}\label{a5}
x_i^*(v_i)=&\mathbbm{1}\left(\tilde{U}_i(1;\boldsymbol{x}^*_{-i}(\cdot),v_i)\ge\tilde{U}_i(0;\boldsymbol{x}^*_{-i}(\cdot),v_i)\right),\nonumber\\=&\mathbbm{1}\left(v_i\le p_0+\frac{1}{\delta}\int \ln\left(\sum_{j\neq i}x_j^*(v_j)+\omega_0\right) d\boldsymbol{F}_{-i}(\boldsymbol{v}_{-i})\right).    
\end{align}
However, it is still difficult to directly outline the structure of users' social activity decisions from \eqref{a5} because of the coupling therein. Given such, we seek and prove the threshold structure in the next step.

\begin{lemma}\label{convex_payoff}
$\tilde{U}_i(x_i;\boldsymbol{x}^*_{-i}(\cdot),v_i)$ in \eqref{incomplete_payoff} is convex in $x_i$ over $[0,1]$.
\end{lemma}
\begin{proof}
Calculating the second-order derivative of $\tilde{\pi}_i(x_i,\boldsymbol{x}_{-i})$ regarding $x_i$ yields
\begin{align}
    \frac{\partial^2 \tilde{\pi}_i(x_i,\boldsymbol{x}_{-i})}{\partial x_i^2}=-\delta\alpha(\alpha-1)x_i^{\alpha-2}\cdot\max\{v_i-p_0,0\}\ge0
\end{align}
with $0<\alpha<1$. So $\tilde{\pi}_i(x_i,\boldsymbol{x}_{-i})$ is convex in $x_i$ over $[0,1]$. The convexity of $\tilde{U}_i(x_i;\boldsymbol{x}^*_{-i}(\cdot),v_i)$ thus follows because convexity is preserved under integrals in \eqref{incomplete_payoff}.
\end{proof}

\subsection{Step III: Existence of Common Valuation Threshold}

Finally, we prove that there exists a valuation threshold for each user $i$'s social activity strategy $x_i^*(v_i)$ in the BNE. 

\begin{lemma}\label{threshold_i}
In the BNE of the social interaction game $\Gamma$, each user $i$'s social activity strategy $x_i^*(v_i)$ is non-increasing in his valuation $v_i$. In particular, there exists a equilibrium valuation threshold $v_i^*\in[p_0,\bar{v}]$ such that
\begin{equation}
x_i^*(v_i) = \mathbbm{1}(v_i\le v^*_i),\quad\forall i\in\mathcal{N}.
\end{equation}
\end{lemma}

\begin{proof}
Based on each individual user $i$'s expected payoff $\tilde{U}_i$ in \eqref{incomplete_payoff} given incomplete information, this user's valuation $v_i$ and the equilibrium strategy profile of the other users $\boldsymbol{x}^*_{-i}$, we consider the following two possible cases concerning $v_i$. Hereafter, we denote $x_i(v_i)\equiv x_i$ for notational simplicity.
\begin{itemize}
\item [(1)] If $v_i\le p_0$, we have
\begin{align}
\tilde{U}_i(x_i;\boldsymbol{x}^*_{-i},v_i)=\int x_i\ln\left(\sum_{j\neq i}x_j^*(v_j)+\omega_0\right) d\boldsymbol{F}_{-i}(\boldsymbol{v}_{-i}),
\end{align}
which increases with $x_i$. Therefore, user $i$'s equilibrium strategy is $x_i^*(v_i)=1$ if $v_i\le p_0$.

\item [(2)] If $v_i>p_0$, we have 
\begin{align}
    \tilde{U}_i(x_i;\boldsymbol{x}^*_{-i},v_i)=\int x_i\ln\left(\sum_{j\neq i}x_j^*(v_j)+\omega_0\right) d\boldsymbol{F}_{-i}(\boldsymbol{v}_{-i})+(1-\delta x_i^\alpha)(v_i-p_0).
\end{align}
For user $i$ with a valuation $v'_i$, the equilibrium strategy $x_i^*(\cdot)$ requires that this user prefers $x_i^*(v'_i)\equiv x'_i$ to $x_i^*(v''_i)\equiv x''_i$. Thus we have

\begin{multline}\label{a6}
\int x'_i\ln\left(\sum_{j\neq i}x_j^*(v_j)+\omega_0\right) d\boldsymbol{F}_{-i}(\boldsymbol{v}_{-i})+(1-\delta \left(x'_i\right)^\alpha)(v'_i-p_0)\ge\\\int x''_i\ln\left(\sum_{j\neq i}x_j^*(v_j)+\omega_0\right) d\boldsymbol{F}_{-i}(\boldsymbol{v}_{-i})+(1-\delta \left(x''_i\right)^\alpha)(v'_i-p_0),
\end{multline}

and we also have, for user $i$ with a valuation $v''_i$,

\begin{multline}\label{a7}
\int x''_i\ln\left(\sum_{j\neq i}x_j^*(v_j)+\omega_0\right) d\boldsymbol{F}_{-i}(\boldsymbol{v}_{-i})+(1-\delta \left(x''_i\right)^\alpha)(v''_i-p_0)\ge\\\int x'_i\ln\left(\sum_{j\neq i}x_j^*(v_j)+\omega_0\right) d\boldsymbol{F}_{-i}(\boldsymbol{v}_{-i})+(1-\delta \left(x'_i\right)^\alpha)(v''_i-p_0).
\end{multline}

Subtracting the right-hand side of \eqref{a7} from the left-hand side of \eqref{a6}, and subtracting the left-hand side of \eqref{a7} from the right-hand side of \eqref{a6}, yields $\delta({x''_i}^\alpha-{x'_i}^\alpha)(v'_i-v''_i)\ge0$. So $x_i''\ge x_i'$ if $v'_i\ge v_i''>p_0$, i.e., user $i$'s equilibrium strategy $x_i^*(v_i)$ is non-increasing in his valuation $v_i$ if $v_i>p_0$. 
\end{itemize}
In conclusion, each user $i$'s equilibrium strategy $x_i^*(v_i)$ is non-increasing in his valuation $v_i$. Furthermore, since the equilibrium social activity level for each user $i$ is either $x_i^*=0$ or $x_i^*=1$, there exists a valuation threshold $v_i^*\in[p_0,\bar{v}]$ such that $x_i^*(v_i)=1$ if $v_i\le v^*_i$ whereas $x_i^*(v_i)=0$ if $v_i> v^*_i$.
\end{proof}

Moreover, we show that such an equilibrium valuation threshold is the same for all users independent of their valuations.

\begin{lemma}\label{threshold_same}
$v^*_i=v^*,\forall i\in\mathcal{N}$.
\end{lemma}

\begin{proof}
The proof is developed by contradiction. According to \eqref{a5} and Lemma~\ref{threshold_i}, the equilibrium valuation threshold $v_i^*$ for user $i$ satisfies the following equality:

\begin{align}\label{a13}
v_i^*= p_0+\frac{1}{\delta}&\int \ln\left(\sum_{j\neq i}\mathbbm{1}(v_j\le v^*_j)+\omega_0\right) d\boldsymbol{F}_{-i}(\boldsymbol{v}_{-i}),\nonumber\\
=p_0+\frac{1}{\delta}&\int \ln\left(\sum_{k\neq i,j}\mathbbm{1}(v_k\le v^*_k)+\mathbbm{1}(v_j\le v^*_j)+\omega_0\right) d\boldsymbol{F}_{-i}\left(\boldsymbol{v}_{-i}\right),\nonumber\\
=p_0+\frac{1}{\delta}&\int\ln\left(\sum_{k\neq i,j}\mathbbm{1}(v_k\le v^*_k)+1+\omega_0\right) F(v^*_j)\cdot d\boldsymbol{F}_{-i,j}\left(v_{-i,j}\right)\nonumber\\&+\frac{1}{\delta}\int\ln\left(\sum_{k\neq i,j}\mathbbm{1}(v_k\le v^*_k)+\omega_0\right)\left(1-F(v^*_j)\right)d\boldsymbol{F}_{-i,j}\left(v_{-i,j}\right),
\end{align}
where we denote for ease of exposition
\begin{align*}
    A\triangleq\frac{1}{\delta}\int\ln\left(\sum_{k\neq i,j}\mathbbm{1}(v_k\le v^*_k)+1+\omega_0\right)d\boldsymbol{F}_{-i,j}\left(v_{-i,j}\right)
\end{align*}
and
\begin{align*}
B \triangleq\frac{1}{\delta}\int\ln\left(\sum_{k\neq i,j}\mathbbm{1}(v_k\le v^*_k)+\omega_0\right)d\boldsymbol{F}_{-i,j}\left(v_{-i,j}\right).
\end{align*}
Notice that $A>B$, and both $A$ and $B$ are independent of $v_i^*$ and $v_j^*$. Suppose by contradiction that $v^*_i\neq v^*_j$ for user $i$ and user $j$, and let $v^*_i> v^*_j$ without loss of generality. It then follows that $F(v_i^*)\ge F(v_j^*)$, and we also have from \eqref{a13} that
\begin{align}\label{a9}
v^*_i=p_0+A\cdot F(v^*_j)+B(1-F(v^*_j)) \text{ and } v^*_j=&p_0+A\cdot F(v^*_i)+B(1-F(v^*_i)).
\end{align}
We hence have $v_i^*-v_j^*=(A-B)\left(F(v^*_j)-F(v^*_i)\right)\le 0$ given that $A>B$ and $F(v_i^*)\ge F(v_j^*)$. This results in contradiction with $v^*_i>v^*_j$. Therefore, we have $v^*_i=v^*,\forall i\in\mathcal{N}$.
\end{proof}
Combining Lemma \ref{threshold_i} and Lemma \ref{threshold_same} establishes $x_i^*(v_i) = \mathbbm{1}(v_i\le v^*)$, $\forall i\in\mathcal{N}$. This completes the proof for Belief \ref{threshold_structure}.

\section{Proof for Theorem \ref{Theorem:p0-vstar}}
Recall that there are only two possible social activity levels in the equilibrium for each individual user $i$: $x_i=1$ (active) and $x_i=0$ (inactive) (see Appendix \ref{StepII_belief}). With a valuation exactly the threshold value $v_i=v^*$, the marginal user $i$ is indifferent between these two social decisions, i.e., $\tilde{U}^*_i(1;\boldsymbol{x}^*_{-i},v^*)=\tilde{U}^*_i(0;\boldsymbol{x}^*_{-i},v^*)$ (see the expected user payoff $\tilde{U}^*_i(x^*_i;\boldsymbol{x}^*_{-i},v_i)$ in \eqref{incomplete_payoff}). More specifically, when $x_i^*=0$, this marginal user only receives the purchase surplus, i.e., $\tilde{U}^*_i(0;\boldsymbol{x}^*_{-i},v^*)=\max\{v^*-p_0,0\}=v^*-p_0$ (see Lemma \ref{threshold_i}, $v^*\ge p_0$). When $x_i^*=1$, this marginal user gains an expected payoff of
\begin{align}
\tilde{U}_i^*(1;\boldsymbol{x}^*_{-i},v^*)&=\int \ln\left(\sum_{j\neq i}\mathbbm{1}(v_j\le v^*)+\omega_0\right) d\boldsymbol{F}_{-i}(\boldsymbol{v}_{-i})+(1-\delta)\max\{v^*-p_0,0\},\nonumber\\
&=\sum_{m=0}^{n-1}\tbinom{n-1}{m}\ln(m+\omega_0)F(v^*)^m\left(1-F(v^*)\right)^{n-1-m}+(1-\delta)(v^*-p_0).
\end{align}
The indifference condition $\tilde{U}^*_i(1;\boldsymbol{x}^*_{-i},v^*)=\tilde{U}^*_i(0;\boldsymbol{x}^*_{-i},v^*)$ then suggests \eqref{vstar_p0_summation} in Theorem \ref{Theorem:p0-vstar}, where we actually have $v^*>p_0$. Notice that the equilibrium valuation threshold $v^*$ should be no larger than $\bar{v}$. This then completes the proof for Theorem \ref{Theorem:p0-vstar}.

\section{Proof for Proposition \ref{unique_user}}\label{proof_v^*}

Recall that we have established and characterized the valuation threshold structure in the BNE of the social interaction game (see Belief \ref{threshold_structure}). Specifically, the users' equilibrium valuation threshold is $v^*=\min\{v^\dagger,\bar{v}\}$, where $v^\dagger$ is the solution to \eqref{vstar_p0_summation} (see Theorem \ref{Theorem:p0-vstar}). Notice that we will fix the seller's uniform price $p_0$ in this proof, since users' social interaction game is defined in Stage I while predicting an arbitrary $p_0$ in Stage II. To show the uniqueness of BNE, we only need to prove the uniqueness of $v^*$. To this end, we need to analyze the solution $v^\dagger$ to \eqref{vstar_p0_summation}, rewritten below:
\begin{equation}
    v^\dagger =p_0+\frac{1}{\delta}\tilde{J}(v^\dagger).
\end{equation}
For ease of exposition, henceforth let
\begin{equation}\label{Phi}
    \Phi(v^\dagger)\triangleq p_0+\frac{1}{\delta}\tilde{J}(v^\dagger)-v^\dagger.
\end{equation}
Namely, we need to analyze the solution $v^\dagger$ to $\Phi(v^\dagger)=0$ and its uniqueness in particular. Notice that when $v^\dagger>\bar{v}$, the CDF function value is $F(v^\dagger)\equiv 1$; then, $\Phi(v^\dagger)$ simply becomes a linearly decreasing function of $v^\dagger$:
\begin{equation}\label{Phi_linear}
    \Phi(v^\dagger) = p_0+\frac{1}{\delta}\ln(n-1+\omega_0)-v^\dagger, \text{ if } v^\dagger>\bar{v}.
\end{equation}
The more involved case when $v^\dagger\le\bar{v}$ becomes our focus in what follows. To tackle it, we present some preliminary results.

To start with, we first establish the concavity of the expected social network benefit function $\tilde{J}(v^\dagger)$ defined in \eqref{expected_social_network}.

\begin{landscape}
\begin{lemma}\label{claim:J_concave}
$\tilde{J}(v^\dagger)$ concavely increases in $v^\dagger$ over $[0,\bar{v}]$, i.e., $\partial \tilde{J}(v^\dagger)/\partial v^\dagger>0$ and $\partial ^2\tilde{J}(v^\dagger)/\partial (v^\dagger)^2<0$. 
\end{lemma}
\begin{proof}
We first take the first-order derivative of $\tilde{J}(v^\dagger)$ in \eqref{FO_J} below.
\begin{align}
\frac{\partial \tilde{J}(v^\dagger)}{\partial v^\dagger}
&=\sum_{m=0}^{n-1}\tbinom{n-1}{m}\ln(m+\omega_0)\left(mF(v^\dagger)^{m-1}\left(1-F(v^\dagger)\right)^{n-1-m}-(n-1-m)F(v^\dagger)^m\left(1-F(v^\dagger)\right)^{n-2-m}\right)f(v^\dagger),\nonumber\\
&=\sum_{m=0}^{n-1}\tbinom{n-1}{m}\ln(m+\omega_0)mF(v^\dagger)^{m-1}\left(1-F(v^\dagger)\right)^{n-1-m}f(v^\dagger)-\sum_{m=0}^{n-1}\tbinom{n-1}{m}\ln(m+\omega_0)(n-1-m)F(v^\dagger)^m\left(1-F(v^\dagger)\right)^{n-2-m}f(v^\dagger),\nonumber\\
&=\sum_{m=1}^{n-1}\tbinom{n-1}{m}\ln(m+\omega_0)mF(v^\dagger)^{m-1}\left(1-F(v^\dagger)\right)^{n-1-m}f(v^\dagger)-\sum_{m=0}^{n-2}\tbinom{n-1}{m}\ln(m+\omega_0)(n-1-m)F(v^\dagger)^m\left(1-F(v^\dagger)\right)^{n-2-m}f(v^\dagger),\nonumber\\
&=\sum_{m=1}^{n-1}\frac{(n-1)!}{(n-1-m)!(m-1)!}\ln(m+\omega_0)F(v^\dagger)^{m-1}\left(1-F(v^\dagger)\right)^{n-1-m}f(v^\dagger)\nonumber\\&\quad-\sum_{m=0}^{n-2}\frac{(n-1)!}{(n-2-m)!m!}\ln(m+\omega_0)F(v^\dagger)^m\left(1-F(v^\dagger)\right)^{n-2-m}f(v^\dagger),\nonumber\\
&=\sum_{m=1}^{n-1}\frac{(n-1)!}{(n-1-m)!(m-1)!}\ln(m+\omega_0)F(v^\dagger)^{m-1}\left(1-F(v^\dagger)\right)^{n-1-m}f(v^\dagger)\tag{Line 5-1}\\&\quad-\sum_{m=1}^{n-1}\frac{(n-1)!}{(n-1-m)!(m-1)!}\ln\left(m-1+\omega_0\right)F(v^\dagger)^{m-1}\left(1-F(v^\dagger)\right)^{n-1-m}f(v^\dagger),\tag{Line 5-2}\\	
&=\sum_{m=1}^{n-1}\frac{(n-1)!}{(n-1-m)!(m-1)!}\left(\ln(m+\omega_0)-\ln(m-1+\omega_0)\right)F(v^\dagger)^{m-1}\left(1-F(v^\dagger)\right)^{n-1-m}f(v^\dagger).\label{FO_J}
\end{align}

One key step is to transform the summation index from $m$ to $m-1$ in Line 5-2. As such, \eqref{FO_J} directly tells that $\partial \tilde{J}(v^\dagger)/\partial v^\dagger\equiv \tilde{J}'(v^\dagger)>0$, i.e., $\tilde{J}(v^\dagger)$ strictly increases in $v^\dagger$. Then the second-order derivative of $\tilde{J}
(v^\dagger)$ follows in \eqref{SO_J} below. 

\begin{align}
\frac{\partial^2 \tilde{J}(v^\dagger)}{\partial (v^\dagger)^2}	
&=\sum_{m=1}^{n-1}\frac{(n-1)!}{(n-1-m)!(m-1)!}\left(\ln(m+\omega_0)-\ln(m-1+\omega_0)\right)(m-1)F(v^\dagger)^{m-2}\left(1-F(v^\dagger)\right)^{n-1-m}\left(f(v^\dagger)\right)^2\tag{Line 1-1}\\&\quad-\sum_{m=1}^{n-1}\frac{(n-1)!}{(n-1-m)!(m-1)!}\left(\ln(m+\omega_0)-\ln(m-1+\omega_0)\right)(n-1-m)F(v^\dagger)^{m-1}\left(1-F(v^\dagger)\right)^{n-2-m}\left(f(v^\dagger)\right)^2,\tag{Line 1-2}\\
&=\sum_{m=2}^{n-1}\frac{(n-1)!}{(n-1-m)!(m-2)!}\left(\ln(m+\omega_0)-\ln(m-1+\omega_0)\right)F(v^\dagger)^{m-2}\left(1-F(v^\dagger)\right)^{n-1-m}\left(f(v^\dagger)\right)^2\nonumber\\&\quad-\sum_{m=1}^{n-2}\frac{(n-1)!}{(n-2-m)!(m-1)!}\left(\ln(m+\omega_0)-\ln(m-1+\omega_0)\right)F(v^\dagger)^{m-1}\left(1-F(v^\dagger)\right)^{n-2-m}\left(f(v^\dagger)\right)^2,\nonumber\\
&=\sum_{m=2}^{n-1}\frac{(n-1)!}{(n-1-m)!(m-2)!}\left(\ln(m+\omega_0)-\ln(m-1+\omega_0)\right)F(v^\dagger)^{m-2}\left(1-F(v^\dagger)\right)^{n-1-m}\left(f(v^\dagger)\right)^2\tag{Line 3-1}\\&\quad-\sum_{m=2}^{n-1}\frac{(n-1)!}{(n-1-m)!(m-2)!}\left(\ln(m-1+\omega_0)-\ln(m-2+\omega_0)\right)F(v^\dagger)^{m-2}\left(1-F(v^\dagger)\right)^{n-1-m}\left(f(v^\dagger)\right)^2,\tag{Line 3-2}\\
&=\sum_{m=2}^{n-1}\frac{(n-1)!}{(n-1-m)!(m-2)!}F(v^\dagger)^{m-2}\left(1-F(v^\dagger)\right)^{n-1-m}\left(f(v^\dagger)\right)^2\left(\ln(m+\omega_0)-2\ln(m-1+\omega_0)+\ln(m-2+\omega_0)\right).\label{SO_J}
\end{align}

Notice that $\partial f(v^\dagger)/\partial v^\dagger = 0$ given $f(v^\dagger) = 1/\bar{v}$ for the uniform valuation distribution. This fact is applied in Line 1. Similar to calculating \eqref{FO_J}, the key step here to derive the second-order derivative is Line 3-2, where we transform the summation index from $m$ to $m-1$. Furthermore, one can easily verify that
\begin{align*}
\ln(m+\omega_0)-2\ln(m-1+\omega_0)+\ln(m-2+\omega_0)\nonumber
=\left(\ln(m+\omega_0)-\ln(m-1+\omega_0)\right)-\left(\ln(m-1+\omega_0)-\ln(m-2+\omega_0)\right)<0;
\end{align*}
thus, we conclude with $\partial ^2\tilde{J}(v^\dagger)/\partial (v^\dagger)^2<0$, i.e., $\tilde{J}(v^\dagger)$ is concave in $v^\dagger$.
\end{proof}
\end{landscape}

Next, we provide facts regarding concave functions to facilitate our solution analysis for the equation $\Phi(v^\dagger)=0$ as well as other analyses in proofs later. 

\begin{fact}\label{Fact1}
Consider a continuous function $g(x)$ defined on the closed interval $[a,b]$. Suppose $g(x)$ is twice-differentiable on the open interval $(a,b)$ and strictly concave. 
\begin{itemize}
\item [(A)] If 
\begin{equation}\label{Con1_fact}
g(b)<0<g(a),
\end{equation} 
there exists a unique solution to $g(x)=0$ and the solution is in the open interval $x^\dagger\in(a,b)$. Moreover, $g'(x^\dagger)<0$.
\item [(B)] If
\begin{equation}\label{Con2_fact}
g(a)>0 \text{ and } g(b)>0,
\end{equation} 
there exists no solution to $g(x)=0$.
\item [(C)] If
\begin{equation}\label{Con3_fact}
g(a)>0 \text{ and } g(b)=0,
\end{equation} 
the unique solution to $g(x)=0$ is $x^\dagger=b$.
\end{itemize}
\end{fact}
\begin{proof} 
\begin{itemize}
    \item [(A)] The existence of a solution to $g(x)=0$ directly follows from the intermediate value theorem (IVT). Next, we will prove the solution's uniqueness. The concavity of a twice-differentiable function $g(x)$ indicates that $\partial^2 g(x)/\partial x^2<0$; thus, $g'(x)\triangleq\partial g(x)/\partial x$ decreases with $x$. Two possible cases immediately arise. (i) If no $x_0\in(a,b)$ exists such that $g'(x_0)=0$, i.e., $g'(x)$ is always non-negative or non-positive over $[a,b]$, $g(x)$ monotonically changes with $x$. Considering \eqref{Con1_fact}, $g(x)$ actually decreases with $x$, and the maximum function value is reached at $x^*=a$. (ii) Otherwise, there exists $x_0\in(a,b)$ such that $g'(x_0)=0$. Hence, $g(x)$ first increases over $[a,x_0]$ and then decreases over $[x_0,b]$, and the maximum point is exactly $x_0$ with $g(x_0)>g(a)>0$. To sum up, there always exists a unique maximum point $x^*\in[a,b]$ with $g(x^*)>0$. Specifically, $g(x)$ decreases with $x$ on $[x^*,b]$. In addition, if $x^*\neq a$, $g(x)$ is always positive and increases with $x$ over $[a,x^*]$. In conjunction with $g(b)<0$, a unique solution $x^\dagger$ to $g(x)=0$ exists where $x^\dagger\in(x^*,b)$ and $g'(x^\dagger)<0$.

    \item [(B)] We start the proof with the fact that $g'(x)$ decreases with $x$. Similarly, consider the two possible cases mentioned above regarding whether there exists $x_0\in(a,b)$ such that $g'(x_0)=0$. (i) Since $g(x)$ monotonically changes with $x$ over $[a,b]$, $g(x)$ is always positive given \eqref{Con2_fact}, i.e., $g(x)>\min\left\{g(a),g(b)\right\}>0$. (ii) Since $g(x)$ first increases over $[a,x_0]$ and then decreases over $[x_0,b]$, $g(x)$ is always positive given \eqref{Con2_fact}, i.e., $g(x)>\max\left\{g(a),g(b)\right\}>0$. In conclusion, $g(x)>0$ always holds. Therefore, no solution exists for $g(x)=0$.

    \item [(C)] The proof follows a similar rationale as above. In both cases (i) and (ii), no matter whether $g(x)$ monotonically changes over $[a,b]$ or first increases in $[a,x_0]$ and then decreases in $[x_0,b]$, we have that $g(x)=0$ is obtained only when $x^\dagger=b$ given \eqref{Con3_fact}.
\end{itemize}
\end{proof}

Now we proceed to analyze the solution to the equation $\Phi(v^\dagger)=0$. Calculating the second-order derivative of $\Phi(v^\dagger)$ over $[0,\bar{v}]$ yields
\begin{equation}
    \frac{\partial^2 \Phi(v^\dagger)}{\partial (v^\dagger)^2}=\frac{1}{\delta}\frac{\partial^2 \tilde{J}(v^\dagger)}{\partial (v^\dagger)^2}<0
\end{equation}
according to Lemma \ref{claim:J_concave}. That is, $\Phi(v^\dagger)$ is also concave over $[0,\bar{v}]$. Furthermore, $\Phi(0)=p_0+\tilde{J}(0)/\delta=p_0+(\ln\omega_0)/\delta>0$.  We also have 
\begin{equation}\label{phi_bar}
    \Phi(\bar{v})=p_0+\frac{1}{\delta}\ln\left(n-1+\omega_0\right)-\bar{v}.
\end{equation}
In what follows, we discuss three possible cases based on the seller's uniform price $p_0$.
\begin{itemize}
    \item [(1)] If $p_0<\bar{v}-1/\delta\cdot\ln\left(n-1+\omega_0\right)$, then we have $\Phi(\bar{v})<0$. (i) For the closed interval $0\le v^\dagger\le\bar{v}$: Since $\Phi(\bar{v})<0<\Phi(0)$, we obtain that there exists a unique solution $v^\dagger\in(0,\bar{v})$ to $\Phi(v^\dagger)=0$ according to Fact \ref{Fact1}.(A). (ii) For $v^\dagger>\bar{v}$: Since $\Phi(\bar{v})<0$ and $\Phi(v^\dagger)$ decreases with $v^\dagger$, no solution to $\Phi(v^\dagger)=0$ exists here. Hence, we have a unique valuation threshold $v^*=v^\dagger\in(0,\bar{v})$. Moreover, $\partial \Phi(v^*)/\partial v^*<0$ according to Fact \ref{Fact1}.(A).
    
    \item [(2)] If $p_0>\bar{v}-1/\delta\cdot\ln\left(n-1+\omega_0\right)$, then we have $\Phi(\bar{v})>0$. (i) For the closed interval $0\le v^\dagger\le\bar{v}$: Since $\Phi(0)>0$ and $\Phi(\bar{v})>0$, we obtain that no solution to $\Phi(v^\dagger)=0$ exists according to Fact \ref{Fact1}.(B). (ii) For $v^\dagger>\bar{v}$: Since $\Phi(\bar{v})>0$ and $\Phi(v^\dagger)$ linearly decreases with $v^\dagger$, there exists a unique solution $v^\dagger\in(\bar{v},+\infty)$ to $\Phi(v^\dagger)=0$ in this case. Hence, we have a unique valuation threshold $v^*\triangleq\min\{v^\dagger,\bar{v}\}=\bar{v}$.
    
    \item [(3)] If $p_0=\bar{v}-1/\delta\cdot\ln\left(n-1+\omega_0\right)$, then we have $\Phi(\bar{v})=0$. (i) For the closed interval $0\le v^\dagger\le\bar{v}$: Since $\Phi(0)>0$ and $\Phi(\bar{v})=0$, we obtain that the unique solution to $\Phi(v^\dagger)=0$ is $\bar{v}$ according to Fact \ref{Fact1}.(C). (ii) For $v^\dagger>\bar{v}$: Since $\Phi(\bar{v})=0$ and $\Phi(v^\dagger)$ linearly decreases with $v^\dagger$, no solution to $\Phi(v^\dagger)=0$ exists in this case. Hence, we have a unique valuation threshold $v^*=v^\dagger=\bar{v}$.
\end{itemize}
The uniqueness of $v^*$ and BNE then follows. This completes the whole proof. 

\section{Proof for Theorem \ref{Theorem_PBE}}\label{proof_TheoremPBE}

To derive the PBE for the whole dynamic Bayesian game, we will combine the users' valuation threshold $v^*$ (see Theorem~\ref{Theorem:p0-vstar}) and the seller's optimal uniform pricing scheme $p_0^*$ in \eqref{base_price_equation1}-\eqref{base_price_equation3} (see Proposition \ref{base_price}) together. In what follows, we develop our analysis based on four cases concerning $v^*$.

\begin{itemize}
	\item [(1)] If $0\le v^*\le \bar{v}/2$, the seller's uniform pricing scheme is given in \eqref{base_price_equation1}. Combining it with \eqref{vstar_p0_summation} yields
	\begin{equation}\label{price_v_1}
		\frac{\bar{v}}{2}=v^*-\frac{1}{\delta}\tilde{J}(v^*).
	\end{equation}
    We now let $\Gamma(v^*)\triangleq\bar{v}/2-v^*+\tilde{J}(v^*)/\delta$, and its second-order derivative over $[0,\bar{v}/2]$ is given by
    \begin{equation}\label{gamma_sod}
        \frac{\partial^2 \Gamma(v^*)}{\partial (v^*)^2}=\frac{1}{\delta}\frac{\partial^2 \tilde{J}(v^*)}{\partial (v^*)^2}<0
    \end{equation}
    according to Lemma \ref{claim:J_concave}. Hence, $\Gamma(v^*)$ is concave over $[0,\bar{v}/2]$. We also have $\Gamma(0)=\bar{v}/2+(\ln\omega_0)/\delta>0$ and $\Gamma(\bar{v}/2)=\tilde{J}(\bar{v}/2)/\delta>0$. It follows that no solution to $\Gamma(v^*)=0$ or \eqref{price_v_1} exists on $[0,\bar{v}/2]$ according to Fact \ref{Fact1}.(B). Therefore, the PBE valuation threshold $v^*$ does not exist on $[0,\bar{v}/2]$.
	\item [(2)] If $\bar{v}/2<v^*\le \bar{v}/(2-\delta)$, the seller's uniform pricing scheme is given in \eqref{base_price_equation2}, i.e., $p_0^*(v^*)=v^*$. This leads to contradiction with $p_0^*=v^*-\tilde{J}(v^*)/\delta$ in \eqref{vstar_p0_summation}, where $\tilde{J}(v^*)>0$. Thus, the PBE valuation threshold  $v^*$ does not exist on $(\bar{v}/2,\bar{v}/(2-\delta)]$.
	\item [(3)] If $\bar{v}/(2-\delta)<v^*<\bar{v}$, the seller's uniform pricing scheme is given in \eqref{base_price_equation3}. Combining it with \eqref{vstar_p0_summation} yields
	\begin{equation}\label{case2_equ}
		\frac{\bar{v}-\delta v^*}{2(1-\delta)}=v^*-\frac{1}{\delta}\tilde{J}(v^*).
	\end{equation}
    We now let $\Gamma(v^*)\triangleq(\bar{v}-\delta v^*)/2(1-\delta)-v^*+\tilde{J}(v^*)/\delta$, and its second-order derivative over $(\bar{v}/(2-\delta),\bar{v})$ is also given by \eqref{gamma_sod}. So $\Gamma(v^*)$ here is concave over $(\bar{v}/(2-\delta),\bar{v})$. We have $\Gamma(\bar{v}/(2-\delta))=\tilde{J}(\bar{v}/(2-\delta))/\delta>0$ and
    \begin{equation}\label{phi_bar}
        \Gamma(\bar{v})=\frac{1}{\delta}\ln\left(n-1+\omega_0\right)-\frac{\bar{v}}{2}.
    \end{equation}
    If and only if $\Gamma(\bar{v})<0$, there exists a unique solution to $\Gamma(v^*)=0$ over $(\bar{v}/(2-\delta),\bar{v})$ (see Fact \ref{Fact1}). Equivalently, if and only  if $\bar{v}/2>1/\delta\cdot \ln(n-1+\omega_0)$, a unique PBE valuation threshold $v^*$ exists over $(\bar{v}/(2-\delta),\bar{v})$. This is the Case II in Theorem \ref{Theorem_PBE}, where $v^*<\bar{v}$ is the unique solution to \eqref{case2_equ} and simplifying \eqref{case2_equ} yields \eqref{vstar_PBE}. Moreover, $\partial \Gamma(v^*)/\partial v^*<0$ according to Fact \ref{Fact1}.(A). Notice that $p_0^*>\bar{v}/2$ can be easily verified by examining \eqref{base_price_equation3}. Additionally, we also examine the impact of $\bar{v}$ on the fraction of socially active users $v^*/\bar{v}$ in Lemma \ref{vbar_fraction} at the end of this appendix.
	\item [(4)] If $v^*=\bar{v}$, the seller's uniform pricing scheme is a special case of \eqref{base_price_equation3}, i.e., $p_0^*=\bar{v}/2$. On the other hand, $v^*=\bar{v}$ if and only if $p_0\ge\bar{v}-1/\delta\cdot\ln\left(n-1+\omega_0\right)$ (see the final part of Appendix \ref{proof_v^*}). Therefore, if and only if $\bar{v}/2\le1/\delta\cdot \ln(n-1+\omega_0)$, a unique PBE valuation threshold $v^*=\bar{v}$ exists. This is the Case I in Theorem \ref{Theorem_PBE}.
\end{itemize}
This completes the proof for Theorem \ref{Theorem_PBE}. 

Notice that we assume $\bar{v}> 2\ln(n-1+\omega_0)$ throughout the PBE analysis.\footref{footnote3} Recall the sufficient and necessary condition for Case I in Theorem \ref{Theorem_PBE}: $\bar{v}/2\le1/\delta\cdot \ln(n-1+\omega_0)$.~Suppose $\bar{v}\le 2\ln(n-1+\omega_0)$, this condition always holds regardless of any $\delta$ value over $(0,1)$. As such, only one trivial equilibrium exists under random user profiling, where $v^*=\bar{v}$ and $p_0^*=\bar{v}/2$.

\begin{lemma}\label{vbar_fraction}
If $\bar{v}/2>1/\delta\cdot \ln(n-1+\omega_0)$ (Case II in Theorem \ref{Theorem_PBE}), $v^*/\bar{v}$ decreases with $\bar{v}$ in the PBE.
\end{lemma}
\begin{proof}
Recall that $v^*$ is the unique solution to \eqref{case2_equ} or $\Gamma(v^*)=0$ over $(\bar{v}/(2-\delta),\bar{v})$. To examine the impact of $\bar{v}$ on $v^*/\bar{v}$, we now transform the variables by letting $a=v^*/\bar{v}$ and $b=\bar{v}$. It then follows that $v^*=ab$, $1/(2-\delta)<a<1$, and
\begin{align}
    \Gamma(v^*)\equiv \Gamma(a,b)= \frac{b-\delta ab}{2(1-\delta)}-ab+\frac{1}{\delta}\tilde{J}_a(a),
\end{align}
where $\tilde{J}_a(a)$ is a function of only variable $a$ here. Then we have the following first-order partial derivatives:
\begin{align}
    \frac{\partial \Gamma(a,b)}{\partial a}=-\frac{\delta b}{2(1-\delta)}-\left(b-\frac{1}{\delta}\frac{\partial\tilde{J}_a(a)}{\partial a}\right)<0,
\end{align}
where the last inequality follows from 
\begin{align}
\tilde{J}'(v^*)=\frac{\partial \tilde{J}(v^*)}{\partial v^*}   \equiv\frac{\partial \tilde{J}(v^*,\bar{v})}{\partial v^*}=\frac{\partial \tilde{J}(v^*,\bar{v})}{\partial (v^*/\bar{v})}\cdot\frac{1}{\bar{v}}=\frac{\partial\tilde{J}_a(a)}{\partial a}\cdot\frac{1}{\bar{v}}=\frac{\partial\tilde{J}_a(a)}{\partial a}\cdot\frac{1}{b}
\end{align}
with $\tilde{J}'(v^*)$ given in \eqref{FO_J} and $\tilde{J}'(v^*)<\delta$ (see Lemma \ref{jprime}); and
\begin{align}
    \frac{\partial \Gamma(a,b)}{\partial b}=\frac{1-\delta a}{2(1-\delta)}-a=\frac{1-(2-\delta)a}{2(1-\delta)}<0,
\end{align}
where the last inequality follows from $a>1/(2-\delta)$. According to the implicit function theorem, we have $\partial a/\partial b<0$, which then completes the proof. 
\end{proof}

\section{Proof for Proposition \ref{Proposition:perfectprofiling}}

The proof consists of two main steps. We first analyze users' social activities in Appendix \ref{Appendix:Perfect_Step1} and then the seller's uniform pricing in Appendix \ref{Appendix:Perfect_Step2}. Notice that we cannot directly obtain Proposition \ref{Proposition:perfectprofiling} from Theorem \ref{Theorem_PBE} for the general case with $\delta\in(0,1)$. This is because $\delta=1$ in the perfect profiling case results in a zero denominator in \eqref{base_price_equation3} for the backward analysis of the seller's uniform pricing (see Proposition \ref{base_price}). But the threshold structure ($v^*$, see Section \ref{forward}) in the forward analysis of users' social decisions remains valid under $\delta=1$.

\subsection{Step 1: Users' Social Activity Analysis} \label{Appendix:Perfect_Step1}

This proof is developed by contradiction. Suppose $v^\ast<\bar{v}$, i.e., there exist some high-valuation users who are inactive in the social network. Given $\delta=1$, we directly obtain the set of non-profiled users $\mathcal{N}_0=\{i\in\mathcal{N};v^*<v_i\le\bar{v}\}$ and the set of profiled users $\mathcal{N}_1=\{i\in\mathcal{N};0\le v_i\le v^*\}$ in this perfect profiling case. Then, we update the seller's posterior belief in Stage II, i.e.,
\begin{equation}\label{perfect_belief}
	\left\{
	\begin{aligned}
		&f(v_i|i\in\mathcal{N}_1)=\frac{1}{v^*}, &\quad\textrm{$0\le v_i\le v^*$,}\\
		&f(v_i|i\in\mathcal{N}_0)=\frac{1}{\bar{v}-v^\ast},  &\quad\textrm{$v^*<v_i\le\bar{v}$.}
	\end{aligned}
	\right.
\end{equation}
Given $v^*$, the seller's expected revenue extracted from the non-profiled users $\tilde{\Pi}_0$ is then given by
\begin{align}
	\tilde{\Pi}_0=|\mathcal{N}_{0}|\mathbb{E}_{v_i\sim f(\cdot|i\in\mathcal{N}_0)}\left[p_0\mathbbm{1}(v_i\ge p_0)\right]=\left\{
	\begin{aligned}
		&\frac{|\mathcal{N}_0|p_0(\bar{v}-p_0)}{\bar{v}-v^*}, &\quad\textrm{if $p_0> v^*$,}\\
		&|\mathcal{N}_0|p_0,  &\quad\textrm{if $p_0\le v^*$;}
	\end{aligned}
	\right.
\end{align}
where we take expectations over all possible non-profiled users in $\mathcal{N}_0$. We hence obtain the seller's optimal uniform price of $p_0^*=\arg\max\tilde{\Pi}_0=\max\{v^*,\bar{v}/2\}\ge v^*$. However, if $v^*<\bar{v}$, we must have $v^*>p_0$ according to \eqref{vstar_p0_summation} in Section \ref{forward}. This results in a contradiction. Therefore, we now have $v^*=\bar{v}$ for the perfect profiling case with $\delta=1$.

\subsection{Step 2: Seller's Uniform Pricing Analysis}\label{Appendix:Perfect_Step2}

To constitute the PBE, users' social activity decisions should be consistent with the seller's belief that is updated according to Bayes' theorem (i.e., belief consistency \cite{fudenberg1991perfect}). Moreover, anticipating the seller's uniform price $p_0$ in Stage II, users must have no incentive to deviate from the equilibrium strategies of $x_i^*=1,\forall i\in\mathcal{N}$. This fact then helps guide our proof in this subsection. Rewrite the final payoff for each user $i$ in \eqref{final_user_payoff} here:
\begin{equation}\label{App:final_user_payoff}
	\pi_i(x_i,\boldsymbol{x}_{-i})=x_i\ln\left(\sum_{j\neq i}x_j+\omega_0\right)+\max\{v_i-p_i,0\}.
\end{equation}
Then we must ensure that, when all other users choose the equilibrium strategy (i.e., $x^*_j=1,\forall j\neq i$), each user $i$'s payoff should satisfy the following inequality:
\begin{equation}\label{App:user_deviation}
    \pi_i(1,\boldsymbol{x}^*_{-i})\ge\pi_i(0,\boldsymbol{x}^*_{-i}).
\end{equation}
Simplifying this yields
$\ln\left(n-1+\omega_0\right)\ge\max\{v_i-p_0,0\}$, which must hold for any user valuation $0\le v_i\le \bar{v}$. That is,
\begin{equation}
	\ln\left(n-1+\omega_0\right)\ge\max_{0\le v_i\le \bar{v}}\max\{v_i-p_0,0\}=\bar{v}-p_0.
\end{equation}
We thus conclude with $p_0^*\ge \bar{v}-\ln\left(n-1+\omega_0\right)$. 

Here, any uniform price $p_0$ that is no less than $\bar{v}-\ln(n-1+\omega_0)$ can constitute the PBE. Without loss of generality, we choose the lowest value of $p_0^*$ here. Any value above $\bar{v}-\ln(n-1+\omega_0)$ would not change the equilibrium welfare outcomes for users and the seller. Note that the uniform price is never realized in the equilibrium under the perfect profiling case, as the seller can always make personalized offers to any user.

\section{Proof for Proposition \ref{Prop:PBE_delta}}

\subsection{Impact of $\delta$ on $v^*$}

We first analyze the impact of $\delta$ on $v^*$ in the PBE. In what follows, we develop our analysis based on the two cases in Theorem \ref{Theorem_PBE}.
\begin{itemize}
 \item \mbox{Case I with $\delta\le 2\ln(n-1+\omega_0)/\bar{v}$:} In this case, $v^*=\bar{v}$ remains unchanged with $\delta$.
\item \mbox{Case II with $\delta> 2\ln(n-1+\omega_0)/\bar{v}$:} Recall that $v^*$ is the unique solution to \eqref{vstar_PBE}, equivalent to $\Theta(v^*,\delta)=0$ with 
\begin{align}
&\Theta(v^*,\delta)\triangleq \delta \bar{v} - 2\delta v^*+\delta^2 v^\ast+2(1-\delta)\tilde{J}(v^\ast).
\end{align}
We then calculate the following partial derivatives regarding $\delta$:
\begin{align}
\frac{\partial \Theta(v^*,\delta)}{\partial \delta}=\bar{v}-2v^*+2\delta v^*-2\tilde{J}(v^*),\label{Gamma_FOC_sign}
\end{align}
and $\partial^2 \Theta(v^*,\delta)/\partial \delta^2=2v^*>0$, which suggests that $\partial \Theta(v^*,\delta)/\partial \delta$ increases with $\delta$ on $(2\ln(n-1+\omega_0)/\bar{v},1)$. We also calculate the following second-order partial derivative:
\begin{align}
\frac{\partial^2 \Theta(v^*,\delta)}{\partial \delta\partial v^*}=-2(1-\delta)-2\tilde{J}'(v^*)<0,    
\end{align}
where $\tilde{J}'(v^*)>0$ according to Lemma \ref{claim:J_concave}. Thus, $\partial \Theta(v^*,\delta)/\partial \delta$ decreases with $v^*$ in the PBE. Next, we discuss the signs of the function values for $\partial \Theta(v^*,\delta)/\partial \delta$ in \eqref{Gamma_FOC_sign}. In particular, when $\delta$ approaches $1$, we have 
\begin{align}\label{limit_1}
    \lim_{\delta\to 1}\frac{\partial \Theta(v^*,\delta)}{\partial \delta}=\bar{v}-2\tilde{J}(v^*)>0
\end{align}
for any feasible $v^*$. The last inequality in \eqref{limit_1} holds because 
\begin{align}
    \tilde{J}(v^*)<\tilde{J}(\bar{v})<\bar{v}/2,
\end{align}
where the first inequality holds because $\tilde{J}(\cdot)$ is an increasing function according to Lemma \ref{claim:J_concave}. The second inequality holds due to the assumption of $\bar{v}> 2\ln(n-1+\omega_0)$.\footref{footnote3} Now consider the other boundary point of $\delta=2\ln(n-1+\omega_0)/\bar{v}$. When $\delta=\hat{\delta} \triangleq 2\ln(n-1+\omega_0)/\bar{v}=2\tilde{J}(\bar{v})/\bar{v}$, for any feasible $v^*$, we have:
\begin{align}
\frac{\partial \Theta(v^*,\delta)}{\partial \delta}|_{\delta=\hat{\delta}}<\max_{v^*\in\left(\frac{\bar{v}}{2-\hat{\delta}},\bar{v}\right)}\frac{\partial \Theta(v^*,\delta)}{\partial \delta}|_{\delta=\hat{\delta}}
<&\frac{\partial \Theta(v^*,\delta)}{\partial \delta}|_{v^*=\frac{\bar{v}}{2-\hat{\delta}},\delta=\hat{\delta}},\nonumber\\
=&\left(\bar{v}-2v^*+\hat{\delta} v^*\right)+\left(\hat{\delta} v^*-2\tilde{J}(v^*)\right),\nonumber\\
=&\left(\bar{v}-2\frac{\bar{v}}{2-\hat{\delta}}+\hat{\delta}\frac{\bar{v}}{2-\hat{\delta}}\right)+2\left(\frac{\tilde{J}(\bar{v})}{2-\hat{\delta}}-\tilde{J}(\frac{\bar{v}}{2-\hat{\delta}})\right),\nonumber\\
=&2\left(\frac{\tilde{J}(\bar{v})}{2-\hat{\delta}}-\tilde{J}(\frac{\bar{v}}{2-\hat{\delta}})\right)<2\left(-\frac{1-\hat{\delta}}{2-\hat{\delta}}\tilde{J}(0)\right)<0.
\end{align}
    where the first inequality holds because $v^*$ in the PBE lies over $(\bar{v}/(2-\delta),\bar{v})$ (see Appendix \ref{proof_TheoremPBE}.(3)). The second inequality holds because $\partial \Theta(v^*,\delta)/\partial \delta$ decreases with $v^*$. The penultimate inequality follows from the strict concavity of $\tilde{J}(\cdot)$ over $[0,\bar{v}]$ (see Lemma \ref{claim:J_concave}). Therefore, since $\partial \Theta(v^*,\delta)/\partial \delta$ increases with $\delta$, there exists a unique solution of $\tilde{\delta}\in(\hat{\delta},1)$ to $\partial \Theta(v^*,\delta)/\partial \delta=0$. Moreover, $\partial \Theta(v^*,\delta)/\partial \delta<0$ for $\delta\in(\hat{\delta},\tilde{\delta})$ whereas $\partial \Theta(v^*,\delta)/\partial \delta>0$ for $\delta\in(\tilde{\delta},1)$. 

Next we turn to the second-order partial derivative of $\Theta(v^*,\delta)$ regarding $v^*$: 
\begin{equation}
    \frac{\partial^2 \Theta(v^*,\delta)}{\partial (v^*)^2}=2(1-\delta)\frac{\partial^2 \tilde{J}(v^\dagger)}{\partial (v^\dagger)^2}<0;
\end{equation}
plus, as $v^*$ in the PBE lies within $(\bar{v}/(2-\delta),\bar{v})$ (see Appendix \ref{proof_TheoremPBE}.(3)), we also examine $\Theta(\bar{v}/(2-\delta),\delta)=2(1-\delta)\tilde{J}(\bar{v}/(2-\delta))>0$ and $\Theta(\bar{v},\delta)=(1-\delta)\cdot\left(2\ln(n-1+\omega_0)-\delta \bar{v}\right)<0$. We have $\partial \Theta(v^*,\delta)/\partial v^*<0$ according to Fact \ref{Fact1}.(A).

Finally, we obtain through the implicit function theorem:
\begin{equation}
    \frac{\partial v^*(\delta)}{\partial \delta}=-\frac{\partial \Theta(v^*,\delta)/\partial \delta}{\partial \Theta(v^*,\delta)/\partial v^*}
\end{equation}
is negative for for $\delta\in(\hat{\delta},\tilde{\delta})$ whereas positive for $\delta\in(\tilde{\delta},1)$. Particularly, the unique solution to $\partial v^*/\partial \delta=0$ is $\tilde{\delta}$.
\end{itemize}

\subsection{Impact of $\delta$ on $p_0^*$}
Next, we analyze the impact of $\delta$ on $p_0^*$ in the PBE. Similarly, our analysis is based on the two cases in Theorem \ref{Theorem_PBE}.

\begin{itemize}
\item \mbox{Case I with $\delta\le 2\ln(n-1+\omega_0)/\bar{v}$:} In this case, $p_0^*=\bar{v}/2$ remains unchanged with $\delta$.
\item \mbox{Case II with $\delta> 2\ln(n-1+\omega_0)/\bar{v}$:} Recall that in the PBE, the seller's equilibrium uniform price $p_0^*$
satisfies
\begin{equation}
    p_0^*=v^*-\frac{1}{\delta}\tilde{J}(v^*)\equiv\frac{\bar{v}-\delta v^*}{2(1-\delta)}.
\end{equation}
Then for $\delta\in(\hat{\delta},\tilde{\delta})$, the last subsection proves that $\partial v^*/\partial \delta<0$. We adopt the expression $p_0^*=(\bar{v}-\delta v^*)/2(2-\delta)$ here to take the first-order derivative of $\delta$:
\begin{equation}
    \frac{\partial p_0^*}{\partial \delta}=\frac{\bar{v}-v^*+\delta(\delta-1)\cdot\partial v^*/\partial \delta}{2(1-\delta)^2}>0;
\end{equation}
and for $\delta\in(\tilde{\delta},1)$, the last subsection proves that $\partial v^*/\partial \delta>0$. Differently yet equivalently, we adopt the expression $p_0^*=v^*-\tilde{J}(v^*)/\delta$ to take the first-order derivative of $\delta$:
\begin{equation}
\frac{\partial p_0^*}{\partial \delta}=\frac{\partial v^*}{\partial \delta}(1-\frac{\tilde{J}'(v^*)}{\delta})+\frac{\tilde{J}(v^*)}{\delta^2}>0,
\end{equation}
where $\tilde{J}'(v^*)<\delta$ (see Lemma \ref{jprime} in the final). To sum up, $\partial p_0^*/\partial \delta$ always increases with $\delta$ over $(\hat{\delta},1)$. 
\end{itemize}

\begin{lemma}\label{jprime}
$\tilde{J}'(v^*)<\delta$ holds for $v^*$ in the BNE of the users' social interaction game and thus in the PBE. 
\end{lemma}
\begin{proof}
Notice that this proof is developed based on our BNE analysis of the users' social interaction game, where $p_0$ is viewed as given, and we revisit $\Phi(v^*)$ in \eqref{Phi}. Since $\partial \Phi(v^*)/\partial v^*<0$ in the BNE (see the final part (1) of Appendix \ref{proof_v^*}), we then have $\partial \Phi(v^*)/\partial v^*=\tilde{J}'(v^*)/\delta-1<0$. Furthermore, since $v^*$ for PBE must be the equilibrium valuation threshold for BNE in the users' social interaction game, we thus complete the proof.
\end{proof}

\section{Proof for Proposition \ref{prop_userpayoff}}\label{Proof:Userpayoff_Increse}

\subsection{Threshold Characterization}\label{threshold_characterization}
This appendix presents the detailed characterization of $\hat{v}$ and $\delta^\dagger(\cdot)$ in Proposition \ref{prop_userpayoff}. Specifically, the user valuation threshold $\hat{v}$ is given by
\begin{align}\label{hat_v}
    \hat{v}\triangleq\bar{v}-\ln(n-1+\omega_0).
\end{align}
Let $\tilde{p}_0^*$ denote the seller's uniform price in the PBE when the profiling accuracy is $\tilde{\delta}$. The profiling accuracy threshold function $\delta^\dagger(\cdot)$ is given by
\begin{align}\label{delta_func}
\delta^\dagger(v_i)=\left\{
\begin{aligned}
	&\tilde{\delta}, &\quad\textrm{if $v_i\le \tilde{p}_0^*$,}\\
	&\text{the unique solution $\delta^\dagger$ to $p_0^*(\delta^\dagger)=v_i$},  &\quad\textrm{if $\tilde{p}_0^*<v_i<\hat{v}$.}
\end{aligned}
\right.
\end{align}
Recall that in the PBE $p_0^*$ (alternatively, $p_0^*(\delta)$) is an increasing function of $\delta$ over $(\tilde{\delta},1)$ (see Proposition \ref{Prop:PBE_delta} and Fig. \ref{FIG:PBE_delta}), and thus $\tilde{p}_0^*<p_0^*(\delta)<\bar{v}-\ln(n-1+\omega_0)$. Accordingly, given $\tilde{p}_0^*<v_i<\hat{v}$, there always exists a unique solution $\delta^\dagger$ to $p_0^*(\delta^\dagger)=v_i$. Moreover, the inverse function of $p_0^*(\delta^\dagger)$ is also an increasing function. Namely, for $\tilde{p}_0^*<v_i<\hat{v}$, the unique solution $\delta^\dagger$ to $p_0^*(\delta^\dagger)=v_i$ increases with $v_i$ over $(\tilde{p}_0^*,\hat{v})$. To sum up, $\delta^\dagger(v_i)$ is non-decreasing in $v_i$.\footref{footnote4}

\subsection{Proof for Proposition \ref{prop_userpayoff}}

We first provide the following lemma (together with Fig. \ref{fig_userpayoff}) to characterize an individual user $i$'s payoff $\tilde{\pi}_i$ in the PBE. 

\begin{lemma}\label{Cor:userpayoff}
In the PBE, the expected user payoff for each user $i$ with valuation $v_i$ is given by
\begin{equation}\label{equation:pbeuserpayoff}
\tilde{\pi}_i(v_i;\delta)=\left\{
\begin{aligned}
&\tilde{J}(v^\ast),&&\textrm{if $\ 0\le v_i\le p_0^*$,}\\
&\tilde{J}(v^\ast)+(1-\delta)(v_i-p_0^*),&&\textrm{if $\ p_0^*< v_i\le v^\ast$,}\\
&v_i-p_0^*,&&\textrm{if $\ v^\ast< v_i\le\Bar{v}$.}
\end{aligned}
\right.
\end{equation}
\end{lemma}

\begin{figure}[h]
\centering
\begin{tikzpicture}[scale = 0.15, domain=0:45]
    \fill[blue!15](0,0) rectangle (32.54/1.5,4.40);
    \fill[blue!15](0,0) rectangle (32.54/1.5,4.40);
    \fill[red!15](23.73/1.5,4.40) -- (32.54/1.5,4.40) -- (32.54/1.5,8.81) -- cycle;
    \fill[red!15](32.54/1.5,0) -- (32.54/1.5,8.81) -- (40/1.5,16.27) -- (40/1.5,0) -- cycle;
    
	\draw[-latex] (0,0) -- (40/1.5+3,0) node[right] {\footnotesize $v_i$};
	\draw[-latex] (0,0) -- (0,23) node[right] {\footnotesize $\tilde{\pi}_i(v_i;\delta)$};
	
	\draw[color=black,densely dashed] (23.73/1.5,0) -- (23.73/1.5,4.40);
 	\draw[color=black,densely dashed] (32.54/1.5,0) -- (32.54/1.5,8.81);
 	\draw[color=black,densely dotted] (0,8.81) -- (32.54/1.5,8.81);
 	\draw[color=black,densely dotted] (0,16.27) -- (40/1.5,16.27);
 	\draw[color=black] (0,20) -- (40/1.5,20);
 	\draw[color=black] (40/1.5,0) -- (40/1.5,20);
 	
 	\draw[color=black,thick] plot[domain=0:23.73,samples=200] (\x/1.5,4.40);
 	\draw[color=black,thick] plot[domain=23.73:32.54,samples=200] (\x/1.5,{4.40+0.5*(\x-23.73)});
 	\draw[color=black,thick] plot[domain=32.54:40,samples=200] (\x/1.5,{\x-23.73});
 	
	\node at (0.1,0.1) [below left] {\scriptsize $O$};
	
	\node at (23.73/1.5,0) [below] {\scriptsize $p_0^*$};
	\draw[color=black,thick] (23.73/1.5,-0.3) -- (23.73/1.5,0.5);
 	\node at (32.54/1.5,0) [below] {\scriptsize $v^\ast$};
 	\draw[color=black,thick] (32.54/1.5,-0.3) -- (32.54/1.5,0.5);
 	\node at (40/1.5,0) [below] {\scriptsize $\bar{v}$};
 	\draw[color=black,thick] (40/1.5,-0.3) -- (40/1.5,0.5);
	
 	\node at (0,4.40) [left] {\scriptsize $\tilde{J}(v^\ast)$};
	\draw[color=black,thick] (-0.3,4.40) -- (0.5,4.40);
	\node at (0,16.27) [left] {\scriptsize $\bar{v}-p_0^*$};
	\draw[color=black,thick] (-0.3,16.27) -- (0.5,16.27);
	\draw[color=black,thick] (40/1.5-0.5,16.27) -- (40/1.5+0.3,16.27);
	\node at (0,20) [left] {\scriptsize $\bar{v}/2$};
	\draw[color=black,thick] (-0.3,20) -- (0.5,20);
 	\draw[color=black,thick] (40/1.5-0.5,20) -- (40/1.5+0.3,20);
	
 	\node at (20/1.5,0) [below left] {\scriptsize $\bar{v}/2$};
 	\draw[color=black,thick] (20/1.5,-0.3) -- (20/1.5,0.5);
	\draw [-,thin] (20/1.5,0) arc (0:-60:1.5);
	
 	\node at (0,10) [above right] {\scriptsize $\tilde{J}(v^\ast)+(1-\delta)(v^\ast-p_0^*)$};
 	\draw[color=black,thick] (-0.3,8.81) -- (0.5,8.81);
 	\draw [-,thin] (0,8.81) arc (-90:-40:5);
\end{tikzpicture}
\centering

\begin{tikzpicture}[scale = 0.15]
\fill[gray!5](-18,1.5) rectangle (31.5,-1.5);
\fill[blue!15](-17,1) rectangle (-10,-1);
\fill[red!15](10,1) rectangle (17,-1);
\node at (-10,0) [right] {\scriptsize Social Network Benefit};
\node at (17,-0.2) [right] {\scriptsize Purchase Surplus};
\end{tikzpicture}
\caption{Expected user payoff $\tilde{\pi}_i(v_i;\delta)$ under different user valuations $v_i$.}
\label{fig_userpayoff}
\end{figure}
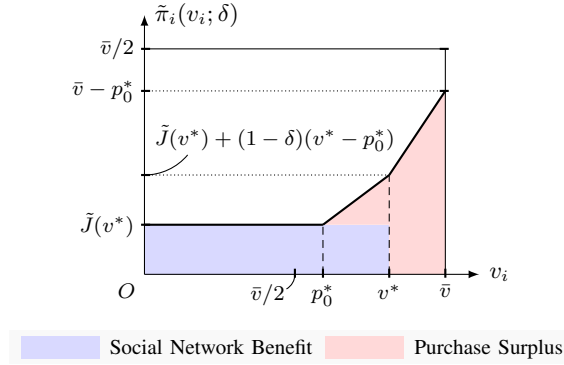

\begin{proof}
Recall that each individual user $i$'s expected final payoff is given in \eqref{expectpayoff}, which hinges on both his and the other users' equilibrium social decisions. Based on the PBE outcomes, we consider the following three possible cases regarding $v_i$. Notice that $p_0^*=v^*-\tilde{J}(v^*)/\delta<v^*$.
\begin{itemize}
    \item \mbox{For $0\le v_i\le p_0^*$:} Since $v_i\le p_0^*<v^*$, an individual user $i$ with valuation $v_i$ chooses $x_i^*=1$ in the PBE (see Belief \ref{threshold_structure}). In addition, he cannot afford the uniform price $p_0^*\ge v_i$ even if the seller fails to profile him and thus ends up with zero purchase surplus. As such, he merely receives the expected social network benefit $\tilde{J}(v^\ast)$ in \eqref{expected_social_network}.
    \item \mbox{For $p_0^*< v_i\le v^\ast$:} Since $v_i\le v^\ast$, an individual user $i$ with valuation $v_i$ chooses $x_i^*=1$ in the PBE (see Belief \ref{threshold_structure}). Given $v_i>p_0^*$, he can receive a mixture of purchase surplus with social network benefit. More specifically, his expected purchase surplus is $(1-\delta)\max\{v_i-p_0^*,0\}=(1-\delta)(v_i-p_0^*)$ and the expected social network benefit $\tilde{J}(v^\ast)$ is given in \eqref{expected_social_network}.
    \item \mbox{For $v^\ast< v_i\le\Bar{v}$: } Since $v_i>v^*$, an individual user $i$ with valuation $v_i$ chooses $x_i^*=0$ in the PBE (see Belief \ref{threshold_structure}). Thus, he receives zero social network benefit but only gains the purchase surplus. Specifically, the seller can only charge the uniform price $p_0^*$ to him, and thus the purchase surplus is $\max\{v_i-p_0^*,0\}=v_i-p_0^*$ given $v_i>v^*>p_0^*$.
\end{itemize}
This completes the proof for Lemma \ref{Cor:userpayoff}.
\end{proof}

Next, we will prove that under the set of conditions in \eqref{prop_userpayoff_condition}, an individual user $i$'s payoff increases with $\delta$. Specifically, we will discuss based on \eqref{hat_v} and \eqref{delta_func}.
\begin{itemize}
    \item \mbox{For $v_i\le \tilde{p}_0^*$ and $\delta>\tilde{\delta}$:} In this case, since $p_0^*(\delta)$ is an increasing function of $\delta$ over $(\tilde{\delta},1)$ (see Proposition \ref{Prop:PBE_delta}), we have $p_0^*(\delta)> p_0^*(\tilde{\delta})=\tilde{p}_0^*$. It follows that $v_i\le \tilde{p}_0^*< p_0^*(\delta)$. Hence, an individual user $i$ with valuation $v_i$ receives his expected payoff of $\tilde{\pi}_i(v_i;\delta)=\tilde{J}(v^\ast)$, according to Lemma \ref{Cor:userpayoff}.
    \item \mbox{For $\tilde{p}_0^*<v_i<\hat{v}$ and $\delta>\delta^\dagger(v_i)$:} Similarly, since $\delta>\delta^\dagger(v_i)>\delta^\dagger(\tilde{p}_0^*)=\tilde{\delta}$, we can have $p_0^*(\delta)>p_0^*\left(\delta^\dagger(v_i)\right)=v_i$ in this case. Hence, an individual user $i$ with valuation $v_i$ also receives his expected payoff of $\tilde{\pi}_i(v_i;\delta)=\tilde{J}(v^\ast)$.
\end{itemize}
Recall that $\tilde{J}(v^\ast)$ increases with $v^*$ (see Lemma \ref{claim:J_concave}) and $v^*$ increases with $\delta$ over $(\tilde{\delta},1)$ (see Proposition \ref{Prop:PBE_delta}). We thus conclude that $\tilde{\pi}_i(v_i;\delta)=\tilde{J}(v^\ast)$ increases with $\delta$ if \eqref{prop_userpayoff_condition} holds. This then completes the proof.

\section{Proof for Proposition \ref{Prop:Myopic-Compare}}

The proof consists of two main steps. Appendix \ref{step1_noaware} first introduces the no-awareness benchmark and its PBE outcomes. Then we compare user payoffs under this benchmark to that of our original model in the presence of user awareness in Appendix \ref{step2_noaware}.

\subsection{PBE Outcomes for No-awareness Benchmark}\label{step1_noaware}

This appendix presents the PBE outcomes for the no-awareness benchmark, where users are not aware of the seller's profiling in social networks.

\begin{lemma}[No-awareness benchmark]\label{Lem:Myopic}
In the benchmark of users lacking awareness, there exists a unique equilibrium as below. Hereafter, we use the superscript ``$\text{no}$" to refer to this no-awareness benchmark.
\begin{itemize}
    \item [(i)] In Stage I, all users choose the maximum social activity levels regardless of their valuations, i.e., $x_i^{\text{no}}(v_i)=1,\forall i\in\mathcal{N}$. In Stage II, the seller sets a uniform price of $p_0^{\text{no}}=\bar{v}/2$.
    
    \item[(ii)] In the equilibrium, the expected user payoff for each individual user $i$ with valuation $v_i$ is given by
\begin{equation}\label{equ:payoff_myopic}
    \tilde{\pi}_i^{\text{no}}(v_i;\delta)=
\ln(n-1+\omega_0)+(1-\delta)\cdot\max\{v_i-\frac{\bar{v}}{2},0\}.
\end{equation} 
\end{itemize}
\end{lemma}

\begin{proof}
Without awareness of the seller's profiling, each user $i$ decides his social activity level merely to maximize the social network benefit $J(x_i,\boldsymbol{x}_{-i})$ in \eqref{social_network_benefit}, which monotonically increases with $x_i$. As such, we have $x_i^{\text{no}}(v_i)=1,\forall i\in\mathcal{N}$. This also makes the distribution of non-profiled users' valuations remain uniform across $[0,\bar{v}]$, and thus the seller sets the uniform price as $p_0^*=\arg\max p_0\left(1-F(p_0)\right)=\bar{v}/2$. The expected user payoff for each user $i$ then follows.
\end{proof}

\subsection{Comparison of User Payoffs}\label{step2_noaware}

This appendix compares the expected user payoff for each user $i$ in \eqref{equ:payoff_myopic} ($\tilde{\pi}_i(v_i;\delta)$, in the absence of user awareness) to \eqref{equation:pbeuserpayoff} ($\tilde{\pi}_i^{\text{no}}(v_i;\delta)$, in the presence of user awareness).

\subsubsection{If $0\le\delta\le \hat{\delta}$}
In the presence of user awareness, we have $v^*=\bar{v}$ and $p_0^*=\bar{v}/2$ (see Theorem \ref{Theorem_PBE}). Accordingly, \eqref{equation:pbeuserpayoff} degenerates into  $\tilde{\pi}_i(v_i;\delta)=\ln(n-1+\omega_0)+(1-\delta)\cdot\max\{v_i-\bar{v}/2,0\}$, the same as $\tilde{\pi}_i^{\text{no}}(v_i;\delta)$ in \eqref{equ:payoff_myopic}.

\subsubsection{If $\hat{\delta}<\delta<1$}
(i) First, consider when $0\le v_i\le v^*$. We will have $\tilde{\pi}_i(v_i;\delta)=
\tilde{J}(v^*)+(1-\delta)\max\{v_i-p_0^*,0\}$ from \eqref{equation:pbeuserpayoff}. Since $v^*<\bar{v}$ in this case, we have $\tilde{J}(v^*)<\tilde{J}(\bar{v})=\ln(n-1+\omega_0)$ according to Lemma \ref{claim:J_concave}. In conjunction with $p_0^*>\bar{v}/2$, one can easily prove that $\tilde{\pi}_i(v_i;\delta)<\tilde{\pi}_i^{\text{no}}(v_i;\delta)$. (ii) Next consider when $v^*<v_i\le \bar{v}$, and notice that we already have $\tilde{\pi}_i(v_i;\delta)<\tilde{\pi}_i^{\text{no}}(v_i;\delta)$ at $v_i=v^*$. We also have $\partial \tilde{\pi}_i(v_i;\delta)/\partial v_i -\partial \tilde{\pi}_i^{\text{no}}(v_i;\delta)/\partial v_i=\delta>0$. That is, $\tilde{\pi}_i(v_i;\delta)-\tilde{\pi}_i^{\text{no}}(v_i;\delta)$ increases with $v_i$ over $(v^*,\bar{v}]$. Since $\tilde{\pi}_i(v_i;\delta)>\tilde{\pi}_i^{\text{no}}(v_i;\delta)$ at $v_i=\bar{v}$ (see Lemma \ref{lemma_v_i} in the final), we can conclude that there exists a unique solution $v^\dagger\in(v^*,\bar{v})$ to $\tilde{\pi}_i(v^\dagger;\delta)=\tilde{\pi}_i^{\text{no}}(v^\dagger;\delta)$. More specifically, we have $v^\dagger=\left(2\ln(n-1±\omega_0)+2p_0^*-(1-\delta)\bar{v}\right)/2\delta$, and $\tilde{\pi}_i(v_i;\delta)<\tilde{\pi}_i^{\text{no}}(v_i;\delta)$ if $v^*<v_i<v^\dagger$ whereas $\tilde{\pi}_i(v_i;\delta)>\tilde{\pi}_i^{\text{no}}(v_i;\delta)$ if $v^\dagger<v_i\le\bar{v}$.

\subsubsection{If $\delta=1$}
In the presence of user awareness, we have $v^*=\bar{v}$ and $p_0^*=\bar{v}-\ln(n-1+\omega_0)$ (see Proposition \ref{Proposition:perfectprofiling}). Substituting $v^*$, $p_0^*$, and $\delta=1$ into \eqref{equation:pbeuserpayoff} yields $\tilde{\pi}_i(v_i;1)=\ln(n-1+\omega_0)$, the same as $\tilde{\pi}_i^{\text{no}}(v_i;1)$ in \eqref{equ:payoff_myopic}.

We thus complete the proof for Proposition \ref{Prop:Myopic-Compare}.

\begin{lemma}\label{lemma_v_i}
$\tilde{\pi}_i(\bar{v};\delta)>\tilde{\pi}_i^{\text{no}}(\bar{v};\delta)$, i.e., $(1+\delta) \bar{v}/2-p_0^*-\ln(n-1+\omega_0)>0$, $\forall \delta\in(\hat{\delta},1)$.
\end{lemma}

\begin{proof}
According to Lemma \ref{claim:J_concave} and Lemma \ref{lemma_1/2} later, we have $0<\tilde{J}'(v^*)<\delta/2$, $\forall \delta\in(\hat{\delta},1)$. Then we have
\begin{align}
&\tilde{J}'(v^*)<\frac{\delta}{2}\cdot\frac{(1-\delta)\bar{v}}{v^*-\delta\bar{v}},\tag{Line 1}\\
\Longleftrightarrow\quad& 2(v^*-\delta\bar{v})\tilde{J}'(v^*)<\delta(1-\delta)\bar{v},\tag{Line 2}\\
\Longleftrightarrow\quad& 2(1-\delta)\left(\frac{v^*}{\delta\bar{v}}-1\right)\tilde{J}'(v^*)<(1-\delta)^2,\tag{Line 3}\\
\Longleftrightarrow\quad& 2(1-\delta)\frac{v^*}{\delta\bar{v}}\tilde{J}'(v^*)<(1-\delta)^2+2(1-\delta)\tilde{J}'(v^*),\tag{Line 4}\\
\Longleftrightarrow\quad& 1-2(1-\delta)\frac{v^*}{\delta\bar{v}}\tilde{J}'(v^*)>2\delta-\delta^2-2(1-\delta)\tilde{J}'(v^*),\label{inequality}
\end{align}
where Line 1 follows from $(1-\delta)\bar{v}>v^*-\delta\bar{v}$ where $v^*<\bar{v}$ if $\hat{\delta}<\delta<1$. Moreover, $v^*-\delta \bar{v}>0$ since $v^*/\bar{v}>1/(2-\delta)>\delta$. To see the meaning of both sides of \eqref{inequality}, recall that $v^*$ is the unique solution to \eqref{vstar_PBE} or \eqref{case2_equ}, equivalent to $\Gamma(v^*)=0$ with 
\begin{align}
	\Gamma(v^*)\triangleq\frac{\bar{v}-\delta v^*}{2(1-\delta)}-v^*+\frac{1}{\delta}\tilde{J}(v^*).
\end{align}
Since we will hereafter focus on the impact of $\bar{v}$ on $v^*$, let us rewrite $\Gamma(v^*)\equiv \Gamma(v^*,\bar{v})$, and notice that $\tilde{J}(v^*)$ also hinges on $\bar{v}$ (alternatively, $\tilde{J}(v^*,\bar{v})$). We now take the following first-order partial derivatives:
\begin{equation}
    \frac{\partial \Gamma(v^*,\bar{v})}{\partial v^*}=-\frac{\delta}{2(1-\delta)}-1+\frac{1}{\delta}\tilde{J}'(v^*)=-\frac{2\delta-\delta^2-2(1-\delta)\tilde{J}'(v^*)}{2(1-\delta)\delta}
\end{equation}
and
\begin{equation}\label{gamma_vbar}
    \frac{\partial \Gamma(v^*,\bar{v})}{\partial \bar{v}}=\frac{1}{2(1-\delta)}+\frac{1}{\delta}\cdot\frac{\partial \tilde{J}(v^*)}{\partial\bar{v}}=\frac{1}{2(1-\delta)}-\frac{1}{\delta}\cdot\frac{v^*}{\bar{v}}\tilde{J}'(v^*)=\frac{1-v^*/\bar{v}\cdot1/\delta\cdot2(1-\delta)\tilde{J}'(v^*)}{2(1-\delta)},
\end{equation}
where $\tilde{J}'(v^*)$ is given in \eqref{FO_J}. In particular, the second equality of \eqref{gamma_vbar} follows from
\begin{align}
\tilde{J}'(v^*)=\frac{\partial \tilde{J}(v^*)}{\partial v^*}   \equiv\frac{\partial \tilde{J}(v^*,\bar{v})}{\partial v^*}=\frac{\partial \tilde{J}(v^*,\bar{v})}{\partial (v^*/\bar{v})}\cdot\frac{1}{\bar{v}}=-\frac{\bar{v}}{v^*}\left(-\frac{\partial \tilde{J}(v^*,\bar{v})}{\partial (v^*/\bar{v})}\cdot\frac{v^*}{\bar{v}^2}\right)=-\frac{\bar{v}}{v^*}\cdot \frac{\partial \tilde{J}(v^*,\bar{v})}{\partial \bar{v}}.
\end{align}
Then \eqref{inequality} indicates that
\begin{align}
    \frac{\partial \Gamma(v^*,\bar{v})}{\partial \bar{v}}>-\delta\frac{\partial \Gamma(v^*,\bar{v})}{\partial v^*}
    \quad\Longleftrightarrow\quad-\frac{\partial \Gamma(v^*,\bar{v})/\partial \bar{v}}{\partial \Gamma(v^*,\bar{v})/\partial v^*}>\delta,
\end{align}
with $\partial \Gamma(v^*,\bar{v})/\partial v^*<0$ (see Appendix \ref{proof_TheoremPBE}.(3)). According to the implicit function theorem, we have $\partial v^*/\partial \bar{v}>\delta$, $\forall \delta\in(\hat{\delta},1)$. It follows that
\begin{align}
1-\delta\frac{\partial v^*}{\partial \bar{v}}<1-\delta^2
\quad\Longleftrightarrow\quad& \frac{1}{2(1-\delta)}-\frac{\delta}{2(1-\delta)}\cdot\frac{\partial v^*}{\partial \bar{v}}<\frac{1+\delta}{2}\label{foc_lambda}.
\end{align}
Denote $\Lambda(\bar{v})\triangleq(1+\delta) \bar{v}/2-p_0^*(\bar{v})-\ln(n-1+\omega_0)$, where $p_0^*$ in the PBE is a function of $\bar{v}$. Recall that $p_0^*=\left(\bar{v}-\delta v^*\right)/2(1-\delta)$ in \eqref{base_price_equation3}; thus we have
\begin{align}
    \Lambda(\bar{v})=\frac{(1+\delta) \bar{v}}{2}-\frac{\bar{v}-\delta v^*}{2(1-\delta)}-\ln(n-1+\omega_0),
\end{align}
and \eqref{foc_lambda} implies that its first-order derivative is positive, i.e.,
\begin{align}
    \frac{\partial\Lambda(\bar{v})}{\partial \bar{v}}=\frac{(1+\delta)}{2}-\frac{1}{2(1-\delta)}+\frac{\delta}{2(1-\delta)}\cdot\frac{\partial v^*}{\partial \bar{v}}>0.
\end{align}
Since $\delta>\hat{\delta}$, we have $\bar{v}>2\ln(n-1+\omega_0)/\delta$. Furthermore, at $\bar{v}=2\ln(n-1+\omega_0)/\delta$, we have $v^*=\bar{v}$ from $\delta=\hat{\delta}$, and thus $\Lambda(\bar{v})=0$. As such, we have $\Lambda(\bar{v})>0$ for $\bar{v}>2\ln(n-1+\omega_0)/\delta$. This then completes the proof.
\end{proof}

\begin{lemma}\label{lemma_1/2}
$\tilde{J}'(v^*)<\delta/2$, $\forall \delta\in(\hat{\delta},1)$.
\end{lemma}

\begin{proof}
According to Lemma \ref{claim:J_concave}, $\tilde{J}'(v^*)$ in \eqref{FO_J} decreases with $v^*$. Recall that $v^*$ lies within $(\bar{v}/(2-\delta),\bar{v})$ in the PBE (see Appendix \ref{proof_TheoremPBE}.(3)), and we further relax this feasible range into $(\bar{v}/2,\bar{v})$. We thus have
\begin{align}\label{1_delta/2}
\tilde{J}'(v^*)<\tilde{J}'\left(\frac{\bar{v}}{2}\right)=\frac{1}{\bar{v}}\left(\frac{1}{2}\right)^{n-2}\cdot\sum_{m=1}^{n-1}\frac{(n-1)!}{(n-1-m)!(m-1)!}\left(\ln(m+\omega_0)-\ln(m-1+\omega_0)\right),
\end{align}
where $\tilde{J}'(\bar{v}/2)$ decreases with $\bar{v}$. Since $\delta>\hat{\delta}$, we have $\bar{v}>2\ln(n-1+\omega_0)/\delta$. Thus, we have
\begin{align}\label{2_delta/2}
\tilde{J}'\left(\frac{\bar{v}}{2}\right)&<\tilde{J}'\left(\frac{\bar{v}}{2}\right)|_{\bar{v}=\frac{2\ln(n-1+\omega_0)}{\delta}},\nonumber\\&=\frac{\delta}{\ln(n-1+\omega_0)}\left(\frac{1}{2}\right)^{n-1}\sum_{m=1}^{n-1}\frac{(n-1)!}{(n-1-m)!(m-1)!}\left(\ln(m+\omega_0)-\ln(m-1+\omega_0)\right)<\frac{\delta}{2},
\end{align}
where the last inequality is proved in Lemma \ref{lemma_1/3} below. Combining \eqref{1_delta/2} and \eqref{2_delta/2} completes the proof.
\end{proof}

\begin{landscape}
\begin{lemma}\label{lemma_1/3}
The following inequality holds:
\begin{align}
    \frac{\delta}{\ln(n-1+\omega_0)}\left(\frac{1}{2}\right)^{n-1}\sum_{m=1}^{n-1}\frac{(n-1)!}{(n-1-m)!(m-1)!}\left(\ln(m+\omega_0)-\ln(m-1+\omega_0)\right)<\frac{\delta}{2},\quad\forall \delta\in(\hat{\delta},1).
\end{align}
\end{lemma}

\begin{proof}
This is equivalent to prove
\begin{align}
&\frac{1}{\ln(n-1+\omega_0)}\left(\frac{1}{2}\right)^{n-1}\sum_{m=1}^{n-1}\frac{(n-1)!}{(n-1-m)!(m-1)!}\left(\ln(m+\omega_0)-\ln(m-1+\omega_0)\right)<\frac{1}{2},\nonumber\\
\Longleftrightarrow\ &\sum_{m=1}^{n-1}\frac{(n-1)!}{(n-1-m)!(m-1)!}\left(\ln(m+\omega_0)-\ln(m-1+\omega_0)\right)< 2^{n-2}\ln(n-1+\omega_0),\nonumber\\
\Longleftrightarrow\ &\sum_{m=1}^{n-1}\frac{(n-1)!}{(n-1-m)!(m-1)!}\left(\ln(m+\omega_0)-\ln(m-1+\omega_0)\right)<2^{n-2}\ln\omega_0+2^{n-2}\sum_{m=1}^{n-1}\left(\ln(m+\omega_0)-\ln(m-1+\omega_0)\right),\nonumber\\
\Longleftrightarrow\ &\sum_{m=1}^{n-1}\sum_{k=1}^{n-1}\frac{(n-2)!}{(n-1-m)!(m-1)!}\left(\ln(m+\omega_0)-\ln(m-1+\omega_0)\right)<2^{n-2}\ln\omega_0+\sum_{k=0}^{n-2}\frac{(n-2)!}{(n-2-k)!k!}\cdot\sum_{m=1}^{n-1}\left(\ln(m+\omega_0)-\ln(m-1+\omega_0)\right),\nonumber\\
\Longleftrightarrow\ &\sum_{m=1}^{n-1}\sum_{k=1}^{n-1}\frac{(n-2)!}{(n-1-m)!(m-1)!}\left(\ln(m+\omega_0)-\ln(m-1+\omega_0)\right)<2^{n-2}\ln\omega_0+\sum_{m=1}^{n-1}\sum_{k=1}^{n-1}\frac{(n-2)!}{(n-1-k)!(k-1)!}\left(\ln(m+\omega_0)-\ln(m-1+\omega_0)\right),\label{long}
\end{align}
where the penultimate line applies the binomial expansion of $2^{n-2}=(1+1)^{n-2}$, and the last line transforms the summation index from $k$ to $k-1$. Since $2^{n-2}\ln\omega_0>0$, in what follows, we will show
\begin{equation}\label{short}
\sum_{m=1}^{n-1}\sum_{k=1}^{n-1}\frac{(n-2)!}{(n-1-m)!(m-1)!}\left(\ln(m+\omega_0)-\ln(m-1+\omega_0)\right)\le\sum_{m=1}^{n-1}\sum_{k=1}^{n-1}\frac{(n-2)!}{(n-1-k)!(k-1)!}\left(\ln(m+\omega_0)-\ln(m-1+\omega_0)\right)    
\end{equation}
to complete the proof for \eqref{long}. For ease of illustration, we hereafter denote 
\begin{align}\label{mathcal_L_C}
    \mathcal{L}_m\triangleq\ln(m+\omega_0)-\ln(m-1+\omega_0),\text{ and } \mathcal{C}_k=\frac{(n-2)!}{(n-1-k)!(k-1)!}.
\end{align}
Before we continue to prove the inequality in \eqref{short}, we first present two symmetry properties regarding $\mathcal{L}_m$ and $\mathcal{C}_k$:
\begin{align}\label{properties}
    \mathcal{C}_k=\mathcal{C}_{n-k}, \forall k\in\{1,\dots,n-1\},\text{ and }\mathcal{L}_m+\mathcal{L}_{m+1}< \mathcal{L}_{m-1}+\mathcal{L}_{m+2};
\end{align}
the former one could be easily verified, and the later one follows from the convexity of the first-order derivative of the logarithm function (i.e., $\mathcal{L}_{m+2}-\mathcal{L}_{m+1}>\mathcal{L}_{m}-\mathcal{L}_{m-1}$). Now \eqref{short} is rewritten as
\begin{align}\label{shortshort}
\sum_{m=1}^{n-1}\sum_{k=1}^{n-1}\mathcal{C}_m\mathcal{L}_{m}-\sum_{m=1}^{n-1}\sum_{k=1}^{n-1}\mathcal{C}_k\mathcal{L}_{m} \le0,
\end{align}
and we consider all possible cases concerning $n-1\ge1$ below:
\begin{itemize}
\item[(1)] When $n-1=2$ (i.e., $n=3$), we have
\begin{align}
\sum_{m=1}^{2}\sum_{k=1}^{2}\mathcal{C}_m\mathcal{L}_{m}-\sum_{m=1}^{2}\sum_{k=1}^{2}\mathcal{C}_k\mathcal{L}_{m}=2\left(\mathcal{C}_1\mathcal{L}_{1}+\mathcal{C}_2\mathcal{L}_{2}\right)-\left(\mathcal{C}_1\mathcal{L}_{1}+\mathcal{C}_2\mathcal{L}_{1}+\mathcal{C}_1\mathcal{L}_{2}+\mathcal{C}_2\mathcal{L}_{2}\right)=0,
\end{align}
with $\mathcal{C}_1=\mathcal{C}_2$. This then indicates that \eqref{shortshort} holds true.

\item[(2)] When $n-1=1$ (i.e., $n=2$), we have $\mathcal{C}_1\mathcal{L}_{1}-\mathcal{C}_1\mathcal{L}_{1}=0$, which indicates that \eqref{shortshort} holds true.

\item [(3)] When $n-1$ is even and $n-1>2$, we have
\begin{align}
    \sum_{m=1}^{n-1}\sum_{k=1}^{n-1}\mathcal{C}_m\mathcal{L}_{m}-\sum_{m=1}^{n-1}\sum_{k=1}^{n-1}\mathcal{C}_k\mathcal{L}_{m}
    =&2\sum_{k=1}^{\frac{n-1}{2}}\left(\sum_{m=1}^{n-1}\mathcal{C}_m\mathcal{L}_{m}\right)-2\sum_{k=1}^{\frac{n-1}{2}}\left(\sum_{m=1}^{n-1}\mathcal{C}_k\mathcal{L}_{m}\right),\tag{Line 1}\\
    =&2\sum_{k=1}^{\frac{n-1}{2}}\left(\left(\sum_{m=1}^{\frac{n-1}{2}}\mathcal{C}_m\mathcal{L}_{m}+\sum_{m=\frac{n+1}{2}}^{n-1}\mathcal{C}_m\mathcal{L}_{m}\right)-\left(\sum_{m=1}^{\frac{n-1}{2}}\mathcal{C}_k\mathcal{L}_{m}+\sum_{m=\frac{n+1}{2}}^{n-1}\mathcal{C}_k\mathcal{L}_{m}\right)\right),\nonumber\\ 
    =&2\sum_{k=1}^{\frac{n-1}{2}}\left(\left(\sum_{m=1}^{\frac{n-1}{2}}\mathcal{C}_m\mathcal{L}_{m}+\sum_{m=1}^{\frac{n-1}{2}}\mathcal{C}_m\mathcal{L}_{n-m}\right)-\left(\sum_{m=1}^{\frac{n-1}{2}}\mathcal{C}_k\mathcal{L}_{m}+\sum_{m=1}^{\frac{n-1}{2}}\mathcal{C}_k\mathcal{L}_{n-m}\right)\right),\tag{Line 3}\\ 
    =&2\sum_{k=1}^{\frac{n-1}{2}}\sum_{m=1}^{\frac{n-1}{2}}\left(\mathcal{C}_m-\mathcal{C}_k\right)(\mathcal{L}_m+\mathcal{L}_{n-m}),\nonumber\\ 
    =&2\sum_{k=1}^{\frac{n-1}{2}}\left(\sum_{m=1}^{k-1}\left(\mathcal{C}_m-\mathcal{C}_k\right)(\mathcal{L}_m+\mathcal{L}_{n-m})+\sum_{m=k+1}^{\frac{n-1}{2}}\left(\mathcal{C}_m-\mathcal{C}_k\right)(\mathcal{L}_m+\mathcal{L}_{n-m})\right),\nonumber\\ 
    =&2\left(\sum_{m=1}^{\frac{n-3}{2}}\sum_{k=m+1}^{\frac{n-1}{2}}\left(\mathcal{C}_m-\mathcal{C}_k\right)(\mathcal{L}_m+\mathcal{L}_{n-m})+\sum_{k=1}^{\frac{n-3}{2}}\sum_{m=k+1}^{\frac{n-1}{2}}\left(\mathcal{C}_m-\mathcal{C}_k\right)(\mathcal{L}_m+\mathcal{L}_{n-m})\right),\nonumber\\ 
    =&2\left(\sum_{k=1}^{\frac{n-3}{2}}\sum_{m=k+1}^{\frac{n-1}{2}}\left(\mathcal{C}_k-\mathcal{C}_m\right)(\mathcal{L}_k+\mathcal{L}_{n-k})+\sum_{k=1}^{\frac{n-3}{2}}\sum_{m=k+1}^{\frac{n-1}{2}}\left(\mathcal{C}_m-\mathcal{C}_k\right)(\mathcal{L}_m+\mathcal{L}_{n-m})\right),\tag{Line 7}\\ 
    =&2\sum_{k=1}^{\frac{n-3}{2}}\sum_{m=k+1}^{\frac{n-1}{2}}\left(\mathcal{C}_m-\mathcal{C}_k\right)\left((\mathcal{L}_m+\mathcal{L}_{n-m})-(\mathcal{L}_k+\mathcal{L}_{n-k})\right)<0,\label{even}
\end{align}
where Line 1 follows from $\mathcal{C}_k=\mathcal{C}_{n-k}$ in \eqref{properties}. While applying this fact, Line 3 transforms the summation index from $m$ to $n-m$. Line 7 exchanges the summation index between $m$ and $k$. The last line follows from $k<m<(n-1)/2$, which leads to $\mathcal{C}_m>\mathcal{C}_k$ (since $\mathcal{C}_k<\mathcal{C}_{k+1}$ for any $k<(n-1)/2$, as one could easily verified) and $\mathcal{L}_m-\mathcal{L}_k<\mathcal{L}_{n-k}-\mathcal{L}_{n-m}$ (because of the convexity of the first-order derivative of the logarithm function and $k<m<n-m<n-k$). Therefore, \eqref{shortshort} holds true in this case.

\item [(4)] When $n-1$ is odd and $n-1>1$, we have
\begin{align}
    &\sum_{m=1}^{n-1}\sum_{k=1}^{n-1}\mathcal{C}_m\mathcal{L}_{m}-\sum_{m=1}^{n-1}\sum_{k=1}^{n-1}\mathcal{C}_k\mathcal{L}_{m}\nonumber\\
    =&\left(\sum_{k=1\text{ and }k\neq \frac{n}{2}}^{n-1}\sum_{m=1\text{ and }m\neq \frac{n}{2}}^{n-1}\mathcal{C}_m\mathcal{L}_{m}+\sum_{m=1}^{n-1}\mathcal{C}_m\mathcal{L}_{m}+\sum_{k=1}^{n-1}\mathcal{C}_{\frac{n}{2}}\mathcal{L}_{\frac{n}{2}}-\mathcal{C}_{\frac{n}{2}}\mathcal{L}_{\frac{n}{2}}\right)\nonumber\\&\qquad-\left(\sum_{k=1\text{ and }k\neq \frac{n}{2}}^{n-1}\sum_{m=1\text{ and }m\neq \frac{n}{2}}^{n-1}\mathcal{C}_k\mathcal{L}_{m}+\sum_{m=1}^{n-1}\mathcal{C}_{\frac{n}{2}}\mathcal{L}_{m}+\sum_{k=1}^{n-1}\mathcal{C}_{k}\mathcal{L}_{\frac{n}{2}}-\mathcal{C}_{\frac{n}{2}}\mathcal{L}_{\frac{n}{2}}\right),\nonumber\\
    =&\underbrace{\left(\sum_{k=1\text{ and }k\neq \frac{n}{2}}^{n-1}\sum_{m=1\text{ and }m\neq \frac{n}{2}}^{n-1}\mathcal{C}_m\mathcal{L}_{m}-\sum_{k=1\text{ and }k\neq \frac{n}{2}}^{n-1}\sum_{m=1\text{ and }m\neq \frac{n}{2}}^{n-1}\mathcal{C}_k\mathcal{L}_{m}\right)}_{\text{(Symmetric) Part A }}+\underbrace{\left(\sum_{m=1}^{n-1}\mathcal{C}_{m}\mathcal{L}_{m}+\sum_{k=1}^{n-1}\mathcal{C}_{\frac{n}{2}}\mathcal{L}_{\frac{n}{2}}-\sum_{m=1}^{n-1}\mathcal{C}_{\frac{n}{2}}\mathcal{L}_{m}-\sum_{k=1}^{n-1}\mathcal{C}_{k}\mathcal{L}_{\frac{n}{2}}\right)}_{\text{(Centric) Part B}}.\label{odd}
\end{align}

Notice that Part A is symmetric such that we can analyze it similarly to \eqref{even}. In particular, for those summation index bound of $(n-1)/2$ and $(n+1)/2$ in \eqref{even}, we will substitute them  with $n/2-1$ and $n/2+1$ here. As such, we can show that Part A of \eqref{odd} is also negative. We now focus on the centric part, Part B, of \eqref{odd} in what follows.
\begin{align}
&\sum_{m=1}^{n-1}\mathcal{C}_{m}\mathcal{L}_{m}+\sum_{k=1}^{n-1}\mathcal{C}_{\frac{n}{2}}\mathcal{L}_{\frac{n}{2}}-\sum_{m=1}^{n-1}\mathcal{C}_{\frac{n}{2}}\mathcal{L}_{m}-\sum_{k=1}^{n-1}\mathcal{C}_{k}\mathcal{L}_{\frac{n}{2}}\nonumber\\
=&\left(\sum_{m=1}^{\frac{n}{2}-1}\mathcal{C}_{m}\mathcal{L}_{m}+\sum_{m=1}^{\frac{n}{2}-1}\mathcal{C}_{m}\mathcal{L}_{n-m}+\mathcal{C}_{\frac{n}{2}}\mathcal{L}_{\frac{n}{2}}\right)+\left(n-1\right)\mathcal{C}_{\frac{n}{2}}\mathcal{L}_{\frac{n}{2}}-\left(\sum_{m=1}^{\frac{n}{2}-1}\mathcal{C}_{\frac{n}{2}}\mathcal{L}_{m}+\sum_{m=1}^{\frac{n}{2}-1}\mathcal{C}_{\frac{n}{2}}\mathcal{L}_{n-m}+\mathcal{C}_{\frac{n}{2}}\mathcal{L}_{\frac{n}{2}}\right)-\left(2\sum_{k=1}^{\frac{n}{2}-1}\mathcal{C}_k\mathcal{L}_{\frac{n}{2}}+\mathcal{C}_{\frac{n}{2}}\mathcal{L}_{\frac{n}{2}}\right),\tag{Line 2}\\
=&\sum_{m=1}^{\frac{n}{2}-1}\left(\mathcal{C}_{m}-\mathcal{C}_{\frac{n}{2}}\right)\left(\mathcal{L}_{m}+\mathcal{L}_{n-m}\right)-2\sum_{m=1}^{\frac{n}{2}-1}\mathcal{C}_m\mathcal{L}_{\frac{n}{2}}+2\sum_{m=1}^{\frac{n}{2}-1}\mathcal{C}_{\frac{n}{2}}\mathcal{L}_{\frac{n}{2}},\tag{Line 3}\\
=&\sum_{m=1}^{\frac{n}{2}-1}\left(\mathcal{C}_{m}-\mathcal{C}_{\frac{n}{2}}\right)\left(\mathcal{L}_{m}+\mathcal{L}_{n-m}-2\mathcal{L}_{\frac{n}{2}}\right)<0,
\end{align}
where Line 2 transforms the summation index from $m$ to $n-m$ while following $\mathcal{C}_m=\mathcal{C}_{n-m}$ in \eqref{properties}.  Line 3 changes the summation index from $k$ to $m$. The last line follows from $m<n/2$, which leads to $\mathcal{C}_m<\mathcal{C}_{n/2}$ and $\mathcal{L}_{n/2}-\mathcal{L}_m<\mathcal{L}_{n-m}-\mathcal{L}_{n/2}$ (because of the convexity of the first-order derivative of the logarithm function in \eqref{properties} and $m<n/2<n-m$). Therefore, \eqref{shortshort} holds true in this case.
\end{itemize}
In conclusion, \eqref{shortshort} always holds true, and this completes the proof for \eqref{short} and thus \eqref{long}.
\end{proof}
\end{landscape}

\section{Proof for Proposition \ref{equilibrium_extension_model}}\label{Proof:PBE_extension}

In this appendix, we focus on the seller's uniform pricing in Stage I and the users' social activity decisions in Stage II, and the analysis is conducted through backward induction (see Section \ref{flexibility_new_equilibrium}).

\subsection{Users' Social Activity Decisions in Stage II}

In Stage II, each user $i\in\mathcal{N}$ decides his social activity level $x_i$ given the seller's announced uniform price $p_0$ in Stage I. At a high level, this differs from the analysis of the original three-stage model in Section~\ref{Sec:Model}, where users make their social activity decisions by predicting the seller's uniform price $p_0$. However, the analysis is similar on the technical level, and thus Belief \ref{threshold_structure}, Theorem \ref{Theorem:p0-vstar}, and Proposition \ref{unique_user} remain true in this extended four-stage model.

\subsection{Seller's Uniform Pricing in Stage I}

In Stage I, the seller determines the uniform price $p_0$ to maximize her total expected revenue $\tilde{\Pi}$, including both the expected profiled revenue $\tilde{\Pi}_1$ and the expected non-profiled revenue $\tilde{\Pi}_0$. Considering the users' social activity strategies in Stage II, we now calculate the expected revenue based on \eqref{platform_sale_revenue}. Specifically, for the expected profiled revenue $\tilde{\Pi}_1$, we have
\begin{align}\label{extension_1}
\tilde{\Pi}_1\triangleq \mathbb{E}_{i\sim \mathcal{N}_1}\left(\sum_{i\in\mathcal{N}_1}v_i\right) =|\mathcal{N}|\text{Prob}\left[i\in\mathcal{N}_1\right]\cdot\mathbb{E}_{v_i\sim f(\cdot|i\in\mathcal{N}_1)}v_i=|\mathcal{N}|\frac{\delta(v^*)^2}{2\bar{v}},
\end{align}
where we take expectations over all possible profiled users $\mathcal{N}_1$; for the expected non-profiled revenue $\tilde{\Pi}_0$, we have
\begin{align}\label{extension_0}
\tilde{\Pi}_0
\triangleq \mathbb{E}_{i\sim \mathcal{N}_0}\left(\sum_{i\in\mathcal{N}_0}p_0d_i\right)=&|\mathcal{N}|\text{Prob}\left[i\in\mathcal{N}_0\right]\cdot\mathbb{E}_{v_i\sim f(\cdot|i\in\mathcal{N}_0)}\left(p_0\mathbbm{1}(v_i\ge p_0)\right),\nonumber\\
=&|\mathcal{N}|\left(1-\delta\frac{v^*}{\bar{v}}\right)p_0\left(1-\frac{1-\delta}{\bar{v}-\delta v^*}p_0\right),
\end{align}
where we take expectations over all possible non-profiled users $\mathcal{N}_0$. Notice that we will restrict our attention to the case when $p_0<v^*$ since this always holds (see \eqref{vstar_p0_summation} in Theorem \ref{Theorem:p0-vstar}). Therefore, the seller's total expected revenue $\tilde{\Pi}$ is
\begin{align}\label{object_flexibility}
\tilde{\Pi}\left(p_0,v^*(p_0)\right)\triangleq\tilde{\Pi}_1+\tilde{\Pi}_0=|\mathcal{N}|\left(\delta\frac{(v^*(p_0))^2}{2\bar{v}}+\left(1-\delta\frac{v^*(p_0)}{\bar{v}}\right)p_0\left(1-\frac{1-\delta}{\bar{v}-\delta v^*(p_0)}p_0\right)\right),
\end{align}
and the optimal uniform price is given by 
\begin{align}\label{opt_p0_extension}
p_0^*=\arg\max_{p_0\in[0,\bar{v}]}\tilde{\Pi}(p_0,v^*(p_0)).    
\end{align}
The remaining analysis proceeds in three steps.

\subsubsection{Step 1} First, we (partially) reformulate the seller's optimization objective $\tilde{\Pi}(p_0,v^*(p_0))$ with variable $p_0$ into one with variable~$v^*$, i.e., $\tilde{\Pi}(v^*)$. More specifically, we develop our analysis based on the following two possible cases concerning $p_0$.

\begin{itemize}
    \item [(i)] For $p_0\ge\bar{v}-1/\delta\cdot\ln\left(n-1+\omega_0\right)$, we have $v^*=\bar{v}$ (see the final part of Appendix \ref{proof_v^*}). The seller's total expected revenue in \eqref{object_flexibility} purely depends on $p_0$ as below:
    \begin{align}
       \tilde{\Pi}(p_0)=|\mathcal{N}|\left(\delta\frac{\bar{v}}{2}+\left(1-\delta\right)p_0\left(1-\frac{p_0}{\bar{v}}\right)\right), 
    \end{align}
    which is concave in $p_0$ over $[\bar{v}-1/\delta\cdot\ln\left(n-1+\omega_0\right),\bar{v}]$. Then the seller's sub-optimal uniform price over the closed interval $[\bar{v}-1/\delta\cdot\ln\left(n-1+\omega_0\right),\bar{v}]$ is given by
    \begin{align}
        p_0^\dagger=\max\left\{\frac{\bar{v}}{2},\bar{v}-\frac{1}{\delta}\cdot\ln\left(n-1+\omega_0\right)\right\},
    \end{align}
    and more specifically, we have:
    \begin{itemize}
        \item \mbox{\textbf{Condition I} with $\bar{v}/2>1/\delta\cdot\ln\left(n-1+\omega_0\right)$:} We have $p_0^\dagger=\bar{v}-1/\delta\cdot\ln\left(n-1+\omega_0\right)$. In addition, $\tilde{\Pi}(p_0)$ decreases with $p_0$ over $[\bar{v}-1/\delta\cdot\ln\left(n-1+\omega_0\right),\bar{v}]$.
        \item \mbox{\textbf{Condition II} with $\bar{v}/2\le1/\delta\cdot\ln\left(n-1+\omega_0\right)$:} We have $p_0^\dagger=\bar{v}/2$. In addition, $\tilde{\Pi}(p_0)$ increases with $p_0$ over $[\bar{v}-1/\delta\cdot\ln\left(n-1+\omega_0\right),\bar{v}/2]$ and decreases over $[\bar{v}/2,\bar{v}]$.
    \end{itemize}
    \item [(ii)] For $p_0<\bar{v}-1/\delta\cdot\ln\left(n-1+\omega_0\right)$, we have $v^*<\bar{v}$ (see the final part of Appendix \ref{proof_v^*}). In this case, we will reformulate the seller's optimization objective $\tilde{\Pi}(p_0,v^*(p_0))$ in \eqref{object_flexibility} with variable $p_0$ into $v^*$, i.e., $\tilde{\Pi}(v^*)$. The main reason for this is that $\tilde{\Pi}(p_0,v^*(p_0))$ in \eqref{object_flexibility} hinges on both $p_0$ and $v^*$ as well as their relationship outlined in Theorem \ref{Theorem:p0-vstar}. Notice that we cannot figure out the closed-form characterization of $v^*$ as a function of $p_0$, i.e., $v^*(p_0)$, but we can have $p_0=v^*-\tilde{J}(v^*)/\delta$ for $v^*<\bar{v}$. Moreover, we have proved that $v^*$ is uniquely determined given any $p_0$ (see Proposition \ref{unique_user}). One could also verify that $\partial v^*/\partial p_0> 0$ for $v^*<\bar{v}$ through the implicit function theorem (see the final part of Appendix \ref{proof_v^*}: $\partial \Phi (v^*,p_0)/\partial v^*<0$ and we also have $\partial \Phi (v^*,p_0)/\partial p_0>0$ in \eqref{Phi}). Therefore, a one-to-one mapping exists between $v^*$ and $p_0$ in this case according to the property of inverse functions. This then allows for our reformulation, and henceforth we can focus on finding $v^*=\arg\max \tilde{\Pi}(v^*)$, where $\tilde{\Pi}(v^*)$ is derived by combining \eqref{object_flexibility} with \eqref{vstar_p0_summation} in Theorem \ref{Theorem:p0-vstar}:
    \begin{align}\label{vstart_reformulate}
        \tilde{\Pi}\left(v^*\right)=|\mathcal{N}|\left(\delta\frac{(v^*)^2}{2\bar{v}}+\left(1-\delta\frac{v^*}{\bar{v}}\right)\left(v^*-\frac{1}{\delta}\tilde{J}(v^*)\right)-\frac{1-\delta}{\bar{v}}\left(v^*-\frac{1}{\delta}\tilde{J}(v^*)\right)^2\right).
    \end{align}
    Moreover, we will narrow the feasible range of $v^*$ to facilitate our analysis in subsequent steps. Since $\partial p_0/\partial v^*> 0$ and the seller always  sets the uniform price non-negative, i.e., $p_0\ge 0$, we then have: 
    \begin{align}\label{v_0}
    v^*\in[v^o,\bar{v}) \text{ with } v^o=\frac{\tilde{J}(v^o)}{\delta}.
    \end{align}
    Now the optimization problem for the seller in this case is $v^*=\arg\max_{v^*\in[v^o,\bar{v}]} \tilde{\Pi}(v^*)$
    However, this new optimization problem is still challenging to analyze because $\tilde{\Pi}\left(v^*\right)$ in \eqref{vstart_reformulate} may not be concave or convex. We will focus on dealing with this non-convex optimization problem in the next step.
\end{itemize}

\begin{landscape}
\subsubsection{Step 2}\label{step2-extension} This step focuses on the non-convex optimization problem $v^*=\arg\max\tilde{\Pi}(v^*)$ with $v^o\le v^*<\bar{v}$, reformulated for the case when $p_0<\bar{v}-1/\delta\cdot\ln\left(n-1+\omega_0\right)$. The main idea in this step is to show the unimodality of $\tilde{\Pi}(v^*)$ in \eqref{vstart_reformulate}. 

We first calculate the first-order derivative of $\tilde{\Pi}(v^*)$:
\begin{align}\label{FOC_Exten}
\frac{\partial \tilde{\Pi}\left(v^*\right)}{\partial v^*}=|\mathcal{N}|\left(1-\frac{1}{\delta}\tilde{J}'(v^*)+\frac{\delta-2}{\bar{v}}v^*+\frac{2-\delta}{\delta\bar{v}}\tilde{J}(v^*)+\frac{2-\delta}{\delta\bar{v}}v^*\tilde{J}'(v^*)-2\frac{1-\delta}{\delta^2\bar{v}}\tilde{J}(v^*)\tilde{J}'(v^*)\right),
\end{align}
and then the third-order derivative:
\begin{align}
\frac{\partial^3 \tilde{\Pi}\left(v^*\right)}{\partial (v^*)^3} 
=&|\mathcal{N}|\left(-\frac{1}{\delta}\tilde{J}'''(v^*)+3\frac{2-\delta}{\delta\bar{v}}\tilde{J}''(v^*)+\frac{2-\delta}{\delta\bar{v}}v^*\tilde{J}'''(v^*)-6\frac{1-\delta}{\delta^2\bar{v}}\tilde{J}'(v^*)\tilde{J}''(v^*)-2\frac{1-\delta}{\delta^2\bar{v}}\tilde{J}(v^*)\tilde{J}'''(v^*)\right),\nonumber\\
=&|\mathcal{N}|\left(-\frac{1}{\delta}\tilde{J}'''(v^*)+\frac{2-2\delta}{\delta\bar{v}}\left(\left(v^*\tilde{J}'''(v^*)-\frac{\tilde{J}(v^*)}{\delta}\tilde{J}'''(v^*)\right)-3\left(\frac{\tilde{J}'(v^*)}{\delta}\tilde{J}''(v^*)-\tilde{J}''(v^*)\right)\right)+\frac{1}{\bar{v}}\left(3\tilde{J}''(v^*)+v^*\tilde{J}'''(v^*)\right)\right),\nonumber\\
=&|\mathcal{N}|\left(-\frac{1}{\delta}\tilde{J}'''(v^*)+\frac{2-2\delta}{\delta\bar{v}}\left(-v^*\tilde{J}'''(v^*)(\frac{\tilde{J}(v^*)}{\delta v^*}-1)-3\tilde{J}''(v^*)(\frac{\tilde{J}'(v^*)}{\delta}-1)\right)+\frac{1}{\bar{v}}\left(3\tilde{J}''(v^*)+v^*\tilde{J}'''(v^*)\right)\right),
\end{align}
where $\partial^2 \tilde{J}(v^*)/\partial (v^*)^2\equiv \tilde{J}''(v^*)$ and $\partial^3 \tilde{J}(v^*)/\partial (v^*)^3\equiv \tilde{J}'''(v^*)$. Lemma \ref{3_negative} below shows that $\partial^3 \tilde{\Pi}\left(v^*\right)/\partial {v^*}^3<0$ always holds. Thus, $\partial \tilde{\Pi}\left(v^*\right)/\partial v^*$ is concave in $v^*$ over $[v^o,\bar{v})$. Moreover, we have that at $v^*=v^o$:
\begin{align}
\frac{\partial \tilde{\Pi}\left(v^*\right)}{\partial v^*}|_{v^*=v^o}&=|\mathcal{N}|\left(1-\frac{1}{\delta}\tilde{J}'(v^o)+\frac{2-\delta}{\delta^2\bar{v}}\tilde{J}(v^o)\tilde{J}'(v^o)-2\frac{1-\delta}{\delta^2\bar{v}}\tilde{J}(v^o)\tilde{J}'(v^o)\right),\tag{Line 1}\\
&=|\mathcal{N}|\left(1-\frac{1}{\delta}\tilde{J}'(v^o)+\frac{1}{\delta\bar{v}}\tilde{J}(v^o)\tilde{J}'(v^o)\right)>0,
\end{align}
where Line 1 applies the fact in \eqref{v_0}. The last inequality follows from Lemma \ref{claim:J_concave} and Lemma \ref{jprime}. At $v^*=\bar{v}$, we have
\begin{align}\label{bar_bar}
\frac{\partial \tilde{\Pi}\left(v^*\right)}{\partial v^*}|_{v^*=\bar{v}}=&|\mathcal{N}|\left(1-\frac{1}{\delta}\tilde{J}'(\bar{v})+\delta-2+\frac{2-\delta}{\delta\bar{v}}\tilde{J}(\bar{v})+\frac{2-\delta}{\delta}\tilde{J}'(\bar{v})-2\frac{1-\delta}{\delta^2\bar{v}}\tilde{J}(\bar{v})\tilde{J}'(\bar{v})\right),   
\end{align}
where $\tilde{J}(\bar{v})=\ln(n-1+\omega_0)$ independent of $\bar{v}$ and $\tilde{J}'(\bar{v})=(n-1)/\bar{v}\cdot\left(\ln(n-1+\omega_0)-\ln(n-2+\omega_0)\right)$ according to \eqref{FO_J} and \eqref{SO_J}. According to Fact \ref{Fact1}, we have:
\begin{itemize}
	\item \mbox{\textbf{Condition A} with $\partial \tilde{\Pi}\left(v^*\right)/\partial v^*|_{v^*=\bar{v}}\ge0$: }We have $\partial \tilde{\Pi}\left(v^*\right)/\partial v^*>0$ over $[v^o,\bar{v})$. Thus, $\tilde{\Pi}\left(v^*\right)$ increases with $v^*$ over $[v^o,\bar{v})$.
	\item {\textbf{Condition B} with $\partial \tilde{\Pi}\left(v^*\right)/\partial v^*|_{v^*=\bar{v}}<0$: } There exists a unique solution $v^\ddag$ to $\partial \tilde{\Pi}\left(v^\ddag\right)/\partial v^\ddag=0$, Then, $\tilde{\Pi}\left(v^*\right)$ increases with $v^*$ over $[v^o, v^\ddag]$ and decreases over $[v^\ddag,\bar{v})$.
\end{itemize}

\begin{lemma}\label{3_negative}
For $v^*\in[0,\bar{v}]$, $\partial^3 \tilde{\Pi}\left(v^*\right)/\partial {v^*}^3<0$.
\end{lemma}
\begin{proof}
Consider the following facts:
\begin{itemize}
    \item [(i)] Lemma \ref{TO_J} indicates that $\tilde{J}'''(v^*)>0$. 
    \item [(ii)] Lemma \ref{J222} indicates that $\tilde{J}(v^*)/v^*>\tilde{J}'(v^*)$, and thus $\tilde{J}(v^*)/\delta v^*-1>\tilde{J}'(v^*)/\delta -1$. In addition, since $\tilde{J}'(v^*)<\delta$ (see Lemma \ref{jprime}) and $0<\tilde{J}(v^*)\le\delta v^*$ (i.e., $p_0^*\ge0$), we have $0\ge \tilde{J}(v^*)/\delta v^*-1>\tilde{J}'(v^*)/\delta -1$.
    \item [(iii)] Lemma \ref{J333} indicates that $3\tilde{J}''(v^*)+v^*\tilde{J}'''(v^*)<0$. In addition, since $\tilde{J}''(v^*)<0$ and $\tilde{J}'''(v^*)>0$, we have $3\tilde{J}''(v^*)<-v^*\tilde{J}'''(v^*)<0$.
\end{itemize}
To sum up, we have $\tilde{J}'''(v^*)>0$, $(-v^*\tilde{J}'''(v^*))(\tilde{J}(v^*)/\delta v^*-1)<3\tilde{J}''(v^*)(\tilde{J}'(v^*)/\delta -1)$, and $3\tilde{J}''(v^*)+v^*\tilde{J}'''(v^*)<0$, which result in $\partial^3 \tilde{\Pi}\left(v^*\right)/\partial {v^*}^3<0$.


\end{proof}

\begin{lemma}\label{TO_J}
For $v^*\in[0,\bar{v}]$, $\tilde{J}'''(v^*)>0$.   
\end{lemma}
\begin{proof}
Calculating the third-order derivative of $\tilde{J}(v^*)$ in \eqref{FO_J} yields
\begin{align}
\frac{\partial^3 \tilde{J}(v^*)}{\partial (v^*)^3}	
=&\sum_{m=3}^{n-1}\frac{(n-1)!}{(n-1-m)!(m-3)!}F(v^*)^{m-3}\left(1-F(v^*)\right)^{n-1-m}\left(f(v^*)\right)^3\left(\ln(m+\omega_0)-2\ln(m-1+\omega_0)+\ln(m-2+\omega_0)\right)\tag{Line 1-1}\\&\quad-\sum_{m=2}^{n-2}\frac{(n-1)!}{(n-2-m)!(m-2)!}F(v^*)^{m-2}\left(1-F(v^*)\right)^{n-2-m}\left(f(v^*)\right)^3\left(\ln(m+\omega_0)-2\ln(m-1+\omega_0)+\ln(m-2+\omega_0)\right),\tag{Line 1-2}\\
=&\sum_{m=3}^{n-1}\frac{(n-1)!}{(n-1-m)!(m-3)!}F(v^*)^{m-3}\left(1-F(v^*)\right)^{n-1-m}\left(f(v^*)\right)^3\left(\ln(m+\omega_0)-2\ln(m-1+\omega_0)+\ln(m-2+\omega_0)\right)\tag{Line 2-1}\\&\quad-\sum_{m=3}^{n-1}\frac{(n-1)!}{(n-1-m)!(m-3)!}F(v^*)^{m-3}\left(1-F(v^*)\right)^{n-1-m}\left(f(v^*)\right)^3\left(\ln(m-1+\omega_0)-2\ln(m-2+\omega_0)+\ln(m-3+\omega_0)\right),\tag{Line 2-2}\\
=&\sum_{m=3}^{n-1}\frac{(n-1)!}{(n-1-m)!(m-3)!}F(v^*)^{m-3}\left(1-F(v^*)\right)^{n-1-m}\left(f(v^*)\right)^3\left(\left(\mathcal{L}_m-\mathcal{L}_{m-1}\right)-\left(\mathcal{L}_{m-1}-\mathcal{L}_{m-2}\right)\right)>0.
\end{align}
Notice that $\partial f(v^*)/\partial v^* = 0$ given $f(v^*) = 1/\bar{v}$ for the uniform valuation distribution. This fact is applied in Line 1. Line 2-2 transforms the summation index from $m$ to $m-1$, and $\mathcal{L}_m$ is defined in \eqref{mathcal_L_C}. The last inequality follows from the convexity of the first-order derivative of the logarithm function. This completes the proof.
\end{proof}

\begin{lemma}\label{J222}
For $v^*\in[0,\bar{v}]$, $\tilde{J}(v^*)/v^*>\tilde{J}'(v^*)$.   
\end{lemma}
\begin{proof}
\begin{align}
\frac{\tilde{J}(v^*)}{v^*}-\tilde{J}'(v^*)=&\sum_{m=0}^{n-1}\frac{(n-1)!}{(n-1-m)!m!}\ln(m+\omega_0)\frac{(v^*)^{m-1}}{\bar{v}^m}\left(1-\frac{v^*}{\bar{v}}\right)^{n-1-m}\nonumber\\&\qquad-\sum_{m=1}^{n-1}\frac{(n-1)!}{(n-1-m)!(m-1)!}\left(\ln(m+\omega_0)-\ln(m-1+\omega_0)\right)\frac{(v^*)^{m-1}}{\bar{v}^m}\left(1-\frac{v^*}{\bar{v}}\right)^{n-1-m},\nonumber\\
=&\sum_{m=0}^{n-1}\frac{(n-1)!}{(n-1-m)!m!}\ln(m+\omega_0)\frac{(v^*)^{m-1}}{\bar{v}^m}\left(1-\frac{v^*}{\bar{v}}\right)^{n-1-m}\nonumber\\&\qquad-\sum_{m=0}^{n-1}\frac{(n-1)!}{(n-1-m)!m!}m\left(\ln(m+\omega_0)-\ln(m-1+\omega_0)\right)\frac{(v^*)^{m-1}}{\bar{v}^m}\left(1-\frac{v^*}{\bar{v}}\right)^{n-1-m},\nonumber\\
=&\sum_{m=0}^{n-1}\frac{(n-1)!}{(n-1-m)!m!}\frac{(v^*)^{m-1}}{\bar{v}^m}\left(1-\frac{v^*}{\bar{v}}\right)^{n-1-m}\left(\ln(m+\omega_0)-m\left(\ln(m+\omega_0)-\ln(m-1+\omega_0)\right)\right),\nonumber\\
=&\sum_{m=0}^{n-1}\frac{(n-1)!}{(n-1-m)!m!}\frac{(v^*)^{m-1}}{\bar{v}^m}\left(1-\frac{v^*}{\bar{v}}\right)^{n-1-m}\left(m\ln(m-1+\omega_0)-(m-1)\ln(m+\omega_0)\right)>0.
\end{align}
The last inequality follows from the fact that $m\ln(m-1+\omega_0)>(m-1)\ln(m+\omega_0)$ holds for any $m\in\{0,\dots,n-1\}$ with $\omega_0>1$. In particular, this is easily satisfied for $m=0$ and $m=1$. For $m>1$, this fact could be verified by showing that the function $\ln(x+\omega_0)/x$ decreases in $x$ over $[0,\infty)$.
\end{proof}

\begin{lemma}\label{J333}
For $v^*\in[0,\bar{v}]$, $3\tilde{J}''(v^*)+v^*\tilde{J}'''(v^*)<0$.   
\end{lemma}
\begin{proof}
\begin{align}
3\tilde{J}''(v^*)+v^*\tilde{J}'''(v^*)=&\sum_{m=2}^{n-1}\frac{(n-1)!}{(n-1-m)!(m-2)!}\frac{(v^*)^{m-2}}{\bar{v}^{m}}\left(1-\frac{v^*}{\bar{v}}\right)^{n-1-m}\left(3\mathcal{L}_m-3\mathcal{L}_{m-1}\right)\nonumber\\&\qquad+\sum_{m=3}^{n-1}\frac{(n-1)!}{(n-1-m)!(m-3)!}\frac{(v^*)^{m-2}}{\bar{v}^{m}}\left(1-\frac{v^*}{\bar{v}}\right)^{n-1-m}\left(\left(\mathcal{L}_m-\mathcal{L}_{m-1}\right)-\left(\mathcal{L}_{m-1}-\mathcal{L}_{m-2}\right)\right),\nonumber\\
=&\sum_{m=2}^{n-1}\frac{(n-1)!}{(n-1-m)!(m-2)!}\frac{(v^*)^{m-2}}{\bar{v}^{m}}\left(1-\frac{v^*}{\bar{v}}\right)^{n-1-m}\left(3\mathcal{L}_m-3\mathcal{L}_{m-1}\right)\nonumber\\&\qquad+\sum_{m=2}^{n-1}\frac{(n-1)!}{(n-1-m)!(m-2)!}\frac{(v^*)^{m-2}}{\bar{v}^{m}}\left(1-\frac{v^*}{\bar{v}}\right)^{n-1-m}(m-2)\left(\left(\mathcal{L}_m-\mathcal{L}_{m-1}\right)-\left(\mathcal{L}_{m-1}-\mathcal{L}_{m-2}\right)\right),\nonumber\\
=&\sum_{m=2}^{n-1}\frac{(n-1)!}{(n-1-m)!(m-2)!}\frac{(v^*)^{m-2}}{\bar{v}^{m}}\left(1-\frac{v^*}{\bar{v}}\right)^{n-1-m}\left((m+1)\left(\mathcal{L}_m-\mathcal{L}_{m-1}\right)-(m-2)\left(\mathcal{L}_{m-1}-\mathcal{L}_{m-2}\right)\right),\nonumber\\
=&\sum_{m=2}^{n-1}\frac{(n-1)!}{(n-1-m)!(m-2)!}\frac{(v^*)^{m-2}}{\bar{v}^{m}}\left(1-\frac{v^*}{\bar{v}}\right)^{n-1-m}\nonumber\\&\qquad\cdot\left(\left((m+1)\ln(m+\omega_0)-(m-2)\ln(m-3+\omega_0)\right)-3\left(m\ln(m-1+\omega_0)-(m-1)\ln(m-2+\omega_0)\right)\right)<0.
\end{align}
The last inequality follows from the fact that
\begin{align*}
(m+1)\ln(m+\omega_0)-(m-2)\ln(m-3+\omega_0)<3\left(m\ln(m-1+\omega_0)-(m-1)\ln(m-2+\omega_0)\right)    
\end{align*}
holds for any $m\in\{2,\dots,n-1\}$ with $\omega_0>1$. In particular, this is satisfied for $m=2$. For $m>2$, this fact follows from the concavity of the first-order derivative of the function $x\ln(x-1+\omega_0)$ over $[0,\infty)$.
\end{proof}
				
\end{landscape}
\subsubsection{Step 3}\label{step3-extension} Finally, we jointly consider Conditions I and II for $p_0\ge\bar{v}-1/\delta\cdot\ln\left(n-1+\omega_0\right)$ and Conditions A and B for $p_0<\bar{v}-1/\delta\cdot\ln\left(n-1+\omega_0\right)$ to characterize the seller's optimal uniform price $p_0^*$ in \eqref{opt_p0_extension}.

We start with checking the contradiction between Condition II and Condition B. Specifically, when Condition II holds, i.e. $\bar{v}/2\le1/\delta\cdot\ln\left(n-1+\omega_0\right)=\tilde{J}(\bar{v})/\delta$, we have $\partial \tilde{\Pi}\left(v^*\right)/\partial v^*|_{v^*=\bar{v}}$ in \eqref{bar_bar}:
\begin{align}
\frac{\partial \tilde{\Pi}\left(v^*\right)}{\partial v^*}|_{v^*=\bar{v}}=&|\mathcal{N}|\left(1-\frac{1}{\delta}\tilde{J}'(\bar{v})+\delta-2+\frac{2-\delta}{\delta\bar{v}}\tilde{J}(\bar{v})+\frac{2-\delta}{\delta}\tilde{J}'(\bar{v})-2\frac{1-\delta}{\delta^2\bar{v}}\tilde{J}(\bar{v})\tilde{J}'(\bar{v})\right),\nonumber\\
=&|\mathcal{N}|\left((\frac{1}{\delta}\tilde{J}'(\bar{v})-1)\left((2-\delta)\left(1-\frac{\tilde{J}(\bar{v})}{\delta\bar{v}}\right)-1\right)+\frac{1}{\delta\bar{v}}\tilde{J}(\bar{v})\tilde{J}'(\bar{v})\right)>0.
\end{align}
The last inequality holds because $\tilde{J}'(\bar{v})<\tilde{J}'(v^*)<\delta$ (see Lemmas \ref{claim:J_concave} and \ref{jprime})
and $(2-\delta)(1-\tilde{J}(\bar{v})/\delta\bar{v})<(2-\delta)/2<1$ following Condition II. Therefore, only Condition A (but not Condition B) can be satisfied when Condition II holds. Thus, we will have the following three possible cases:
\begin{itemize}
    \item \mbox{Case I in Proposition \ref{equilibrium_extension_model}:} This is the case when only Condition II holds, i.e. $\bar{v}/2\le1/\delta\cdot\ln\left(n-1+\omega_0\right)$; then Condition A is naturally satisfied. Recall that $\partial v^*/\partial p_0> 0$; therefore, we have $\tilde{\Pi}(p_0)$ in \eqref{object_flexibility} increases with $p_0$ over $[0,\bar{v}-1/\delta\cdot\ln\left(n-1+\omega_0\right))$, and continues to increase over $[\bar{v}-1/\delta\cdot\ln\left(n-1+\omega_0\right),\bar{v}/2]$ and decreases over $[\bar{v}/2,\bar{v}]$. Accordingly, we have the unique optimal uniform price $p_0^*=\bar{v}/2$.
    \item \mbox{Case II in Proposition \ref{equilibrium_extension_model}:} This is the case when both Condition I and Condition A hold. Then we have $\tilde{\Pi}(p_0)$ in \eqref{object_flexibility} increases with $p_0$ over $[0,\bar{v}-1/\delta\cdot\ln\left(n-1+\omega_0\right))$ and decreases over $[\bar{v}-1/\delta\cdot\ln\left(n-1+\omega_0\right),\bar{v}]$. Accordingly, we have the unique optimal uniform price $p_0^*=\bar{v}-1/\delta\cdot\ln\left(n-1+\omega_0\right)$, which is higher than $\bar{v}/2$ (see Condition I).
    \item \mbox{Case III in Proposition \ref{equilibrium_extension_model}:} This is the case when both Condition I and Condition B hold. Then we have $\tilde{\Pi}(p_0)$ in \eqref{object_flexibility} increases with $p_0$ over $[0,p_0(v^\ddag)]$ and decreases over $[p_0(v^\ddag),\bar{v}-1/\delta\cdot\ln\left(n-1+\omega_0\right))$, and continues to decrease over $[\bar{v}-1/\delta\cdot\ln\left(n-1+\omega_0\right),\bar{v}]$. Accordingly, we have the unique optimal uniform price $p_0(v^\ddag)=v^\ddag-\tilde{J}(v^\ddag)/\delta$, and $v^\ddag<\bar{v}$ is the unique solution to $\partial \tilde{\Pi}\left(v^\ddag\right)/\partial v^\ddag=0$ (see Condition B).
\end{itemize}
In the end, we will characterize the mean valuation threshold $\tilde{v}(\delta)$ indifferent between Case II and Case III in Proposition \ref{equilibrium_extension_model}.\footref{footnote6} Equivalently, we will focus on $\partial \tilde{\Pi}\left(v^*\right)/\partial v^*|_{v^*=\bar{v}}$ for Conditions A and B when $\bar{v}/2>1/\delta\cdot\ln\left(n-1+\omega_0\right)$ (see Condition I). Notice that $\partial \tilde{\Pi}\left(v^*\right)/\partial v^*|_{v^*=\bar{v}}$ in \eqref{bar_bar} is indeed a function of $\bar{v}$; for ease of exposition, denote $M(\bar{v})\triangleq\partial \tilde{\Pi}\left(v^*\right)/\partial v^*|_{v^*=\bar{v}}$, i.e.,
\begin{align}\label{M_vbar}
    M(\bar{v})&=|\mathcal{N}|\left(1-\frac{1}{\delta}\tilde{J}'(\bar{v})+\delta-2+\frac{2-\delta}{\delta\bar{v}}\tilde{J}(\bar{v})+\frac{2-\delta}{\delta}\tilde{J}'(\bar{v})-2\frac{1-\delta}{\delta^2\bar{v}}\tilde{J}(\bar{v})\tilde{J}'(\bar{v})\right),\nonumber\\
    &=|\mathcal{N}|\left(\delta-1+\frac{2-\delta}{\delta\bar{v}}\tilde{J}+\frac{1-\delta}{\delta\bar{v}}\tilde{L}'-2\frac{1-\delta}{\delta^2\bar{v}^2}\tilde{J}\tilde{L}'\right),
\end{align}
where both $\tilde{J}\equiv\tilde{J}(\bar{v})$ and $\tilde{L}'\equiv (n-1)\left(\ln(n-1+\omega_0)-\ln(n-2+\omega_0)\right)$ are constants independent of $\bar{v}$. Lemma \ref{M_decreasing} shows that $M(\bar{v})$ decreases with $\bar{v}$ over $[2\tilde{J}/\delta,+\infty)$. In addition, we have $M(2\tilde{J}/\delta)=|\mathcal{N}|\delta/2>0$ and $\lim_{\bar{v}\to\infty}=|\mathcal{N}|(\delta-1)<0$. Hence, there exists a unique solution $\tilde{v}'$ to $M(\bar{v})=0$ over $[2\tilde{J}/\delta,+\infty)$ such that $M(\bar{v})\ge 0$ (Condition A) for $2\tilde{J}/\delta\le \bar{v}\le \tilde{v}'$ whereas $M(\bar{v})<0$ (Condition B) for $\bar{v}>\tilde{v}'$. In other words, the mean valuation threshold $\tilde{v}(\delta)$ in Proposition \ref{equilibrium_extension_model} is given by $\tilde{v}(\delta)=\tilde{v}'/2$.

This completes the proof for Proposition \ref{equilibrium_extension_model}.

\begin{lemma}\label{M_decreasing}
$M(\bar{v})$ in \eqref{M_vbar} decreases with $\bar{v}$ over $[2\tilde{J}/\delta,+\infty)$. 
\end{lemma}
\begin{proof}
Calculating the first-order derivative of $M(\bar{v})$ yields
\begin{align}
    \frac{\partial M(\bar{v})}{\partial \bar{v}}=\frac{|\mathcal{N}|}{\bar{v}^3}\left(\left((\delta-1)\tilde{L}'-(2-\delta)\tilde{J}\right)\frac{\bar{v}}{\delta}+4\frac{1-\delta}{\delta^2}\tilde{J}\tilde{L}'\right),
\end{align}
where we have
\begin{align}
\left((\delta-1)\tilde{L}'-(2-\delta)\tilde{J}\right)\frac{\bar{v}}{\delta}+4\frac{1-\delta}{\delta^2}\tilde{J}\tilde{L}'\le&\left(\left((\delta-1)\tilde{L}'-(2-\delta)\tilde{J}\right)\frac{\bar{v}}{\delta}+4\frac{1-\delta}{\delta^2}\tilde{J}\tilde{L}'\right)|_{\bar{v}=2\frac{\tilde{J}}{\delta}}
,\nonumber\\=&\frac{2\tilde{J}\tilde{L}'}{\delta^2}\left((1-\delta)-(2-\delta)\frac{\tilde{J}}{\tilde{L}'}\right)<0
\end{align}
always holds over $[2\tilde{J}/\delta,+\infty)$. The last inequality follows from $\tilde{J}>\tilde{L}'$ (i.e., verified by showing that the function $\ln(x+\omega_0)/x$ decreases in $x$). Therefore, we have $\partial M(\bar{v})/\partial\bar{v}<0$ always holds to complete our proof.
\end{proof}

\section{Proof for Corollary \ref{extension_compar}}
We develop our analysis based on the three cases in Proposition \ref{equilibrium_extension_model}.
\begin{itemize}
    \item \mbox{Case I with $\bar{v}/2\le1/\delta\cdot\ln\left(n-1+\omega_0\right)$:} We have $v^*=v^e=\bar{v}$ and $p_0^*=p_0^e=\bar{v}/2$ (see Proposition \ref{equilibrium_extension_model} and Theorem \ref{Theorem_PBE}).

    \item \mbox{Case II with $1/\delta\cdot\ln\left(n-1+\omega_0\right)<\bar{v}/2\le\tilde{v}(\delta)$:} We have $v^*<\bar{v}=v^e$ and
    \begin{align}
    p_0^*=v^*-\frac{1}{\delta}\tilde{J}(v^*)<\bar{v}-\frac{1}{\delta}\tilde{J}(\bar{v})=p_0^e,   
    \end{align}
    which follows from the fact that $\partial p_0^*(v^*)/\partial v^*> 0$ for $v^*<\bar{v}$.

    \item \mbox{Case III with $\bar{v}/2>\tilde{v}(\delta)$:} Theorem \ref{Theorem_PBE} suggests that $v^*$ satisfies \eqref{vstar_PBE} or equivalently, \eqref{case2_equ} in this case. We then substitute \eqref{case2_equ} into \eqref{FOC_Exten} and simplify it to yield
    \begin{align}
        \frac{\partial \tilde{\Pi}(v)}{\partial v}|_{v=v^*}=|\mathcal{N}|\left(1+\frac{\delta-2}{\bar{v}}v^*+\frac{2-\delta}{\delta\bar{v}}\tilde{J}(v^*)\right)=|\mathcal{N}|\frac{\tilde{J}(v^*)}{\bar{v}}>0.
    \end{align}
    Furthermore, Proposition \ref{equilibrium_extension_model} indicates that $\partial \tilde{\Pi}(v)/\partial v|_{v=v^e}=0$ (see Condition B in Appendix \ref{Proof:PBE_extension}). Recall the concavity of $\partial \tilde{\Pi}(v)/\partial v$ and that $\partial \tilde{\Pi}(v)/\partial v|_{v=\bar{v}}<0<\partial \tilde{\Pi}(v)/\partial v|_{v=v^o}$ (see Appendix \ref{Proof:PBE_extension}), we thus have $v^*<v^e$ according to Fact~\ref{Fact1}. Then $p_0^*<p_0^e$ follows from the fact that $\partial p_0(v)/\partial v> 0$ for $v<\bar{v}$.
\end{itemize}
This then completes the proof.

\section{Proof for Proposition \ref{sale_revenue_compare}}
In the extended four-stage model, the seller sets the uniform price $p_0^e$ to maximize the total expected revenue $\tilde{\Pi}$ in \eqref{object_flexibility}, i.e., $p_0^e=\arg\max_{p_0\in[0,\bar{v}]}\tilde{\Pi}(p_0,v^*(p_0))$. Also, notice that $p_0^e$ is uniquely determined (see Appendix \ref{step3-extension}). It then follows that $\tilde{\Pi}^e\ge \tilde{\Pi}^*$ since $p_0^e\ge p_0^*$. Hereafter, we use the superscript $e$ to refer to the extended four-stage model and $*$ for the original three-stage model. Then we turn to the expected profiled revenue $\tilde{\Pi}_1$ in \eqref{extension_1}, and we will have $\tilde{\Pi}_1^e\ge\tilde{\Pi}_1^*$ since $v^e\ge v^*$. For the expected non-profiled revenue $\tilde{\Pi}_0$ in \eqref{extension_0}, we have $\tilde{\Pi}_0^e=\tilde{\Pi}_0^*$ due to $v^e=v^*$ and $p_0^e=p_0^*$ if $\bar{v}/2\le1/\delta\cdot\ln\left(n-1+\omega_0\right)$ (Case I of Proposition \ref{equilibrium_extension_model}). Otherwise, we have $\partial^3 \tilde{\Pi}_0(v)/\partial v^3=\partial^3 \tilde{\Pi}(v)/\partial v^3<0$ according to Lemma \ref{3_negative}. Thus, $\partial \tilde{\Pi}_0(v)/\partial v$ is concave in $v$ over $[v^o,\bar{v})$ and $\partial \tilde{\Pi}_0(v)/\partial v=\partial \tilde{\Pi}(v)/\partial v-\partial \tilde{\Pi}_1(v)/\partial v=\partial \tilde{\Pi}(v)/\partial v-|\mathcal{N}|\delta v/\bar{v}$. Moreover, we have that at $v=v^o$:
\begin{align}
\frac{\partial \tilde{\Pi}_0(v)}{\partial v}|_{v=v^o}=|\mathcal{N}|\left(1-\frac{1}{\delta}\tilde{J}'(v^o)+\frac{1}{\delta\bar{v}}\tilde{J}(v^o)\tilde{J}'(v^o)-\delta\frac{v^o}{\bar{v}}\right)=|\mathcal{N}|\left(1-\frac{1}{\delta}\tilde{J}'(v^o)\right)\left(1-\delta\frac{v^o}{\bar{v}}\right)>0,
\end{align}
where the first equality applies the fact in \eqref{v_0}. We now further analyse based on the conditions for Case II and Case III in Proposition \ref{equilibrium_extension_model}, respectively.
\begin{itemize}
    \item Consider Condition B (Case III of Proposition \ref{equilibrium_extension_model}), we have $\partial \tilde{\Pi}_0(v)/\partial v|_{v=v^e}=0-|\mathcal{N}|\delta v^e/\bar{v}<0$ at $v=v^e$. According to Fact \ref{Fact1}, there exists a unique solution $v^m$ to $\partial \tilde{\Pi}_0(v^m)/\partial v^m=0$ such that if and only if $v\in(v^m,v^e]$ we will have $\partial \tilde{\Pi}_0(v)/\partial v<0$ and thus $\tilde{\Pi}_0(v)$ decreases with $v$. We also have for $v^*<v^e<\bar{v}$:
    \begin{align}\label{vstar_FOC_0}
    \frac{\partial \tilde{\Pi}_0(v)}{\partial v}|_{v=v^*}= |\mathcal{N}|\left(1-\frac{2}{\bar{v}}p_0^*-\frac{1}{\bar{v}}\tilde{J}(v^*)\right)<0,
    \end{align}
    where the first equality applies \eqref{vstar_PBE} and the last inequality follows from $p_0^*>\bar{v}/2$. Thus, we have $v^m<v^*<v^e$ and thus $\tilde{\Pi}_0^*>\tilde{\Pi}_0^e$.
    \item Consider Condition A (Case II of Proposition \ref{equilibrium_extension_model}), since \eqref{vstar_FOC_0} holds for $v^*<\bar{v}=v^e$, we then similarly have $\partial \tilde{\Pi}_0(v)/\partial v<0$ if and only if $v\in(v^m,v^*]$. Furthermore, due to the concavity of $\partial \tilde{\Pi}_0(v)/\partial v$ over $[v^o,\bar{v})$, we have $\partial \tilde{\Pi}_0(v)/\partial v<0$ for $v\in(v^m,\bar{v}]$. Namely, $\tilde{\Pi}_0(v)$ decreases with $v$ over $(v^m,\bar{v}]$. Therefore, we have $\tilde{\Pi}_0^*>\tilde{\Pi}_0^e$ with $v^*<\bar{v}=v^e$.
\end{itemize}

Finally, we derive the upper bound of the improvement ratio of the total expected revenue $\tilde{\Pi}$ in \eqref{object_flexibility}. We hereafter focus on the case when $\bar{v}/2>1/\delta\cdot\ln\left(n-1+\omega_0\right)$, since the improvement ratio will be $0$ otherwise (i.e., $\tilde{\Pi}^e=\tilde{\Pi}^*$). Recall that in the original three-stage model, we have $p_0^*=(\bar{v}-\delta v^*)/2(1-\delta)\equiv v^*-\tilde{J}(v^*)/\delta$ (see Appendix \ref{proof_TheoremPBE}.(3)). Thus, we substitute $p_0^*=(\bar{v}-\delta v^*)/2(1-\delta)$ into $\tilde{\Pi}$ in \eqref{object_flexibility} to yield:
\begin{align}\label{bound}
\tilde{\Pi}^*(v^*)=|\mathcal{N}|\frac{\bar{v}^2-2\delta \bar{v}v^*+(2-\delta)\delta (v^*)^2}{4(1-\delta)\bar{v}}=\frac{|\mathcal{N}|}{4(1-\delta)\bar{v}}\cdot\left((2-\delta)\delta\left(v^*-\frac{\bar{v}}{2-\delta}\right)^2+\frac{2(1-\delta)}{2-\delta}\bar{v}^2\right).
\end{align}
Since $\bar{v}/(2-\delta)<v^*<\bar{v}$ (see Appendix \ref{proof_TheoremPBE}.(3)), we then have
\begin{align}
\tilde{\Pi}^*(v^*)\ge\min_{v^*\in\left(\frac{\bar{v}}{2-\delta},\bar{v}\right)}\tilde{\Pi}^*(v^*)>\tilde{\Pi}^*\left(\frac{\bar{v}}{2-\delta}\right)=\frac{|\mathcal{N}|\bar{v}}{2(2-\delta)}.
\end{align}
For the extended four-stage model, we have
\begin{align}
\tilde{\Pi}^e\left(p_0^e,v^e\right)=&|\mathcal{N}|\left(\delta\frac{(v^e)^2}{2\bar{v}}+\frac{1-\delta}{\bar{v}}p_0^e\left(\frac{\bar{v}-\delta v^e}{1-\delta}-p_0^e\right)\right)
<\max_{p_0^e}\tilde{\Pi}^e\left(p_0^e;v^e\right),\nonumber\\
=&\tilde{\Pi}^e\left(\frac{\bar{v}-\delta v^e}{2(1-\delta)},v^e\right)
=\tilde{\Pi}^*(v^e)
\le\max_{v^e\in\left(\frac{\bar{v}}{2-\delta},\bar{v}\right]}\tilde{\Pi}^*(v^e)
=\tilde{\Pi}^*(\bar{v})
=\frac{|\mathcal{N}|\bar{v}(1+\delta)}{4},
\end{align}
where the first equality follows from \eqref{object_flexibility}, and $\tilde{\Pi}^*(\cdot)$ is given by \eqref{bound} where we have substituted $p_0^e=(\bar{v}-\delta v^e)/2(1-\delta)$ into \eqref{object_flexibility}. Also one can show that $v^e>\bar{v}/(2-\delta)$ by verifying that
\begin{align}
    \frac{\partial \tilde{\Pi}(v)}{\partial v}|_{v=\frac{\bar{v}}{2-\delta}}=|\mathcal{N}|\frac{1}{\delta\bar{v}}\tilde{J}(\frac{\bar{v}}{2-\delta})\left((2-\delta)\left(1-\frac{1}{\delta}\tilde{J}'(\frac{\bar{v}}{2-\delta})\right)+\tilde{J}'(\frac{\bar{v}}{2-\delta})\right)>0.
\end{align}

Therefore, the improvement ratio is upper bounded by $\left(\tilde{\Pi}^e/\tilde{\Pi}^*-1\right)<(\delta-\delta^2)/2$. This then completes the proof.

\section{Proof for Proposition \ref{Theorem:PBE_general}}\label{Proposition_9_proof}
To start with, we first analyze the condition in Assumption \ref{ass:general_cdf}.(ii), i.e., $\Omega(p_0)\triangleq p_0(1-F(p_0))$ is concave in $p_0$ over $[0,\bar{v}]$. We thus have the following derivatives regarding $p_0$:
\begin{align}\label{concavity_revenue}
    \frac{\partial\Omega(p_0)}{\partial p_0}=1-F(p_0)-p_0f(p_0) \text{ and }\frac{\partial^2\Omega(p_0)}{\partial (p_0)^2}=-2f(p_0)-p_0f'(p_0)\le0,
\end{align}
where $f'(p_0)\equiv\partial f(p_0)/\partial p_0$. Now we turn back to the dynamic Bayesian game analysis under Assumption \ref{ass:general_cdf}. Notice that we only need to analyze the users' social activity decisions in Stage I and the seller's uniform price in Stage II. For the seller's personalized pricing scheme in Stage II and the users' purchase decisions in Stage III, they are the same as in Section \ref{subsection:game}. More specifically, we will alternate backward analysis with forward analysis in Stages I and II as in Section \ref{Sec:alternate}. Moreover, the users' equilibrium valuation threshold $v^*$ (see Theorem \ref{Theorem:p0-vstar}), which is established through the forward analysis of users' social activity decisions in Section \ref{forward}, remains valid for any continuous distribution of user valuations. Hence, we next focus on the backward analysis of the seller's uniform pricing in Appendix \ref{general_backward}. Finally, we characterize PBE and establish its uniqueness in Appendix \ref{general_PBE_appendix}.

\subsection{Backward Analysis of the Seller's Uniform Pricing in Stage II}\label{general_backward}
After user profiling, the seller's posterior belief for a non-profiled user's valuation $f(v_i|i\in\mathcal{N}_{0})$ in Stage II is given by
\begin{equation}
f(v_i|i\in\mathcal{N}_{0})=\left\{
	\begin{aligned}
		&\frac{(1-\delta)f(v_i)}{1-\delta F(v^*)}, &\quad\textrm{if $0\le v_i\le v^*$,}\\
		&\frac{f(v_i)}{1-\delta F(v^*)},  &\quad\textrm{if $v^*< v_i\le \bar{v}$,}
	\end{aligned}
	\right.
\end{equation}
which could be obtained through Bayes' theorem. Then we calculate the seller's expected non-profiled revenue $\tilde{\Pi}_0$ below: 
\begin{align}
\tilde{\Pi}_0\triangleq|\mathcal{N}_0|\int_{p_0}^{\bar{v}}p_0f(v_i|i\in\mathcal{N}_0)dv_i=|\mathcal{N}_0|p_0\left(1-\frac{1-\delta}{1-\delta F(v^*)}F(p_0)\right).
\end{align}
Notice that we will restrict our attention to the case when $p_0<v^*$ as this always holds in the PBE (see \eqref{vstar_p0_summation} in Theorem \ref{Theorem:p0-vstar}). Hence, the seller's optimal uniform pricing scheme (if existing at all) is given by
\begin{align}\label{opt_p0}
    p_0^*(v^*)=\arg\max_{p_0<v^*}p_0\left(1-\frac{1-\delta}{1-\delta F(v^*)}F(p_0)\right);
\end{align}
for ease of exposition, we hereafter denote $\Xi(p_0)\triangleq p_0\left(1-aF(p_0)\right)$, where $a\triangleq(1-\delta)/(1-\delta F(v^*))\in(0,1]$. Then we have the following derivatives regarding $p_0$:
\begin{align}
    \frac{\partial\Xi(p_0)}{\partial p_0}=1-aF(p_0)-ap_0f(p_0) \text{ and }\frac{\partial^2\Xi(p_0)}{\partial (p_0)^2}=-2af(p_0)-ap_0f'(p_0)\le0,
\end{align}
where the last inequality follows from \eqref{concavity_revenue}. So $\Xi(p_0)$ is concave on $[0,\bar{v}]$ and $\partial \Xi(p_0)/\partial p_0$ decreases with $p_0$ on $[0,\bar{v}]$. To ensure the existence of the optimal solution $p_0^*$ in \eqref{opt_p0}, where $p_0^*<v^*$ should be strictly satisfied in the PBE, the maximum of $\Xi(p_0)$ must be achieved within the semi-open interval of $[0,v^*)$. Since $\partial \Xi(p_0)/\partial p_0|_{p_0=0}=1>0$ and $\partial \Xi(p_0)/\partial p_0$ decreases with $p_0$, we must have $\partial \Xi(p_0)/\partial p_0|_{p_0=v^*}<0$, i.e., 
\begin{align}\label{condition_p0}
1-F(v^*)-(1-\delta)v^*f(v^*)<0,
\end{align}
to ensure the existence of $p_0^*(v^*)$ in \eqref{opt_p0} and the existence of PBE as well. More specifically, $p_0^*(v^*)$ is the unique solution to $\partial \Xi(p_0^*)/\partial p_0^* =0$ with $p_0^*\in(0,v^*)$. Alternatively, we can rewrite this necessary and sufficient condition in \eqref{condition_p0} as $v^*>\check{v}$, where $\check{v}\in(0,\bar{v})$ is the unique valuation such that $1-F(\check{v})-(1-\delta)\check{v}f(\check{v})=0$ (see Lemma \ref{v_check} at the end of this subsection). 

The proposition below concludes the seller's optimal uniform pricing scheme in Stage II.

\begin{proposition}\label{appendix_price}
Given an arbitrary users' valuation threshold $v^*$ in Stage I, if and only if $v^*>\check{v}$, the seller's optimal uniform pricing scheme $p_0^*(v^*)$ in Stage II exists and is the unique solution to
\begin{align}\label{app_0}
    1-aF(p_0^*)-ap_0^*f(p_0^*) =0
\end{align}
with $p_0^*\in(0,v^*)$. In addition, $\partial p_0^*/\partial v^*\le 0$ (see Lemma \ref{p0_v_decrease} in the final of this subsection).
\end{proposition}

\begin{lemma}\label{v_check}
The condition in \eqref{condition_p0} is equivalent to $v^*>\check{v}$, where $\check{v}\in(0,\bar{v})$ is the unique solution to $1-F(\check{v})-(1-\delta)\check{v}f(\check{v})=0$.
\end{lemma}
\begin{proof}
Denote $h(v^*)\triangleq1-F(v^*)-(1-\delta)v^*f(v^*)$. We have its first-order derivative below:
\begin{align}
    \frac{\partial h(v^*)}{\partial v^*}&=-f(v^*)-(1-\delta)\left(f(v^*)+v^*f'(v^*)\right),\nonumber\\
    &=-(1-\delta)\left(2f(v^*)+v^*f'(v^*)\right)-\delta f(v^*)\le 0, 
\end{align}
where the last inequality follows from \eqref{concavity_revenue}. This indicates that $h(v^*)$ decreases with $v^*$ over $[0,\bar{v}]$. Furthermore, we have $h(0)=1>0$ and $h(\bar{v})=-(1-\delta)\bar{v}f(\bar{v})<0$. This then completes the proof.
\end{proof}

\begin{lemma}\label{p0_v_decrease}
For $p_0^*(v^*)$ of \eqref{opt_p0}, $\partial p_0^*/\partial v^*\le 0$.
\end{lemma}

\begin{proof}
Recall that $p_0^*(v^*)$ is the unique solution to \eqref{app_0} with $p_0^*\in(0,v^*)$. Rewrite \eqref{app_0} as 
\begin{align}
H(p_0^*,v^*)\triangleq1-\delta F(v^*)-(1-\delta)\left(F(p_0^*)+p_0^* f(p_0^*)\right)=0,
\end{align}
and we then have the following first-order partial derivative:
\begin{align}
    \frac{\partial H(p_0^*,v^*)}{\partial v^*}=-\delta f(v^*)<0 \text{ and } \frac{\partial H(p_0^*,v^*)}{\partial p_0^*}=-(1-\delta)(2f(p_0^*)+p_0^*f'(p_0^*))\le 0.
\end{align}
Applying the implicit function theorem then completes the proof.
\end{proof}

\subsection{Analysis of PBE}\label{general_PBE_appendix}
To derive the PBE for the whole dynamic Bayesian game, we now combine the users' valuation threshold $v^*$ (see Theorem~\ref{Theorem:p0-vstar}) and the seller's optimal uniform pricing scheme $p_0^*$ (see Proposition \ref{appendix_price}) together. Since the feasible range of $v^*$ is $(\check{v},\bar{v}]$ in the PBE (see Lemma \ref{v_check}), we develop our analysis based on the following two cases concerning $v^*$.
\begin{itemize}
    \item [(1)] If $v^*=\bar{v}$, the seller's uniform pricing scheme is the unique solution to $1-F(p_0^*)-p_0^*f(p_0^*) =0$ (see Proposition~\ref{appendix_price}). Namely, $p_0^*=\hat{p}_0\triangleq \arg\max _{p_0}p_0(1-F(p_0))$, which has been defined in Section \ref{appendix:general}. On the other hand, $v^*=\bar{v}$ if and only if $p_0\ge\bar{v}-1/\delta\cdot\ln\left(n-1+\omega_0\right)$ (see the final part of Appendix \ref{proof_v^*} together with Lemma \ref{last_lemma}). Therefore, if and only if $\bar{v}-\hat{p}_0\le1/\delta\cdot \ln(n-1+\omega_0)$, a unique PBE valuation threshold $v^*=\bar{v}$ exists. This is the Case I in Proposition \ref{Theorem:PBE_general}.
    
    \item [(2)] If $\check{v}<v^*<\bar{v}$, the seller's uniform pricing scheme $p_0^*(v^*)$ is the unique solution to \eqref{app_0}. Combining it with \eqref{vstar_p0_summation} yields
    \begin{equation}\label{general_a}
		p_0^*(v^*)=v^*-\frac{1}{\delta}\tilde{J}(v^*).
	\end{equation}
    We now denote $K(v^*)\triangleq p_0^*(v^*)-v^*+\tilde{J}(v^*)/\delta$, and its first-order derivative over $(\check{v}, \bar{v})$ is given by
    \begin{align}
        \frac{\partial K(v^*)}{\partial v^*}=\frac{\partial p_0^*(v^*)}{\partial v^*}-\left(1-\frac{1}{\delta}\tilde{J}'(v^*)\right)<0,
    \end{align}
    where the inequality holds due to Lemma \ref{p0_v_decrease} and Lemma \ref{jprime} (together with Lemma \ref{last_lemma}). Hence, $K(v^*)$ is decreasing over $(\check{v}, \bar{v})$. In particular, when $v^*=\check{v}$, we have $\partial \Xi(p_0)/\partial p_0|_{p_0=v^*}=0$ and then $\partial \Xi(p_0)/\partial p_0>0$ always holds across $[0,v^*]$, which indicates that $\Xi(p_0)$ increases with $p_0$ over $[0,v^*]$; so we will have $p_0^*(v^*)\to v^*$ as $v^*$ approaches $\check{v}$. Hence, $\lim_{v^*\to \check{v}}K(v^*)=\tilde{J}(\check{v})/\delta>0$ and we also have 
    \begin{align}
        K(\bar{v})=\hat{p}_0-\bar{v}+\frac{1}{\delta}\ln\left(n-1+\omega_0\right).
    \end{align}
    If and only if $K(\bar{v})<0$, there exists a unique solution to $K(v^*)=0$ over $(\check{v},\bar{v})$. Equivalently, if and only  if $\bar{v}-\hat{p}_0>1/\delta\cdot \ln(n-1+\omega_0)$, a unique PBE valuation threshold $v^*$ exists on $(\check{v},\bar{v})$. This is Case II in Proposition \ref{Theorem:PBE_general}, where $v^*<\bar{v}$ is the unique solution to \eqref{general_a}. 
\end{itemize}
This completes the proof for Proposition \ref{Theorem:PBE_general}.

\begin{landscape}
\begin{lemma}\label{last_lemma}
Under Assumption \ref{ass:general_cdf}, Proposition \ref{unique_user} and Lemma \ref{jprime} continue to hold.
\end{lemma}
\begin{proof}
Recall that in Appendix \ref{proof_v^*} for Proposition \ref{unique_user}, the condition of a uniform valuation distribution is applied in Lemma \ref{claim:J_concave}, more specifically, in Line 1 of \eqref{SO_J} where we calculated $\partial^2\tilde{J}(v^\dagger)/\partial (v^\dagger)^2$. We now recalculate the second-order derivative of $\tilde{J}(v^\dagger)$ based on \eqref{FO_J} below:
\begin{align}
\frac{\partial^2 \tilde{J}(v^\dagger)}{\partial (v^\dagger)^2}	
&=\sum_{m=1}^{n-1}\frac{(n-1)!}{(n-1-m)!(m-1)!}\left(\ln(m+\omega_0)-\ln(m-1+\omega_0)\right)(m-1)F(v^\dagger)^{m-2}\left(1-F(v^\dagger)\right)^{n-1-m}\left(f(v^\dagger)\right)^2\tag{Line 1-1}\\&\quad-\sum_{m=1}^{n-1}\frac{(n-1)!}{(n-1-m)!(m-1)!}\left(\ln(m+\omega_0)-\ln(m-1+\omega_0)\right)(n-1-m)F(v^\dagger)^{m-1}\left(1-F(v^\dagger)\right)^{n-2-m}\left(f(v^\dagger)\right)^2\tag{Line 1-2}\\
&\quad+\sum_{m=1}^{n-1}\frac{(n-1)!}{(n-1-m)!(m-1)!}\left(\ln(m+\omega_0)-\ln(m-1+\omega_0)\right)F(v^\dagger)^{m-1}\left(1-F(v^\dagger)\right)^{n-1-m}f'(v^\dagger),\tag{Line 1-3}\\
&=\eqref{SO_J}+\sum_{m=1}^{n-1}\frac{(n-1)!}{(n-1-m)!(m-1)!}\left(\ln(m+\omega_0)-\ln(m-1+\omega_0)\right)F(v^\dagger)^{m-1}\left(1-F(v^\dagger)\right)^{n-1-m}\frac{\partial f(v^\dagger)}{\partial v^\dagger},
\end{align}
Since \eqref{SO_J} $<0$ and $\partial f(v^\dagger)/\partial v^\dagger<0$ due to the concavity of $F(v^\dagger)$ in Assumption \ref{ass:general_cdf}.(i). Therefore, $\partial ^2\tilde{J}(v^\dagger)/\partial (v^\dagger)^2<0$, i.e., the concavity of $\tilde{J}(v^\dagger)$, continues to hold under Assumption \ref{ass:general_cdf}. Note that the remaining analysis in Appendix \ref{proof_v^*} only relies on the concavity of $\tilde{J}(v^\dagger)$ but not the uniform distribution assumption. The remaining analysis in Appendix \ref{proof_v^*} and Proposition \ref{unique_user} continue to hold under Assumption \ref{ass:general_cdf}. Also, the proof for Lemma \ref{jprime} only relies on the concavity of $\tilde{J}(v^\dagger)$, and thus Lemma \ref{jprime} continue to hold under Assumption \ref{ass:general_cdf}.
\end{proof}
\end{landscape}

\section{Performance Evaluation}\label{Appendix:Performance}

In this appendix, we evaluate the performance of our proposed pricing mechanism for the seller, through which we will understand the power of harnessing social network benefits (see Section \ref{subsec:effect_seller}). To evaluate the performance, we introduce two additional benchmarks adapted from the existing literature (e.g., \cite{conitzer2012hide,2017Is,zhang2011perils,2020Voluntary}). Specifically, we outline the two benchmarks for the pricing mechanism along with the one we proposed below.

\begin{itemize}
    \item \textit{Profile-independent Pricing (PIP)} (e.g., \cite{conitzer2012hide,2017Is,zhang2011perils}): The first benchmark involves the seller's optimal pricing without differentiation based on user profiles. Variations of this benchmark are commonly adopted in the literature on behavior-based price discrimination. In this case, the seller charges a uniform price of $\bar{v}/2$ to all users.
        
    \item \textit{Traditional-profiling-based Pricing (TPP)} (e.g., \cite{2017Is,2020Voluntary}): Another benchmark is wherein the seller personalizes price offerings through traditional user profiling. In this case, users voluntarily disclose private information without considering any social network benefits. Detailed equilibrium analyses are provided in Proposition \ref{Benchmark:no-social-network}, also named no-social-network benchmark there.
        
    \item \textit{Our Social-profiling-based Pricing (SPP)}: In our proposed pricing mechanism, the seller personalizes the price offerings through social profiling, which harnesses users' social network benefits as well as the positive network externality among users' information revelation.
\end{itemize}

\begin{figure}[h]
	\centering
	\includegraphics[width=0.4\linewidth]{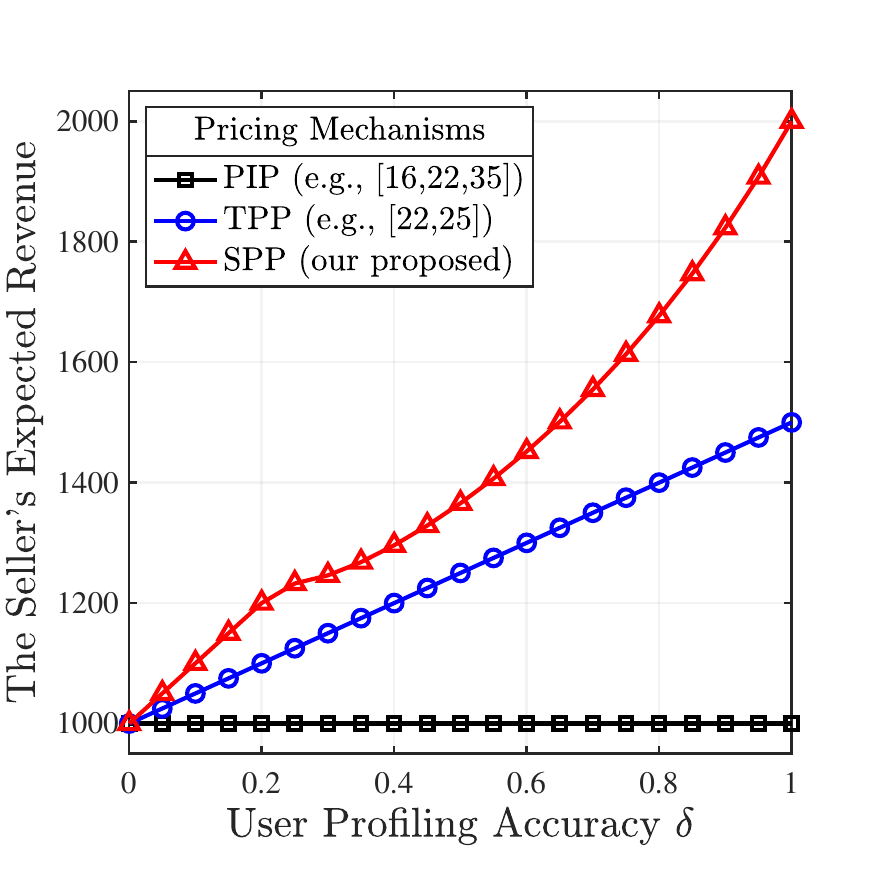}
	\caption{The seller's expected revenue under the three pricing mechanisms~(PIP, TPP, SPP) versus user profiling accuracy $\delta$. Here, we consider $n=100$ users with uniformly distributed valuations and fix the mean valuation $\bar{v}/2$ as $20$.}
    \label{fig_perf_revenue}
\end{figure}

Fig. \ref{fig_perf_revenue} numerically compares the seller's expected revenue in the three pricing mechanisms. In particular, our SPP mechanism always outperforms the PIP and TPP benchmarks, with maximum improvements of $100\%$ and $33.3\%$ achieved at $\delta=1$, respectively. The superiority of our SPP over TPP is attributed to the social network benefits that the seller can harness to harvest more users' disclosed information. 

\section{Extension to Partially Connected Social Networks: Equilibrium Results}\label{Appendix:Partially}

In this appendix, we will present key structural results for the PBE of the new extension to partially connected social networks in Section \ref{subsec:extension:partial}. Similar to the analysis in Section \ref{subsec:extension:benefits}, we can show that \textit{the valuation threshold for users' social activity decisions in Stage I varies among users depending on their different network positions}, which also encode their respective levels of connectivity. We thus denote these valuation thresholds as $\{v_i^*,\forall i\in\mathcal{N}\}$. Furthermore, these valuation thresholds satisfy the following system of equations: 
\begin{equation*}
    v^*_i=p_0^i+\frac{1}{\delta}\mathbb{E}_{v_j\sim F(\cdot)}\ln\left(\sum_{j\neq i }g_{ij}\mathbbm{1}(v_j\le v_j^*)+\omega_0\right),\ \forall i\in\mathcal{N},
\end{equation*}
which extends from the original formulation in \eqref{vstar_p0_summation} for a fully connected social network with a common valuation threshold~$v^*$. Moreover, for the seller's uniform pricing $\{p_0^i,\forall i \in\mathcal{N}\}$ in Stage II, we can easily extend from the original optimal scheme \eqref{base_price_equation3} in Proposition \ref{base_price} and obtain:
\begin{equation*}
    p_0^i=\frac{\bar{v}-\delta v^*_i}{2(1-\delta)},\ \forall i \in\mathcal{N},
\end{equation*}
and further note that we have $p_0^i<v^*_i$ in the PBE. With these structural characterizations and properties of PBE well established, we now proceed with empirical analysis in Section \ref{subsec:extension:partial}, specifically exploring how network connectivity affects users' social activity decisions.

\section{Empirical Studies on Real-world Consumer Profiles}\label{Appendix:Real_Consumers}

In this appendix, we apply our proposed pricing mechanism to real-world consumer profiles, utilizing an annually-updated consumer behavior and shopping habits dataset from \cite{consumer}. Recall that most of our previous analyses have assumed a uniform distribution for user valuations. However, as we will demonstrate later, this empirical experiment using more general distributions further confirms the robustness of our previous pricing methods and major insights.

In what follows, we first use the empirical data to estimate users' valuation distribution, which serves as valuable market information. Subsequently, we derive the equilibrium behaviors of both (sampled) users and the seller based on this estimation. Finally, we examine the impact of profiling accuracy on the seller's empirical revenue gain.

\begin{figure}[h]
    \centering
    \subfigure[User valuation distribution]{
        \begin{minipage}{0.45\linewidth}
            \centering
            \includegraphics[width=0.8\linewidth]{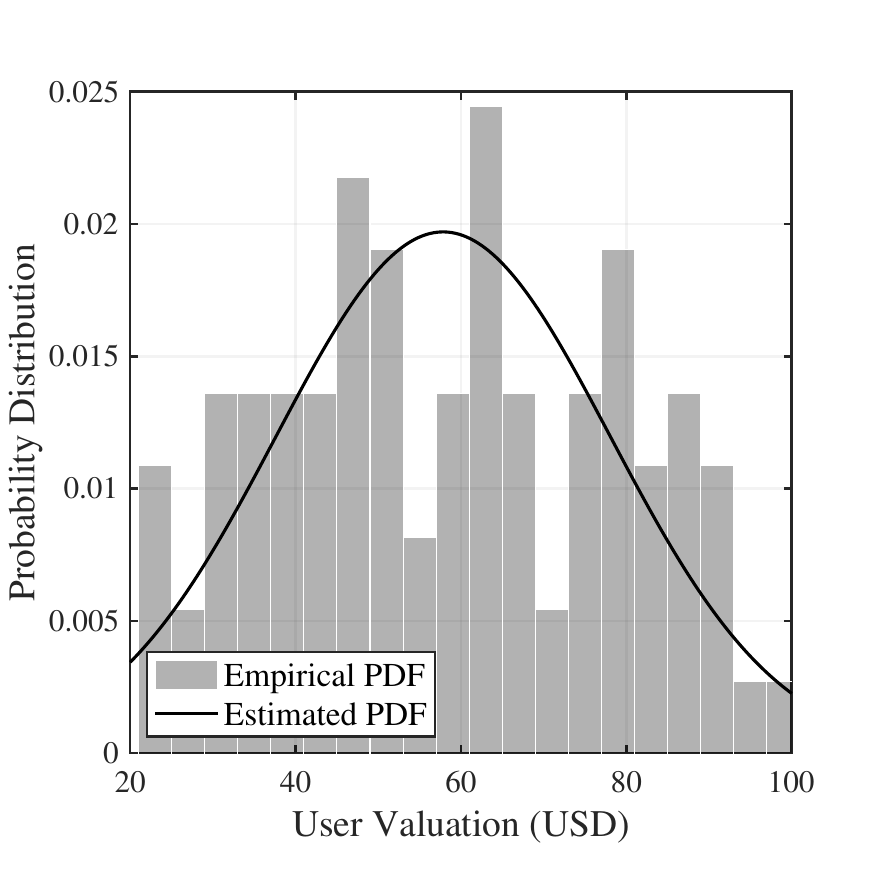}
            \label{fig_empirical_data}
		\end{minipage}
	}
	\subfigure[Fraction of socially active users]{
        \begin{minipage}{0.45\linewidth}
            \centering
            \includegraphics[width=0.8\linewidth]{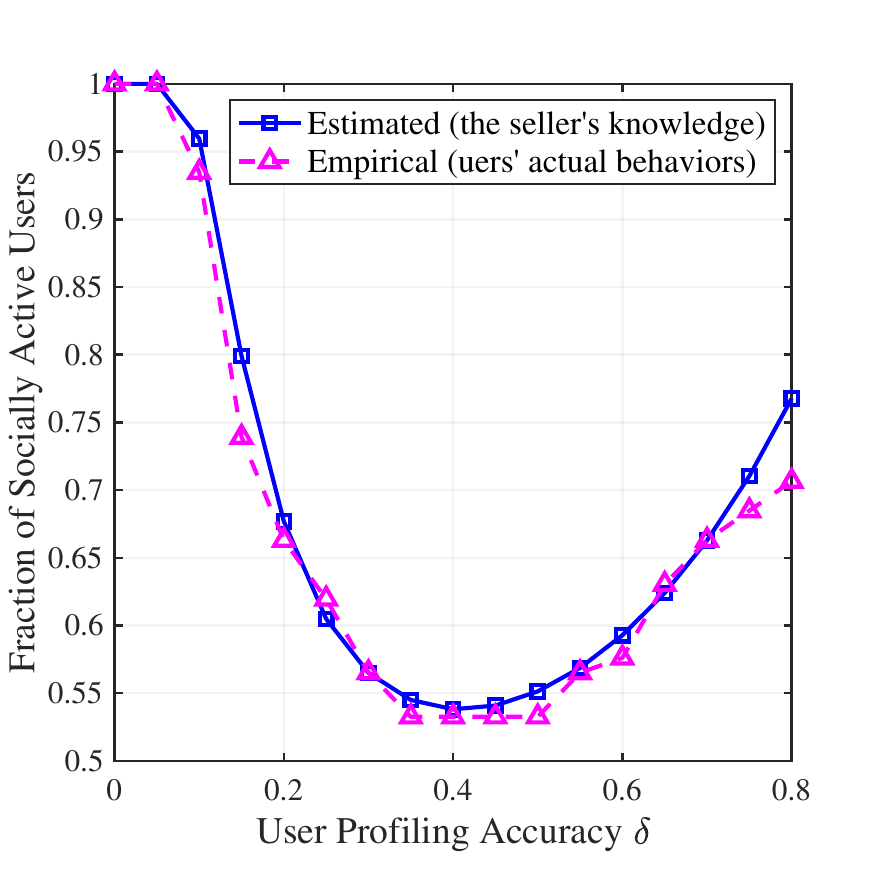}
            \label{fig_empirical_vstar}
		\end{minipage}
	}
    \subfigure[The seller's uniform price]{
        \begin{minipage}{0.45\linewidth}
            \centering
            \includegraphics[width=0.8\linewidth]{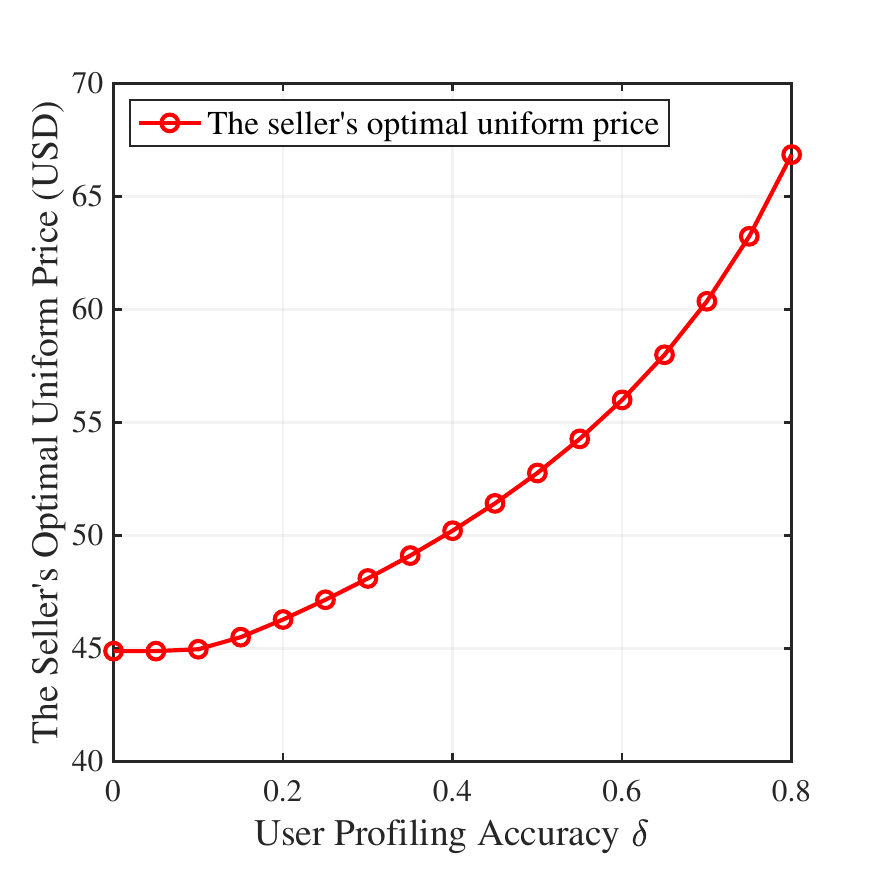}
            \label{fig_empirical_p0}
		\end{minipage}
	}
	\subfigure[The seller's revenue]{
        \begin{minipage}{0.45\linewidth}
            \centering
            \includegraphics[width=0.8\linewidth]{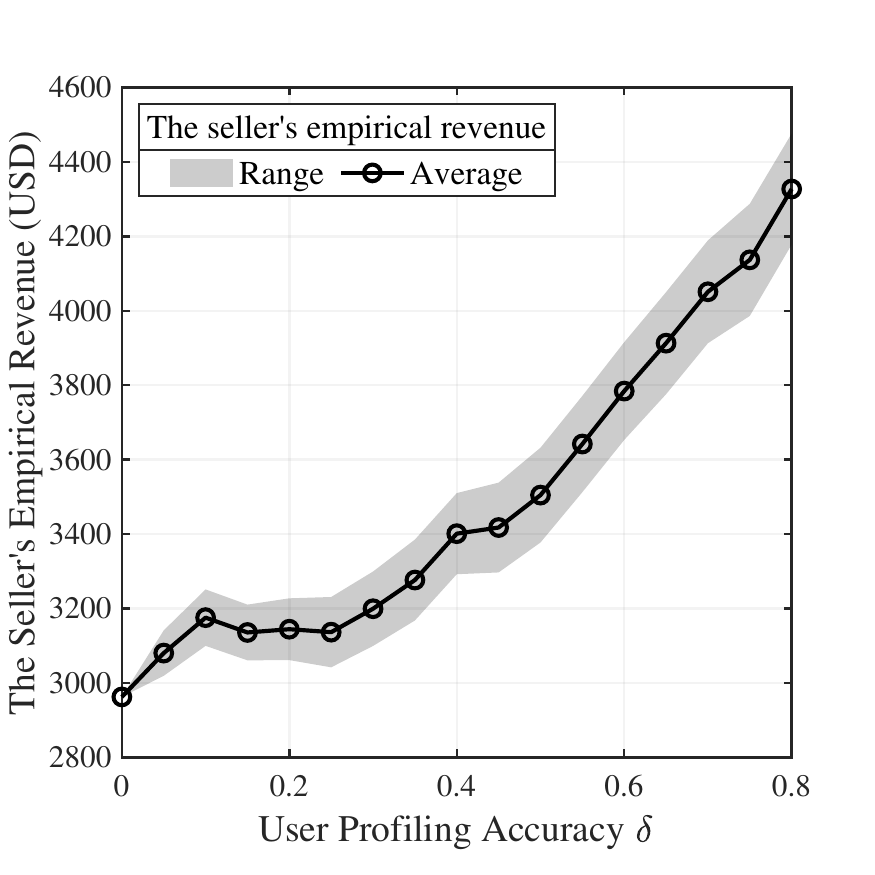}
            \label{fig_empirical_revenue}
		\end{minipage}
	}
\caption{The user valuation distribution, the fraction of socially active users and the seller's optimal uniform price $p_0^*$ in the PBE (whenever it exists), and the seller's empirical revenue based on a real-world dataset from \cite{consumer}.}
\end{figure}

\subsection{Empirical Statistics}\label{Appendix:empirical:statistics}

The dataset from \cite{consumer} offers consumer information such as demographic details, purchased items, purchase amounts, and previous purchase history. Our analysis focuses specifically on handbag purchases without discounts applied and the corresponding $92$ consumer profiles.\footnote{We filter the dataset using the criteria "Item Purchased = Handbag" AND "Discount Applied = No." In particular, we focus on purchases made without any promotional discounts applied. This allows us to exclude the potential psychological effects of price discounts on consumer perceptions and purchase intentions \cite{buyukdaug2020effect,simonson2004anchoring}.} Moreover, we associate each consumer's purchase amount with his product valuation, representing the monetary value of his transaction on the purchased item denoted in United~States Dollars (USD). To incorporate the social network benefit in our modeling framework, we further hypothetically consider these $92$ consumers as social network users within a community.

Fig. \ref{fig_empirical_data} depicts the empirical statistics of these $92$ users' valuation distribution. We estimate the valuation distribution as a normal distribution $N(\mu,\sigma^2)$ truncated over $[20,100]$, and our data fitting results conclude that the mean is $\mu=57.84$ (USD) and the standard deviation is $\sigma = 20.25$ (USD).\footnote{By assuming a normal distribution, the mean $\mu$ and the standard deviation $\sigma$ are the two parameters that we need to estimate. To proceed, we first estimate these two parameters through the maximum likelihood estimation, and we further statistically verify the goodness-of-fit using the Kolmogorov–Smirnov test.} This estimated valuation distribution, represented by the solid curve in Fig. \ref{fig_empirical_data}, serves as the prior market information regarding user valuation and plays a crucial role in guiding the decision-making process for both the users and the seller (in our dynamic Bayesian setting).

\subsection{Robustness of Our PBE Results}
Next, we conduct numerical analyses to examine the PBE outcomes based on the estimated valuation distribution in Fig. \ref{fig_empirical_data}. The PBE analysis under this truncated normal distribution follows a similar rationale as presented at the beginning~of~Section \ref{appendix:general}. In particular, the seller's estimation of the fraction of socially active users in equilibrium~is~depicted by the blue solid curve in Fig. \ref{fig_empirical_vstar}, and this estimation serves as the basis for the~seller~to determine the optimal uniform pricing in Fig. \ref{fig_empirical_p0}. Furthermore, we investigate the actual~social behaviors of the $92$ users based on their individual valuations. The empirical findings, represented by the purple dashed curve in Fig. \ref{fig_empirical_vstar}, have demonstrated a high degree of consistency with~the seller's estimation of the blue solid curve. This consistency further validates the effectiveness of our approximation of the empirical valuation distribution with the estimated one in Fig. \ref{fig_empirical_data}.

Overall, the empirical results in Fig. \ref{fig_empirical_vstar} and Fig. \ref{fig_empirical_p0} closely align with the findings in Fig. \ref{FIG:PBE_delta}, which assumed a uniform valuation distribution over $[0,\bar{v}]$. Specifically,~the fraction of socially active users in Fig. \ref{fig_empirical_vstar} exhibits a decreasing-then-increasing trend as the~profiling accuracy $\delta$ increases. Moreover, Fig. \ref{fig_empirical_p0} illustrates that the seller progressively raises the uniform price $p^*_0$ with higher values of $\delta$, thereby motivating users to engage actively online. These empirical findings are consistent with Proposition~\ref{Prop:PBE_delta} derived under the assumption of a uniform valuation distribution. As a final remark, the key insights derived from the analysis based on a uniform distribution remain robust when considering more general distribution functions.

\subsection{Revenue Evaluation}

Finally, we numerically analyze the seller's actual revenue obtained from the $92$ users. Specifically, we generate random profiling results for the seller based on the user profiling accuracy $\delta$ and perform extensive simulation experiments with $10,000$ repetitions to obtain the seller's revenue results. In Fig. \ref{fig_empirical_revenue}, the black dots represent the seller's average~empirical revenue under different profiling accuracy levels $\delta$, while the transparent gray patch indicates the range of her revenue distribution (spanning from the maximum to the minimum value).

As shown in Fig. \ref{fig_empirical_revenue}, the seller's revenue exhibits a predominantly upward trend with~higher profiling accuracy $\delta$. This trend is consistent with the finding in Fig. \ref{fig_perf_revenue}, which assumed a uniform valuation distribution. Moreover, due to the~irregular distribution of the empirical valuation data, we observe a slight drop in the seller's empirical revenue gain within the range of small $\delta$ values between $0.1$ and $0.25$ in Fig. \ref{fig_empirical_revenue}.

\end{document}